%% file: main.tex
\documentclass{LMCS}

\usepackage[latin1]{inputenc}
\usepackage[english]{babel}
\usepackage[T1]{fontenc}
\usepackage{microtype}
\usepackage{amsmath,amssymb,amsthm}
\usepackage{hyperref,enumerate}
\usepackage{graphics}
\usepackage{ifthen}

\newcommand{\includeTikZ}[1]{\includegraphics{#1}}

\theoremstyle{definition}
\newtheorem{notation}[thm]{Notation}

\newcounter{claimcnt}[thm]
\newtheorem{sclaim}[claimcnt]{Claim}
\newenvironment{proofofclaim}
  {\begin{trivlist}\item\textit{Proof.}}
  {\end{trivlist}}

\newcommand{\varqedsymbol}{\ensuremath{\lrcorner}}
\newcommand{\varqed}{\hfill\varqedsymbol}
\newcommand{\varqedeq}{\tag*{\varqedsymbol}}

\newcommand{\proofpart}[1]{\addvspace{1.15\topsep}
  \noindent\textit{{#1}:}}
\newcommand{\proofpartnl}[1]{\addvspace{1.15\topsep}
  \noindent\textit{#1}\\[\partopsep]}
\newcommand{\onlyifdir}{\proofpart{``Only if'' direction}}
\newcommand{\ifdir}{\proofpart{``If'' direction}}
\newcommand{\ad}[1]{\proofpart{Ad~{#1}}}
\newcommand{\fromto}[2]{\proofpart{{#1}~$\Longrightarrow$~{#2}}}
\newcommand{\step}[2]{\proofpartnl{Step~{#1}:~{#2}}}

\input{macros}

\def\doi{7 (3:09) 2011}
\lmcsheading%
{\doi}
{1--55}
{}
{}
{Aug.~28, 2010}
{Sep.~\phantom01, 2011}
{}   

\begin{document}

\title[Answering Non-Monotonic Queries in Relational Data Exchange]
  {Answering Non-Monotonic Queries\\in Relational Data Exchange\rsuper*}
\author[A.~Hernich]{Andr{\'e} Hernich}
\address{Institut für Informatik, Humboldt-Universtität, Berlin, Germany}
\email{hernich@informatik.hu-berlin.de}
\titlecomment{{\lsuper*}A preliminary version of this article, \cite{Hernich:ICDT10},
  appeared in the proceeding of the 13th international conference
  on database theory (ICDT 2010).}

\begin{abstract}
  Relational data exchange is the problem of translating
  relational data from a source schema into a target schema,
  according to a specification of the relationship
  between the source data and the target data.
  One of the basic issues is how to answer queries
  that are posed against target data.
  While consensus has been reached on the definitive semantics
  for \emph{monotonic queries},
  this issue turned out to be considerably more difficult
  for \emph{non-monotonic queries}.
  Several semantics for non-monotonic queries have been proposed
  in the past few years.

  This article proposes a new semantics for non-monotonic queries,
  called the \emph{\XGCWA-semantics}.
  It is inspired by semantics from the area of deductive databases.
  We show that the \XGCWA-semantics
  coincides with the standard open world semantics on monotonic queries,
  and we further explore the (data) complexity of evaluating non-monotonic
  queries under the \XGCWA-semantics.
  In particular, we introduce a class of schema mappings
  for which universal queries can be evaluated under the \XGCWA-semantics
  in polynomial time (data complexity)
  on the core of the universal solutions.
\end{abstract}

\maketitle

\section{Introduction}
\label{sec:intro}

Data exchange is the problem of translating databases
from a source schema into a target schema,
whereby providing access to the source database
through a materialized database over the target schema.
It is a special case of data integration \cite{Lenzerini:PODS02}
and arises in tasks like data restructuring,
updating data warehouses using ETL processes,
or in exchanging data between different, possibly independently created,
applications (see, e.g., \cite{Clio05,FKMP:TCS336-1}).
Tools for dealing with data exchange are available for quite a while
\cite{SHTGL77,Clio05,MPR+09}.
Fundamental concepts and algorithmic issues in data exchange
have been studied recently by Fagin, Kolaitis, Miller, and Popa
in their seminal paper \cite{FKMP:TCS336-1}.
For a comprehensive overview on data exchange,
the reader is referred to \cite{FKMP:TCS336-1}
or any of the surveys \cite{Kolaitis:PODS05,Barcelo:SR38-1,HS:LOFT08,ABL+10}.

This article deals with \emph{relational} data exchange,
which received a lot of attention in the data exchange community
(see, e.g., the survey articles cited above).
In this setting,
the mapping from source data to target data is described
by a \emph{schema mapping} $M = (\sigma,\tau,\Sigma)$
which consists of relational database schemas $\sigma$ and $\tau$
(finite sets of relation names with associated arities),
called \emph{source schema} and \emph{target schema}, respectively,
and a finite set $\Sigma$ of constraints
(typically, sentences in some fragment of first-order logic)
which can refer to the relation names in $\sigma$ and $\tau$.
Typical constraints are tuple generating dependencies (tgds),
which come in two flavors -- st-tgds and t-tgds --,
and equality generating dependencies (egds).
For example, st-tgds are first-order sentences of the form
\(
  \forall \tup{x}, \tup{y}\,
  \bigl(
    \phi(\tup{x},\tup{y})
    \limplies
    \exists \tup{z}\,
    \psi(\tup{x},\tup{z})
  \bigr),
\)
where $\phi$ is a conjunction of relation atoms over $\sigma$,
and $\psi$ is a conjunction of relation atoms over $\tau$.
Their precise definitions are deferred
to Section~\ref{sec:basics/data-exchange}.
Given a relational database instance $S$ over $\sigma$
(called \emph{source instance} for $M$),
a \emph{solution} for $S$ under $M$
is a relational database instance $T$ over $\tau$
such that the instance $S \union T$ over $\sigma \union \tau$
that consists of the relations in $S$ and $T$
satisfies all the constraints in $\Sigma$.

An important task in relational data exchange
is to answer queries that are posed against the target schema
of a schema mapping.
The answer to a query should be semantically consistent
with the source data and the schema mapping,
that is, it should reflect the information in the source instance
and the schema mapping as good as possible.
Since a source instance usually has more than one solution,
a fundamental question is:
What is the semantics of a query,
that is, which tuples constitute the set of answers to a query
over the target schema of a schema mapping and a given source instance?
Furthermore, in data exchange the goal is to answer queries
using a materialized solution,
without access to the source instance.%
\footnote{A common assumption is that the source instance is not available
  after the data exchange has been performed \cite{FKMP:TCS336-1}.}
This brings us to a second fundamental question:
Given a source instance,
which solution should we compute in order to be able to answer queries?

Concerning the first question,
the \emph{certain answers semantics}, introduced in \cite{FKMP:TCS336-1},
has proved to be adequate for answering a wide range of queries such as
unions of conjunctive queries
(a.k.a.\ existential positive first-order queries).
Under the certain answers semantics,
a query $q$ is answered by the set of all tuples that are answers to $q$
no matter which solution $q$ is evaluated on.
More precisely, the certain answers consist of all those tuples $\tup{a}$
such that $q(\tup{a})$ is true in all solutions.
Concerning the second question,
the \emph{universal solutions} proposed in \cite{FKMP:TCS336-1}
have proved to be very useful.
Universal solutions can be regarded as most general solutions
in the sense that they contain sound and complete information.
In a number of settings, they can be computed efficiently
\cite{FKMP:TCS336-1,FKP:TODS30-1,GN:JACM08,HS:PODS07,DNR:PODS08,%
  Marnette:PODS09,GM:ICDT10}.
It was shown that the certain answers to unions of conjunctive queries
can be computed by evaluating such a query on an arbitrary universal solution,
followed by a simple post-processing step \cite{FKMP:TCS336-1}.
Similar results hold for other monotonic queries,
like unions of conjunctive queries with inequalities
\cite{FKMP:TCS336-1,DNR:PODS08,ABR:ICDT09}.

For many non-monotonic queries,
the certain answers semantics yields results that intuitively
do not seem to be accurate \cite{FKMP:TCS336-1,ABFL:PODS04,Libkin:PODS06}.%
\footnote{It was also pointed out in \cite{ABFL:PODS04,Libkin:PODS06}
  that similar problems arise for the universal solution-based semantics
  from \cite{FKP:TODS30-1}.}
The following example illustrates the basic problem:

\begin{exa}
  \label{exa:copying}
  Consider a schema mapping $M = (\set{R},\set{R'},\Sigma)$,
  where $R,R'$ are binary relation symbols
  and $\Sigma$ contains the single st-tgd
  \[
    \theta\, \isdef\, \forall x,y\, \bigl(R(x,y) \limplies R'(x,y)\bigr).
  \]
  Let $S$ be a source instance for $M$
  where $R$ is interpreted by $R^S \isdef \set{(a,b)}$.
  Since schema mappings describe translations from source to target,
  it seems natural to assume that $M$ and $S$ together
  give a complete description of the solutions for $S$ under $M$,
  namely that such a solution contains the tuple $(a,b)$ in $R'$
  (as implied by $\theta$ and the tuple $(a,b)$ in $R^S$),
  but no other tuple
  (since this is not implied by $M$ and $S$).
  In particular, it seems natural to assume that the instance $T$
  which interprets $R'$ by the relation $\set{(a,b)}$
  is the only solution for $S$ under $M$,
  and that the answer to the query
  \[
    q(x,y)\, \isdef\,
    R'(x,y) \land \forall z\, \bigl(R'(x,z) \limplies z = y\bigr)
  \]
  with respect to $M$ and $S$ is $\set{(a,b)}$.
  However, the certain answers to $q$ with respect to $M$ and $S$ are empty.
\end{exa}

The assumption that a schema mapping $M$ and a source instance $S$ for $M$
give a complete description of the solutions for $S$ under $M$
corresponds to the \emph{closed world assumption (CWA)} \cite{Reiter:LD78},
as opposed to the \emph{open world assumption (OWA)}
underlying the certain answers semantics.
To remedy the problems mentioned above,
Libkin \cite{Libkin:PODS06} proposed semantics based on the CWA,
which were later extended to a more general setting \cite{HS:PODS07}
(a combined version of \cite{Libkin:PODS06} and \cite{HS:PODS07}
 appeared in \cite{HLS11}).
While the CWA-semantics work well in a number of situations
(e.g., if a unique inclusion-minimal solution exists),
they still lead to counter-intuitive answers in certain other situations
\cite{LS:PODS08,AK:PODS08}.
To this end,
Libkin and Sirangelo \cite{LS:PODS08} proposed a combination of the CWA
and the OWA,
whereas Afrati and Kolaitis \cite{AK:PODS08} studied a restricted version
of the CWA-semantics,
and showed it to be useful for answering aggregate queries.
Henceforth, we use the term \emph{non-monotonic semantics}
to refer to the semantics from \cite{HLS11,LS:PODS08,AK:PODS08}.
In contrast,
we call the certain answers semantics \emph{OWA-semantics}.

A drawback of the non-monotonic semantics is that most of them
are not invariant under logically equivalent schema mappings.
That is, they do not necessarily lead to the same answers
with respect to schema mappings
specified by logically equivalent sets of constraints
(see Section~\ref{sec:problems}).
Since logically equivalent schema mappings intuitively
specify the same translation of source data to the target schema,
it seems natural, though, that the answer to a query
is the same on logically equivalent schema mappings.
Furthermore, it can be observed that the non-monotonic semantics
do not necessarily reflect the standard semantics
of first-order quantifiers (see Section~\ref{sec:problems}).
For example, consider a schema mapping $M = (\set{P},\set{Q},\set{\theta})$
with
\(
  \theta = \forall x\, (P(x) \limplies \exists y\, Q(x,y)),
\)
and let $S$ be a source instance for $M$ with a single element $a$ in $P$.
Under almost all of the non-monotonic semantics,
the answer to the query
$q = \text{``Is there exactly one $y$ with $Q(a,y)$?''}$
is \emph{true}.
However, existential quantification $\exists y\, Q(x,y)$
is typically interpreted as:
there is one $y$ with $Q(x,y)$, or there are two $y$ with $Q(x,y)$,
or there are three $y$ with $Q(x,y)$, and so on.
To be consistent with this interpretation,
the answer to $q$ should be \emph{false},
as otherwise the possibility of having two or more $y$ with $Q(x,y)$
is excluded.
Another reason for why it is natural to answer $q$ by \emph{false} is that
$\theta$ can be expressed equivalently as
\(
  \theta' = \forall x (P(x) \limplies \biglor_c\, Q(x,c)),
\)
where $c$ ranges over all possible values.
Since $M$ and $M' = (\set{P},\set{Q},\set{\theta'})$ are logically equivalent,
the answer to $q$ should either be \emph{true} or \emph{false}
with respect to both $M$ and $M'$.
Letting the answer be \emph{true} would not reflect the intended
meaning of the disjunction in $\theta'$,
unless we wish to interpret disjunctions exclusively.

This article introduces a new semantics
for answering non-monotonic queries,
called \emph{\XGCWA-semantics},
that is invariant under logically equivalent schema mappings,
and intuitively reflects the standard semantics of first-order quantifiers.
The starting point for the development of the \XGCWA-semantics
is the observation that query answering with respect to schema mappings
is very similar to query answering on \emph{deductive databases} \cite{GMN:CS84}
(see Section~\ref{sec:deductive}),
and that non-monotonic query answering on deductive databases
is a well-studied topic
(see, e.g.,
 \cite{Reiter:LD78,Minker:CADE82,YH:JAR1-2,Chan:TLDE5-2,GMN:CS84,DFN:AR2}).
Many of the query answering semantics proposed in this area
can be applied with minor modifications
to answer queries in relational data exchange.
Therefore, it seems obvious to study these semantics
in the context of data exchange.
This is done in Section~\ref{sec:deductive}.
More precisely,
we consider the semantics based on Reiter's \emph{CWA} \cite{Reiter:LD78},
the \emph{generalized CWA (GCWA)} \cite{Minker:CADE82},
the \emph{extended GCWA (EGCWA)} \cite{YH:JAR1-2},
and the \emph{possible worlds semantics (PWS)} \cite{Chan:TLDE5-2}.
It turns out that the semantics based on Reiter's CWA and the EGCWA
are too strong,
the GCWA-based semantics is too weak,
and the PWS is not invariant under logically equivalent schema mappings.
On the other hand,
the GCWA-based semantics seems to be a good starting point
for developing the \XGCWA-semantics.

In contrast to the other non-monotonic semantics,
the \XGCWA-semantics is defined with respect to all possible schema mappings.
It is based on the new concept of \emph{\XGCWA-solutions},
in the sense that, under the \XGCWA-semantics,
the set of answers to a query $q(\tup{x})$
with respect to a schema mapping $M$ and a source instance $S$
consists of all tuples $\tup{a}$
such that $q(\tup{a})$ holds in all \XGCWA-solutions for $S$ under $M$.
\XGCWA-solutions have a very simple definition
in many of the settings considered in the data exchange literature
(e.g., with respect to schema mappings specified by st-tgds and egds):
in these settings they are basically unions of inclusion-minimal solutions.

The major part of this article deals with the \emph{data complexity}
of evaluating queries under the \XGCWA-semantics.
Data complexity here means that the schema mapping and the query
are fixed (i.e., they are not part of the input).
We show that the \XGCWA-semantics and the OWA-semantics
coincide for monotonic queries
(Proposition~\ref{prop:certXGCWA-for-monotonic-queries}),
so that all results on evaluating monotonic queries under the OWA-semantics
carry over to the \XGCWA-semantics.
On the other hand,
there are simple schema mappings
(e.g., schema mappings specified by LAV tgds),
and simple non-monotonic Boolean first-order queries
for which query evaluation under the \XGCWA-semantics is \co\NP-hard
or even undecidable
(Propositions~\ref{prop:CQ-with-neg-coNP-complete}
 and \ref{prop:EX*ALL-undecidable}).

The main result (Theorem~\ref{thm:universal-queries})
shows that \emph{universal queries}
(first-order queries of the form $\forall \tup{x}\, \phi$
 with $\phi$ quantifier-free)
can be evaluated in polynomial time under the \XGCWA-semantics,
provided the schema mapping is specified by \emph{packed} st-tgds,
which we introduce in this article.
Packed st-tgds are st-tgds of the form
\(
  \forall \tup{x} \forall \tup{y} (
    \phi(\tup{x},\tup{y}) \limplies \exists \tup{z}\, \psi(\tup{x},\tup{z})
  ),
\)
where every two distinct atomic formulas in $\psi$
share a variable from $\tup{z}$.
This is a rather strong restriction,
but still allows for non-trivial use of existential quantifiers
in st-tgds.
Surprisingly,
the undecidability result mentioned above involves a schema mapping
defined by packed st-tgds,
and a first-order query starting with a block of existential quantifiers
and containing just one universal quantifier.
The main result does not only state
that universal queries can be evaluated in polynomial time
under the \XGCWA-semantics and schema mappings defined by packed st-tgds,
but it also shows that the answers can be computed
from the \emph{core of the universal solutions}
(\emph{core solution}, for short),
without access to the source instance.
The core solution is the smallest universal solution
and has been extensively studied in the literature
(see, e.g., \cite{FKP:TODS30-1,GN:JACM08,HLS11,DNR:PODS08,%
  Marnette:PODS09,GM:ICDT10}).
Furthermore, since the core solution can be used to evaluate
unions of conjunctive queries under the OWA-semantics,
we need only one solution, the core solution,
to answer both types of queries, unions of conjunctive queries
and universal queries.

The article is organized as follows.
In Section~\ref{sec:basics}, we fix basic definitions and mention basic results
that are used throughout this article.
Section~\ref{sec:problems} shows that the previously proposed
non-monotonic semantics are not necessarily invariant
under logically equivalent schema mappings,
and that they do not necessarily reflect
the standard semantics of first-order quantifiers.
In Section~\ref{sec:deductive},
we then study several of the query answering semantics for deductive databases
in the context of data exchange.
The new \XGCWA-semantics is introduced and illustrated
in Section~\ref{sec:XGCWA},
and the data complexity of answering queries under the \XGCWA-semantics
is explored in Section~\ref{sec:complexity}.

\section{Preliminaries}
\label{sec:basics}

We use standard terminology from database theory,
but slightly different notation.
See, e.g., \cite{AHV95} for a comprehensive introduction to database theory.

\subsection{Databases}

A \emph{schema} is a finite set $\sigma$ of relation symbols,
where each $R \in \sigma$ has a fixed arity $\arity(R) \geq 1$.
An \emph{instance} $I$ over $\sigma$ assigns to each $R \in \sigma$
a finite relation $R^I$ of arity $\arity(R)$.
The \emph{active domain of $I$} (the set of all values that occur in $I$)
is denoted by $\dom(I)$.
As usual in data exchange, we assume that $\dom(I) \subseteq \Dom$,
where $\Dom$ is the union of two fixed disjoint infinite sets --
the set $\Const$ of all \emph{constants},
and the set $\Nulls$ of all \emph{(labeled) nulls}.
Constants are denoted by letters $a,b,c,\dotsc$ and variants like $a',a_1$.
Nulls serve as placeholders, or variables, for unknown constants;
we will denote them by $\bot$ and variants like $\bot',\bot_1$.
Let $\const(I) \isdef \dom(I) \intersection \Const$ and
$\nulls(I) \isdef \dom(I) \intersection \Nulls$.
An instance is called \emph{ground} if it contains no nulls.

An \emph{atom} is an expression of the form $R(\tup{t})$,
where $R$ is a relation symbol, and $\tup{t} \in \Dom^{\arity(R)}$.
We often identify an instance $I$ with the set of all atoms $R(\tup{t})$
with $\tup{t} \in R^I$,
that is, we often view $I$ as the set
$\set{R(\tup{t}) \mid R \in \sigma,\, \tup{t} \in R^I}$.
An atom $R(\tup{t})$ is called \emph{ground} if $\tup{t}$ contains no nulls.

We extend mappings $f\colon X \to Y$, where $X$ and $Y$ are arbitrary sets,
to tuples, atoms, and instances as follows.
For a tuple $\tup{t} = (t_1,\dotsc,t_n) \in X^n$,
we let $f(\tup{t}) \isdef (f(t_1),\dotsc,f(t_n))$;
for an atom $A = R(\tup{t})$, we let $f(A) \isdef R(f(\tup{t}))$;
and for an instance $I$, we let $f(I) \isdef \set{f(A) \mid A \in I}$.
A mapping $f\colon X \to Y$ is called \emph{legal} for an instance $I$
if $\dom(I) \subseteq X$, and $f(c) = c$ for all $c \in \const(I)$.
The set of all mappings that are legal for $I$ is denoted by $\legal(I)$.
For a tuple $\tup{t} = (t_1,\dotsc,t_k)$,
we sloppily write $f\colon \tup{t} \to Y$ for a mapping $f\colon X \to Y$
with $X = \set{t_1,\dotsc,t_k}$.
Given $\tup{t} = (t_1,\dotsc,t_k)$, $\tup{u} = (u_1,\dotsc,u_l)$,
and an element $v$,
we also write $v \in \tup{t}$ if $v \in \set{t_1,\dotsc,t_k}$,
and we let
\(
  \tup{t} \intersection \tup{u}
  \isdef
  \set{t_1,\dotsc,t_k} \intersection \set{u_1,\dotsc,u_l}.
\)

Let $I$ and $J$ be instances.
A \emph{homomorphism} from $I$ to $J$ is a mapping
$h\colon \dom(I) \to \dom(J)$ such that $h \in \legal(I)$
and $h(I) \subseteq J$.
If $h(I) = J$, then $J$ is called a \emph{homomorphic image} of $I$.
$I$ and $J$ are \emph{homomorphically equivalent}
if there is a homomorphism from $I$ to $J$, and a homomorphism from $J$ to $I$.
An \emph{isomorphism} from $I$ to $J$ is a homomorphism $h$ from $I$ to $J$
such that $h$ is bijective, and $h^{-1}$ is a homomorphism from $J$ to $I$.
If there is an isomorphism from $I$ to $J$,
we call $I$ and $J$ \emph{isomorphic}, and denote this by $I \isomorphic J$.
A \emph{core} is an instance $K$
such that there is no homomorphism from $K$ to a proper subinstance of $K$.
A \emph{core} of $I$ is a core $K \subseteq I$
such that there is a homomorphism from $I$ to $K$.
It is known \cite{HN:DM109} that if $I$ and $J$ are homomorphically equivalent,
$K$ is a core of $I$, and $K'$ is a core of $J$, then $K \isomorphic K'$.
In particular, any two cores of $I$ are isomorphic,
so that we can speak of \emph{the} core of $I$.

Given some property $P$ of instances,
a \emph{minimal} instance with property $P$
is an instance $I$ with property $P$
such that there is no instance $J \subsetneq I$ with property $P$.

\subsection{Queries}

As usual \cite{AHV95},
a $k$-ary query over a schema $\sigma$ is a mapping from instances
over $\sigma$ to $\Dom^k$
that is $C$-generic for some finite set $C \subseteq \Const$
(i.e., the query is invariant under renamings of values in $\Dom \setminus C$).
In the context of queries defined by logical formulas,
we will often use the words \emph{formula} and \emph{query} as synonyms.
Whenever we speak of a first-order formula (\FO-formula) over a schema $\sigma$,
we mean a FO formula over the vocabulary that consists of
all relation symbols in $\sigma$, and all constants in $\Const$.

Let $\phi$ be a \FO-formula over $\sigma$,
and let $\dom(\phi)$ be the set of all constants in $\phi$.
An assignment for $\phi$ in an instance $I$
is a mapping from the free variables of $\phi$ to $\dom(I) \union \dom(\phi)$,
which we extend to $\Const$ via $\alpha(c) \isdef c$ for all $c \in \Const$.
We write $I \models \phi(\alpha)$ to indicate that
$\phi$ is satisfied in $I$ under $\alpha$ in the naive sense.%
\footnote{Nulls that may occur in $I$ are treated as if they were constants.
  In general, this may lead to counter-intuitive semantics
  \cite{Lipski:TODS4-3,IL:JACM31-4},
  since distinct nulls may represent the same constant.}
The relation $\models$ is defined as usual,
the only difference being that constants in $\phi$
are interpreted by themselves,
and quantifiers range over $\dom(I) \union \dom(\phi)$.
That is, we apply the \emph{active domain semantics} \cite{AHV95}.
For example, we have $I \models R(u_1,\dotsc,u_{\arity(R)})(\alpha)$
precisely if $(\alpha(u_1),\dotsc,\alpha(u_{\arity(R)})) \in R^I$;
$I \models (u_1 = u_2)(\alpha)$ precisely if $\alpha(u_1) = \alpha(u_2)$;
and $I \models (\exists x\, \phi)(\alpha)$
precisely if there is an $v \in \dom(I) \union \dom(\phi)$
with $I \models \phi(\alpha[v/x])$,
where $\alpha[v/x]$ is the assignment defined like $\alpha$,
except that $x$ is mapped to $v$.
For an \FO-formula $\phi(x_1,\dotsc,x_k)$
and a tuple $\tup{u} = (u_1,\dotsc,u_k) \in (\dom(I) \union \dom(\phi))^k$,
we often write $I \models \phi(\tup{u})$ instead of $I \models \phi(\alpha)$,
where $\alpha(x_i) = u_i$ for each $i \in \set{1,\dotsc,k}$.

A query $q(\tup{x})$ over $\sigma$ is \emph{monotonic}
if $q(I) \subseteq q(J)$ for all instances $I,J$ over $\sigma$ with
$I \subseteq J$.
It is easy to see that all queries preserved under homomorphisms are monotonic.
Here, a query $q(\tup{x})$ over $\sigma$
is \emph{preserved under homomorphisms} if and only if for all instances $I,J$
over $\sigma$, all homomorphisms $h$ from $I$ to $J$,
and all tuples $\tup{t} \in q(I)$, we have $h(\tup{t}) \in q(J)$.
For example, conjunctive queries, unions of conjunctive queries,
Datalog queries, and the $\DatalogConstNeq$ queries of \cite{ABR:ICDT09}
are preserved under homomorphisms.
\emph{Unions of conjunctive queries with inequalities}
(see, e.g., \cite{FKMP:TCS336-1} for a definition)
are an example of monotonic queries that are not preserved under homomorphisms.

At various places in sections~\ref{sec:problems}--\ref{sec:XGCWA},
we will need formulas of the infinitary logic $\LInf$.
A $\LInf$ formula over a schema $\sigma$ is built
from atomic FO formulas over $\sigma$
using negation, existential quantification,
universal quantification, \emph{infinitary disjunctions} $\biglor \Phi$,
where $\Phi$ is an arbitrary set of $\LInf$ formulas over $\sigma$,
and \emph{infinitary conjunctions} $\bigland \Phi$,
where $\Phi$ is an arbitrary set of $\LInf$ formulas over $\sigma$.
The semantics of infinitary disjunctions and infinitary conjunctions
is the obvious one: for an assignment $\alpha$ of the variables that occur
in the formulas in $\Phi$ we have
$I \models \biglor \Phi(\alpha)$ if and only if there is some $\phi \in \Phi$
with $I \models \phi(\alpha)$,
and $I \models \bigland \Phi(\alpha)$ if and only if for all $\phi \in \Phi$,
$I \models \phi(\alpha)$.

\subsection{Data Exchange}
\label{sec:basics/data-exchange}

A \emph{schema mapping} $M = \Des$
consists of disjoint schemas $\sigma$ and $\tau$,
called \emph{source schema} and \emph{target schema}, respectively,
and a finite set $\Sigma$ of constraints in some logical formalism
over $\sigma \union \tau$ \cite{FKMP:TCS336-1}.
To introduce and to study query answering semantics in a general setting,
we assume that for all schema mappings $M = \Des$ considered in this article,
$\Sigma$ consists of $\LInf$ sentences over $\sigma \union \tau$
(that are $C$-generic for some finite $C \subseteq \Const$).
For \emph{algorithmic results}, however,
we restrict attention to schema mappings $M = \Des$,
where $\Sigma$ consists of
\emph{source-to-target tuple generating dependencies (st-tgds)} and
\emph{equality generating dependencies (egds)},
which have been prominently considered in data exchange.
Here, an \emph{st-tgd} is a FO sentence of the form
\begin{align*}
  \tgd,
\end{align*}
where $\phi$ is a conjunction of relational atomic FO formulas over $\sigma$
with free variables $\tup{x}\tup{y}$,
and $\psi$ is a conjunction of relational atomic FO formulas over $\tau$
with free variables $\tup{x}\tup{z}$.
A \emph{full st-tgd} is a st-tgd without existentially quantified
variables $\tup{z}$,
and a \emph{LAV tgd} is a st-tgd with a single atomic formula in $\phi$.
An \emph{egd} is a FO sentence of the form
\begin{align*}
  \egd,
\end{align*}
where $\phi$ is a conjunction of relational atomic FO formulas over $\tau$
with free variables $\tup{x}$,
and $x_i,x_j$ are variables in $\tup{x}$.

Let $M = \Des$ be a schema mapping.
A \emph{source instance} for $M$ is a \emph{ground} instance over $\sigma$,
and a \emph{target instance} for $M$ is an instance over $\tau$.
Given a source instance $S$ for $M$,
a \emph{solution} for $S$ under $M$ is a target instance $T$ for $M$
such that $S \union T \models \Sigma$,
that is, the instance $S \union T$ over $\sigma \union \tau$
satisfies all the constraints in $\Sigma$.

A \emph{universal solution} for $S$ under $M$ is a solution $T$ for $S$
under $M$ such that for all solutions $T'$ for $S$ under $M$
there is a homomorphism from $T$ to $T'$.
Note that all universal solutions for $S$ under $M$
are homomorphically equivalent,
which implies that their cores are isomorphic.
Hence, up to isomorphism there is a unique target instance,
denoted by $\core(M,S)$,
that is isomorphic to the cores of all universal solutions for $S$ under $M$.
For many schema mappings $M$,
$\core(M,S)$ is a solution for $S$ under $M$.
For example, if $\Sigma$ contains only st-tgds,
then $\core(M,S)$ is a solution for $S$ under $M$ \cite{FKP:TODS30-1},
which can be computed in polynomial time:

\begin{thm}[\cite{FKP:TODS30-1}]
  \label{thm:core-algorithm}
  Let $M = \Des$ be a schema mapping, where $\Sigma$ consists of st-tgds.
  Then there is a polynomial time algorithm that,
  given a source instance $S$ for $M$,
  outputs $\core(M,S)$.
\end{thm}

Besides $\core(M,S)$,
the \emph{canonical universal solution} for $S$ under $M$,
which is denoted by $\cansol(M,S)$,
plays an important role in data exchange.
In the following, we give the definition of $\cansol(M,S)$
from \cite{ABFL:PODS04} for the case that $\Sigma$ contains only st-tgds.
Let $\CJ$ be the set of all triples $(\theta,\tup{a},\tup{b})$ such that
$\theta$ is a st-tgd in $\Sigma$ of the form $\tgd$,
and $S \models \phi(\tup{a},\tup{b})$.
Starting from an empty target instance for $M$,
$\cansol(M,S)$ is created by adding atoms for each element in $\CJ$ as follows.
For each $j = (\theta,\tup{a},\tup{b}) \in \CJ$, where $\theta = \tgd$,
let $\tup{\bot}_j$ be a $\length{\tup{z}}$-tuple of pairwise distinct nulls
such that for all $j' \in \CJ$ with $j' \neq j$,
the set of nulls in $\tup{\bot}_j$ is disjoint from the set of nulls in
$\tup{\bot}_{j'}$,
and add the atoms in $\psi(\tup{a},\tup{\bot}_j)$ to the target instance.

\subsection{The Certain Answers}
\label{sec:basics/ca}

Given a query $q(\tup{x})$ over $\tau$,
and a set $\CT$ of instances over $\tau$,
we define the \emph{certain answers to $q(\tup{x})$ on $\CT$} by
\begin{align*}
  \cert(q,\CT)\, \isdef\, \bigintersection\, \set{q(T) \mid T \in \CT}.
\end{align*}
The set of the certain answers to $q(\tup{x})$ on $M$ and $S$
under the \emph{OWA-semantics}, as defined in \cite{FKMP:TCS336-1},
is then defined as
\begin{align*}
  \certOWA{M}{S}{q}\, \isdef\,
  \cert(q,\set{T \mid \text{$T$ is a solution for $S$ under $M$}}).
\end{align*}
If $q$ is preserved under homomorphisms,
$\certOWA{M}{S}{q}$ can be computed from a single universal solution:

\begin{prop}[\cite{FKMP:TCS336-1}]\label{prop:OWA-UCQ}
  Let $M = \Des$ be a schema mapping,
  let $S$ be a source instance for $M$,
  let $T$ be a universal solution for $S$ under $M$,
  and let $q(\tup{x})$ be a query that is preserved under homomorphisms.
  Then
  \[
    \certOWA{M}{S}{q}\, =\,
    \set{\tup{a} \in q(T) \mid \text{$\tup{a}$ contains only constants}}.
  \]
\end{prop}

\section{Review of Non-Monotonic Semantics
  in Relational Data Exchange}
\label{sec:problems}

As mentioned in the introduction,
for many non-monotonic queries,
the OWA-semantics is counter-intuitive.
In \cite{FKP:TODS30-1},
Fagin, Kolaitis, and Popa propose an alternative semantics,
where the set of answers to a query $q(\tup{x})$
on a schema mapping $M$ and a source instance $S$ is defined by
\(
  \cert(q,\set{T \mid \text{$T$ is a universal solution for $S$ under $M$}}).
\)
However, their semantics has similar problems as the OWA-semantics
\cite{ABFL:PODS04,HLS11}.
For example, note that Example~\ref{exa:copying} goes through unchanged
if the universal solution-semantics is used instead of the OWA-semantics.

Libkin \cite{Libkin:PODS06} realized that the counter-intuitive behavior
of the OWA-semantics and the universal solution-semantics
can be remedied by adopting the CWA.
He then introduced semantics based on the CWA.
These semantics were designed for schema mappings defined by st-tgds,%
\footnote{Strictly speaking, Libkin considered sentences of the form
  \(
    \forall \tup{x}\, \bigl(
      \phi(\tup{x}) \limplies \exists \tup{z}\, \psi(\tup{x},\tup{y})
    \bigr),
  \)
  where $\phi$ is a first-order formula over the source schema,
  and $\psi$ is a conjunction of relational atoms over the target schema.}
but were later extended to schema mappings defined by st-tgds,
t-tgds (see, e.g., \cite{FKMP:TCS336-1} for a definition of t-tgds),
and egds.

To define the CWA-semantics,
we need to introduce \emph{CWA-solutions}.
For simplicity, we give their definition only for schema mappings
defined by st-tgds.
Let $M = \Des$ be a schema mapping where $\Sigma$ consists of st-tgds,
and let $S$ be a source instance for $M$.
Libkin identified the following requirements
that should be satisfied by any CWA-solution $T$ for $S$ under $M$:
\begin{enumerate}[(1)]
\item
  Every atom in $T$ is justified by $M$ and $S$.
\item
  Justifications are applied at most once.
\item
  Every Boolean conjunctive query true in $T$
  is a logical consequence of $S$ and $\Sigma$.
\end{enumerate}
Here, a justification for an atom
consists of an st-tgd $\tgd$ in $\Sigma$
and assignments $\tup{a},\tup{b}$ for $\tup{x},\tup{y}$
such that $S \models \phi(\tup{a},\tup{b})$.
Such a justification can be applied with an assignment $\tup{u}$ for $\tup{z}$,
and would then justify the atoms in $\psi(\tup{a},\tup{u})$.
We require that all atoms that are justified belong to $T$.
Once all requirements are properly formalized,
it turns out that a \emph{CWA-solution} for $S$ under $M$
is a universal solution for $S$ under $M$
that is a homomorphic image of $\cansol(M,S)$ \cite{HLS11}.

We now recall the simplest of the four CWA-semantics
introduced in \cite{HLS11},
which we henceforth call \emph{CWA-semantics}.
Recall that solutions in data exchange may contain nulls.
Nulls represent unknown constants,
in particular, two distinct nulls may represent the same constant.
Treating nulls as constants may thus lead to counter-intuitive answers
\cite{IL:JACM31-4}.
A standard way to answer queries
while taking into account the semantics of nulls
is to return the certain answers.
To this end,
one views an instance $T$ over $\tau$ as a set $\rep(T)$ of ground instances
obtained by substituting concrete constants for each null.
Formally,
we let
\[
  \rep(T) \,\isdef\, \set{v(T) \mid \text{$v$ is a valuation of $T$}},
\]
where a \emph{valuation} of $T$ is a mapping
$v\colon \dom(T) \to \Const$ with $v \in \legal(T)$.
Then the certain answers to a query $q(\tup{x})$ on $T$ are
\[
  \cert(q,T)\, \isdef\, \cert(q,\rep(T)).
\]
Now, Libkin's idea was to define the
\emph{CWA-answers to a query $q(\tup{x})$ on $M$ and $S$}
by
\begin{align*}
  \certCWA{M}{S}{q}
  \,\isdef\,
    \bigintersection\,
    \set{\cert(q,T) \mid \text{$T$ is a CWA-solution for $S$ under $M$}}.
\end{align*}
That is, $\certCWA{M}{S}{q}$ is the set of the certain answers to $q(\tup{x})$
on the CWA-solutions for $S$ under $M$,
but instead of answering $q(\tup{x})$ on an individual CWA-solution $T$
by $q(T)$,
the certain answers are used.
As shown in \cite{HLS11},
$\certCWA{M}{S}{q}$ can be computed from the canonical solution:
$\certCWA{M}{S}{q} = \cert(q,\cansol(M,S))$.

While the CWA-semantics works well in a number of situations
(e.g., for schema mappings defined by full st-tgds, full t-tgds and egds),
and in particular resolves the problems observed in \cite{ABFL:PODS04,HLS11}
and Example~\ref{exa:copying},
one of the drawbacks of the CWA-semantics is
that it is not invariant under logically equivalent schema mappings.
That is,
there are schema mappings $M_1 = (\sigma,\tau,\Sigma_1)$ and
$M_2 = (\sigma,\tau,\Sigma_2)$,
where $\Sigma_1$ is logically equivalent to $\Sigma_2$,
and a query $q(\tup{x})$
such that the answers to $q(\tup{x})$ with respect to $M_1$
differ from the answers to $q(\tup{x})$ with respect to $M_2$:

\begin{exa}
  \label{ex:logical-equivalence}
  Let $M_1 = (\sigma,\tau,\Sigma_1)$ and $M_2 = (\sigma,\tau,\Sigma_2)$
  be schema mappings, where $\sigma$ contains a unary relation symbol $P$,
  $\tau$ contains a binary relation symbol $E$, and
  \begin{align*}
    \Sigma_1 & \,\isdef\, \bigl\{
      \forall x\, \bigl(P(x) \limplies E(x,x)\bigr)
    \bigr\}, \\
    \Sigma_2 & \,\isdef\, \Sigma_1 \union \bigl\{
      \forall x\, \bigl(P(x) \limplies \exists z\, E(x,z)\bigr)
    \bigr\}.
  \end{align*}
  Then $M_1$ and $M_2$ are logically equivalent.

  Let $S$ be an instance over $\sigma$ with $P^S = \set{a}$.
  Furthermore, let $T_1$ and $T_2$ be instances over $\tau$
  with $E^{T_1} = \set{(a,a)}$ and $E^{T_2} = \set{(a,a),(a,\bot)}$.
  Note that $T_i = \cansol(M_i,S)$ for each $i \in \set{1,2}$.
  Even more, $T_1$ is the unique CWA-solution for $S$ under $M_1$,
  whereas $T_2$ is a CWA-solution for $S$ under $M_2$.
  Thus, for the query
  \begin{align*}
    q(x) \,\isdef\, \exists z\, \bigl(
      E(x,z) \land \forall z' (E(x,z') \limplies z' = z)
    \bigr),
  \end{align*}
  we obtain different answers $\certCWA{M_1}{S}{q} = \set{a}$ and
  $\certCWA{M_2}{S}{q} = \emptyset$
  to the same query $q$ on logically equivalent schema mappings $M_1$ and $M_2$.
\end{exa}

\begin{rem}
  If we replace $q$ by the query $q'$
  which asks for all $x$ such that there are at least two $z$ with $E(x,z)$,
  then $q'$ is answered differently on $M_1$ and $M_2$
  with respect to the maybe answers-semantics from \cite{HLS11}.
  The other two semantics in \cite{HLS11} are invariant
  under logically equivalent schema mappings, though.
\end{rem}

Intuitively, logically equivalent schema mappings specify
the same translation of source data to the target schema,
so it seems natural that the answer to a query
is the same on logically equivalent schema mappings.
Furthermore, invariance under logically equivalent schema mappings
seems to be a good way to achieve a ``syntax-independent'' semantics,
that is, a semantics that does not depend on the concrete presentation
of the sentences of the schema mapping like the CWA-semantics.
Gottlob, Pichler, and Savenkov \cite{GPS:VLDB09}
suggest a different approach for achieving syntax-independence
that is based on first normalizing a schema mapping,
and then ``applying'' the semantics.
Let me mention that weaker notions of equivalence between schema mappings
have been considered in the literature \cite{FKNP:PODS08}.
Instead of requiring invariance under logical equivalence,
one could use any of these weaker notions.

There is another drawback of the CWA-semantics, though:
it does not necessarily reflect the standard semantics of FO quantifiers.
Typically, existential quantification $\exists x\, \phi$ is interpreted as:
there is one $x$ satisfying $\phi$, or there are two $x$ satisfying $\phi$,
or there are three $x$ satisfying $\phi$, and so on.
But this is not necessarily reflected by the CWA-semantics:

\begin{exa}
  \label{ex:Libkin-CWA-failure}
  Let $M = (\sigma,\tau,\Sigma)$ be the schema mapping,
  where $\sigma$ contains a unary relation symbol $P$,
  $\tau$ contains a binary relation symbol $E$,
  and $\Sigma$ consists of the st-tgd
  \begin{align*}
    \theta \,\isdef\, \forall x\, \bigl(
      P(x) \limplies \exists z\, E(x,z)
    \bigr).
  \end{align*}
  Let $S$ be the source instance for $M$ with $P^S = \set{a}$,
  and let $T \isdef \cansol(M,S)$.
  Then up to renaming of nulls,
  $T$ is the unique CWA-solution for $S$ under $M$,
  and $E^T = \set{(a,\bot)}$.
  Therefore, for the query
  \[
    q(x) \,\isdef\, \exists z\, \bigl(
      E(x,z) \land \forall z'\, (E(x,z') \limplies z' = z)
    \bigr),
  \]
  we have $\certCWA{M}{S}{q} = \set{a}$.
  In other words, $\certCWA{M}{S}{q}$ excludes the possibility
  that there is more than one value $z$ with $E(a,z)$.
  However, this seems to be too strict
  since if we interpret $\exists z\, E(x,z)$ in the standard first-order way,
  $S$ and $\theta$ tell us that there is one $z$ satisfying $E(a,z)$,
  or there are two $z$ satisfying $E(a,z)$,
  or there are three $z$ satisfying $E(a,z)$, and so on.
  In particular, they explicitly state that it is possible
  that there is more than one $z$ satisfying $E(a,z)$.
\end{exa}

\begin{rem}
  The example also shows that the potential certain answers semantics
  from \cite{HLS11}
  does not necessarily reflect the standard semantics of FO quantifiers.
  For the remaining two semantics in \cite{HLS11},
  we can replace $q$ by the query $q'$
  asking for all $x$ for which there are at least two $z$ with $E(x,z)$.
  Then the set of answers to $q'$ under those semantics will be empty.
  In other words, they tell us that it is not possible
  that there are more than two $z$ that satisfy $E(x,z)$,
  which is not desired.
\end{rem}

In \cite{LS:PODS08},
Libkin and Sirangelo proposed a generalization of the CWA-semantics
based on a combination of the CWA and the OWA,
using which we can resolve the problem described
in Example~\ref{ex:Libkin-CWA-failure}.
The idea is to annotate each position (occurrence of a variable)
in the head of an st-tgd as \emph{closed} or \emph{open},
where open positions correspond to those positions
where more than one value may be created.
For example, recall the schema mapping $M$ and the source instance $S$
from Example~\ref{ex:Libkin-CWA-failure},
and let $\alpha$ be the following annotation of $\theta$:
\begin{align*}
  \theta_\alpha \,\isdef\, \forall x\, \bigl(
    P(x) \limplies \exists z\, E(x^\textit{closed},z^\textit{open})
  \bigr).
\end{align*}
Then the valid solutions for $S$ under $M$ and $\alpha$
are all solutions for $S$ under $M$
of the form $\set{E(a,b) \mid b \in X}$, where $X$ is a finite set of constants.
That is,
the first position
-- which corresponds to the closed position in the head of $\theta_\alpha$ --
is restricted to the value $a$ assigned to $x$ when applying $\theta$ to $S$,
while the second position
-- which corresponds to the open position in the head of $\theta_\alpha$ --
is unrestricted in the sense that an arbitrary finite number
of values may be created at this position.
Let us write $\operatorname{sol}_\alpha(M,S)$
for the set of all these solutions.
The set of answers to a query $q(\tup{x})$ on $M$, $\alpha$ and $S$
is then
\[
  \cert_\alpha(q,M,S)\ \isdef\ \cert(q,\operatorname{sol}_\alpha(M,S)).
\]
In particular, it is easy to see that for the query $q$
in Example~\ref{ex:Libkin-CWA-failure}
we have $\cert_\alpha(q,M,S) = \emptyset$, as desired.
However, the ``mixed world'' semantics may still be counter-intuitive:

\begin{exa}
  \label{exa:3}
  Let $M = (\sigma,\tau,\Sigma)$ be defined by $\sigma = \set{R}$,
  $\tau = \set{E,F}$, and $\Sigma = \set{\theta}$,
  where
  \begin{align}
    \label{eq:3-tgd}
    \theta \,\isdef\,
    \forall x,y\, \Bigl(R(x,y) \limplies
      \exists z\, \bigl(E(x,z) \land F(z,y)\bigr)\Bigr).
  \end{align}
  Intuitively, $\theta$ states that ``if $R(x,y)$,
  then there is at least one $z$ such that $E(x,z)$ and $F(z,y)$ hold.''
  There could be exactly one such $z$,
  but there could also be more than one such $z$.
  The possibility that there are precisely two such $z$,
  or precisely three such $z$ et cetera is perfectly consistent with $\theta$,
  and should not be denied when answering queries.
  So, it seems that the only solutions for $S = \set{R(a,b)}$ under $M$
  should be those solutions $T$
  for which there is a finite set $X \subseteq \Const$
  such that $T = \set{E(a,x) \mid x \in X} \union \set{F(x,b) \mid x \in X}$.
  In particular, we should expect that the answer to
  \[
    q(x,y) \,\isdef\, \exists^{= 1} z\, \bigl(E(x,z) \land F(z,y)\bigr)
  \]
  on $M$ and $S$ is empty,
  and that the answer to
  \[
    q'(x) \,\isdef\, \forall z\, \bigl(E(x,z) \limplies \exists y\, F(z,y)\bigr)
  \]
  on $M$ and $S$ is $\set{a}$.

  Now consider an annotation $\alpha$ for $\theta$
  where the occurrence of $z$ in $E(x,z)$ is \emph{closed}.
  According to the definition in \cite{LS:PODS08},
  for all valid solutions $T$ for $S$ under $M$ and $\alpha$,
  and all tuples $(c,d),(c',d') \in E^T$ we have $d = d'$,
  so that $\cert_\alpha(q,M,S) = \set{(a,b)}$.
  Hence, $\cert_\alpha(q,M,S)$ excludes the possibility
  that there is more than one $z$ satisfying $E(x,z)$ and $F(z,y)$,
  although $\theta$ explicitly states that it is possible
  that more than one such $z$ exists.
  The same is true if the occurrence of $z$ in $F(z,y)$ is \emph{closed}.

  Let us finally consider an annotation $\alpha'$ where both occurrences of $z$
  in the head of $\theta$ are \emph{open}.
  According to the definition in \cite{LS:PODS08},
  we have $\cert_{\alpha'}(q',M,S) = \emptyset$,
  since
  \[
    T^* \,\isdef\, \set{E(a,c),E(a,c'),F(c,b)}
  \]
  would be a valid solution for $S$ under $M$ and $\alpha'$
  according to this definition.
  Intuitively, the ``mixed world'' semantics is ``too open''
  in that it allows $E(a,c')$ to occur in $T^*$
  without enforcing that the corresponding atom $F(c',b)$ is present in $T^*$.
\end{exa}

\begin{rem}
  \label{rem:inclusive}
  The existential quantifier in \eqref{eq:3-tgd}
  can be expressed via an infinite disjunction
  over all possible choices of values for $z$
  (recall that nulls are just place-holders for unknown constants,
   so we do not have to consider nulls here),
  as in
  \begin{align}
    \label{eq:3-tgd/d}
    \theta' \,\isdef\,
     \forall x,y \left(R(x,y) \limplies
      \biglor_{c \in \Const} \bigl(E(x,c) \land F(c,y)\bigr)\right).
  \end{align}
  The interpretation of the existential quantifier
  discussed in Example~\ref{exa:3}
  then corresponds to an inclusive interpretation of the disjunction
  in \eqref{eq:3-tgd/d}.
  Intuitively, the desired set of solutions for $S = \set{R(a,b)}$ under $M$ --
  the set of solutions of the form
  $T = \set{E(a,x) \mid x \in X} \union \set{F(x,b) \mid x \in X}$
  for some finite set $X \subseteq \Const$ --
  is the smallest set of solutions that reflects an inclusive interpretation
  of the disjunction in \eqref{eq:3-tgd/d}.
\end{rem}

\begin{rem}
  Afrati and Kolaitis \cite{AK:PODS08}
  showed that a restriction of Libkin's CWA is useful
  for answering \emph{aggregate queries}.
  Their semantics is defined with respect to schema mappings
  specified by st-tgds.
  In principle, we could use it to answer non-aggregate queries
  as follows.
  Let $M = \Des$ be a schema mapping where $\Sigma$ is a set of st-tgds,
  let $S$ be a source instance for $M$,
  and let $q(\tup{x})$ be a query over $\tau$.
  Instead of answering $q(\tup{x})$ by the certain answers to $q(\tup{x})$
  on $\rep(\core(M,S))$,
  we answer $q(\tup{x})$ by the certain answers to $q(\tup{x})$
  on the set of all \emph{endomorphic images} of $\cansol(M,S)$.
  Here, an endomorphic image of $\cansol(M,S)$
  is an instance $T$
  such that $T = h(\cansol(M,S))$ for some homomorphism $h$
  from $\cansol(M,S)$ to $\cansol(M,S)$.
  However, the endomorphic images-semantics seems to be too strong
  -- it is stronger than the CWA-semantics.
  In particular, Example~\ref{ex:logical-equivalence}
  shows that the endomorphic images-semantics is not invariant
  under logically equivalent schema mappings,
  and Example~\ref{ex:Libkin-CWA-failure}
  shows that it does not necessarily reflect the standard semantics
  of FO quantifiers.
\end{rem}

Our goal is to develop a semantics for answering non-monotonic queries
that is invariant under logically equivalent schema mappings,
and is just ``open enough'' to interpret existential quantifiers
(when viewed as an infinite disjunction,
 as suggested in Remark~\ref{rem:inclusive}) inclusively.
We start by studying semantics for answering non-monotonic queries
on deductive databases.

\begin{rem}
  Note the following side-effect
  of invariance under logical equivalent schema mappings.
  Consider the schema mappings $M_1$, $M_2$ and source instance $S$
  from Example~\ref{ex:logical-equivalence}.
  Then under a reasonable closed world semantics,
  $T = \set{E(a,a)}$ should be the only ``valid solution'' for $S$ under $M_1$.
  Invariance under logical equivalence thus enforces
  that $T$ is also the only ``valid solution'' for $S$ under $M_2$.
  This seems to be counter-intuitive,
  but only as long as one considers the two st-tgds in $M_2$ isolated.
\end{rem}

\section{Deductive Databases and Relational Data Exchange}
\label{sec:deductive}

Query answering on deductive databases
and query answering in relational data exchange
are very similar problems.
Both require answering a query on a ground database
that is equipped with a set of constraints.
On the other hand,
answering non-monotonic queries on deductive databases
is a well-studied topic.
In this section, we translate some of the semantics
that were proposed to answer non-monotonic queries on deductive databases
into the context of relational data exchange.

A \emph{deductive database} \cite{GMN:CS84} over a schema $\sigma$ is a set
of FO sentences, called \emph{clauses}, of the form
\begin{align}
  \label{eq:db}
  \forall \tup{x} \bigl(
    \lnot R_1(\tup{y}_1) \lor \dotsb \lor \lnot R_m(\tup{y}_m) \lor
    R'_1(\tup{z}_1) \lor \dotsb \lor R'_n(\tup{z}_n)
  \bigr),
\end{align}
where $m$ and $n$ are nonnegative integers with $m + n \geq 1$,
$R_1,\dotsc,R_m,R'_1,\dotsc,R'_n$ are relation symbols in $\sigma$,
and $\tup{y}_1,\dotsc,\tup{y}_m,\tup{z}_1,\dotsc,\tup{z}_n$ are tuples
containing elements of $\Const$ and $\tup{x}$.
A \emph{model} of a deductive database $\DDB$ over $\sigma$
is a \emph{ground} instance $I$ over $\sigma$ with $I \models \DDB$
(i.e., $I$ satisfies all clauses in $\DDB$),
and a query $q(\tup{x})$ is usually answered by $\cert(q,\CI)$,
where $\CI$ is a set of models of $\DDB$ that depends on the particular
semantics.

Several semantics for answering non-monotonic queries
on deductive databases were proposed
(see, e.g., \cite{Reiter:LD78,Minker:CADE82,YH:JAR1-2,Chan:TLDE5-2},
 and the survey articles \cite{GMN:CS84,DFN:AR2}).
Often, these semantics can be applied with some minor modifications
to more general sets of logical sentences such as
\begin{align*}
  \DDB_{M,S} \,\isdef\, \Sigma
  \union
  \set{R(\tup{t}) \mid R \in \sigma,\, \tup{t} \in R^S}
  \union
  \set{\lnot R(\tup{t}) \mid
    R \in \sigma,\, \tup{t} \in \Const^{\arity(R)} \setminus R^S},
\end{align*}
where $M = \Des$ is a schema mapping, and $S$ is a source instance for $M$.
Note that, if $\Sigma$ consists only of full st-tgds,
then $\DDB_{M,S}$ is logically equivalent to a deductive database,
since any full st-tgd of the form
\(
  \forall \tup{x} (
    R_1(\tup{y}_1) \land \dotsb \land R_m(\tup{y}_m) \limplies
    R'(\tup{z})
  )
\)
is logically equivalent to the clause
\(
  \forall \tup{x} (
    \lnot R_1(\tup{y}_1) \lor \dotsb \lor \lnot R_m(\tup{y}_m) \lor R'(\tup{z})
  ).
\)

In the following,
we concentrate on Reiter's \emph{CWA} \cite{Reiter:LD78},
since this is the basic assumption underlying all previous approaches
for non-monotonic query answering in relational data exchange.
However, we also consider variants of Reiter's CWA  such as
semantics based on Minker's \emph{generalized CWA (GCWA)} \cite{Minker:CADE82},
Yahya's and Henschen's \emph{extended GCWA (EGCWA)} \cite{YH:JAR1-2}
as well as Chan's \emph{possible worlds semantics (PWS)} \cite{Chan:TLDE5-2}.

\subsection{Reiter's Closed World Assumption (RCWA)}
\label{sec:deductive/CWA}

Reiter's \emph{closed world assumption (RCWA)},%
\footnote{We write RCWA instead of CWA
  to avoid confusion with Libkin's formalization of the CWA.}
formalized by Reiter in \cite{Reiter:LD78},
assumes that every ground atom that is not implied by a database is false.
This is a common assumption for relational databases.

Reiter formalized the RCWA
as follows.
For a deductive database $\DDB$ and a formula $\phi$, we write
$\DDB \models \phi$ if and only if for all instances $I$ with $I \models \DDB$,
we have $I \models \phi$.
Given a deductive database $\DDB$ over a schema $\sigma$,
let
\begin{align*}
  \overline{\DDB} \,\isdef\, \set{
    \lnot R(\tup{t})
    \mid
    R \in \sigma,\, \tup{t} \in \Const^{\arity(R)},\, \DDB \not\models R(\tup{t})
  },
\end{align*}
which contains negations of all \emph{ground} atoms $R(\tup{t})$
(i.e., $\tup{t}$ is a tuple of constants)
that are assumed to be false under the RCWA.
The models of $\DDB \union \overline{\DDB}$
are called \emph{RCWA-models} of $\DDB$.
Under the RCWA,
a query $q(\tup{x})$ over $\sigma$ is answered by $\cert(q,\CI)$,
where $\CI$ is the set of all RCWA-models of $\DDB$.

Translated into the relational data exchange framework, we obtain:

\begin{defi}[RCWA-solution, RCWA-answers]
  \label{def:RCWA-solutions}
  Let $M = \Des$ be a schema mapping, let $S$ be a source instance for $M$,
  and let $q(\tup{x})$ be a query over $\tau$.
  \begin{enumerate}[(1)]
  \item
    An \emph{RCWA-solution} for $S$ under $M$
    is a ground target instance $T$ for $M$
    such that $S \union T$ is a RCWA-model of $\DDB_{M,S}$
    (recall the definition of $\DDB_{M,S}$ from \eqref{eq:db}),
  \item
    We call $\certRCWA{M}{S}{q} \isdef \cert(q,\CI)$,
    where $\CI$ is the set of all RCWA-solutions for $S$ under $M$,
    the \emph{RCWA-answers to $q(\tup{x})$ on $M$ and $S$}.
  \end{enumerate}
\end{defi}

\noindent Note that RCWA-solutions are \emph{ground} (i.e., contain no nulls),
in contrast to other notions of solutions such as plain solutions,
universal solutions, or CWA-solutions presented in previous sections.

The RCWA is a very strong assumption.
For example,
if an RCWA-solution for $S$ under $M$ exists,
it is the unique minimal (ground) solution for $S$ under $M$:

\begin{prop}
  \label{prop:RCWA-solutions}
  Let $M$ be a schema mapping, and let $S$ be a source instance for $M$.
  Then a solution $T$ for $S$ under $M$
  is an RCWA-solution for $S$ under $M$
  if and only if $T$ is contained in every ground solution for $S$ under $M$.
\end{prop}

\begin{proof}
  Suppose $T$ is an RCWA-solution for $S$ under $M$,
  and let $T'$ be a ground solution for $S$ under $M$.
  If there is an atom $R(\tup{t}) \in T \setminus T'$,
  then $D_{M,S} \not\models R(\tup{t})$,
  so that $\lnot R(\tup{t}) \in \overline{D_{M,S}}$,
  and hence $T$ is no RCWA-solution for $S$ under $M$.
  Therefore, $T \subseteq T'$.
  To prove the other direction,
  suppose that $T$ is contained in every ground solution for $S$ under $M$.
  Then, for all atoms $R(\tup{t}) \in T$ we have $D_{M,S} \models R(\tup{t})$,
  so that $T$ is an RCWA-solution for $S$ under $M$.
\end{proof}

It is also not hard to see that $\fcertRCWA$ coincides with $\fcertCWA$
on schema mappings defined by full st-tgds.

\begin{prop}
  \label{prop:RCWA-vs-CWA-for-full-tgds}
  Let $M = \Des$ be a schema mapping, where $\Sigma$ consists of full st-tgds,
  let $S$ be a source instance for $M$,
  and let $q(\tup{x})$ be a query over $\tau$.
  Then, $\certRCWA{M}{S}{q} = \certCWA{M}{S}{q}$.
\end{prop}

\begin{proof}
  Since $\Sigma$ consists of full st-tgds,
  there is a unique minimal ground solution $T_0$ for $S$ under $M$,
  which is also the unique CWA-solution for $S$ under $M$,
  and, by Proposition~\ref{prop:RCWA-solutions},
  the unique RCWA-solution for $S$ under $M$.
  Consequently, $\certRCWA{M}{S}{q} = \certCWA{M}{S}{q} = q(T_0)$.
\end{proof}

However, for schema mappings that contain non-full st-tgds,
$\fcertRCWA$ may lead to answers that are inconsistent with $M$ and $S$.
This is illustrated by the following example,
which is based on Example 8 in \cite{Reiter:LD78}.

\begin{exa}
  \label{ex:CWA-failure}
  Let $M = (\set{P},\set{E},\Sigma)$,
  where $\Sigma \isdef \set{\forall x (P(x) \to \exists z\, E(x,z))}$,
  and let $S$ be the source instance for $M$ with $P^S = \set{a}$.
  Since there is no unique minimal ground solution,
  Proposition~\ref{prop:RCWA-solutions} implies
  that there is no RCWA-solution for $S$ under $M$.
  Consequently, the RCWA-answers to
  \(
    q(x) \isdef \exists z\, E(x,z)
  \)
  on $M$ and $S$ are empty.
  In other words,
  $\certRCWA{M}{S}{q}$ tells us that there is no value $z$ satisfying $E(a,z)$.
  This is clearly inconsistent with $M$ and $S$,
  which tell us that there is a value $z$ satisfying $E(a,z)$,
  and hence that the set of answers should be $\set{a}$.
\end{exa}

\subsection{The Generalized Closed World Assumption (GCWA)}
\label{sec:deductive/GCWA}

Minker \cite{Minker:CADE82} extended Reiter's CWA to the
\emph{generalized closed world assumption (GCWA)} as follows.
Recall the definition of \emph{minimal instance possessing some property}
from Section~\ref{sec:basics}.
Let $\DDB$ be a deductive database over a schema $\sigma$.
A \emph{minimal model} of $D$ is a minimal instance
with the property of being a model of $D$.
Let
\begin{align*}
  \overline{\overline{\DDB}} \,\isdef\, \set{
    \lnot R(\tup{t}) \mid
    \text{$R \in \sigma$, $\tup{t} \in \Const^{\arity(R)}$,
      $\tup{t} \notin R^I$ for all minimal models $I$ of $\DDB$}
  },
\end{align*}
which, analogous to $\overline{\DDB}$ for the case of the RCWA,
contains negations of all ground atoms
that are assumed to be false under the GCWA.
The models of $\DDB \union \overline{\overline{\DDB}}$ are called
\emph{GCWA-models} of $\DDB$,
and a query $q(\tup{x})$ over $\sigma$ is answered by $\cert(q,\CI)$,
where $\CI$ is the set of all GCWA-models of $\DDB$.

The intuition behind the above definitions is that each ground atom in some
minimal model of $\DDB$ is in some sense an atom that $\DDB$ ``speaks'' about.
For ground atoms that do not occur in any minimal model of $\DDB$,
this means that they are merely ``invented'',
and can therefore safely be assumed to be false.

Translated into the relational data exchange framework, we obtain:

\begin{defi}[GCWA-solution, GCWA-answers]
  \label{def:GCWA-solutions}
  Let $M = \Des$ be a schema mapping, let $S$ be a source instance for $M$,
  and let $q(\tup{x})$ be a query over $\tau$.
  \begin{enumerate}[(1)]
  \item
    A \emph{GCWA-solution} for $S$ under $M$
    is a ground target instance $T$ for $M$
    such that $S \union T$ is a GCWA-model of $\DDB_{M,S}$.
  \item
    We call $\certGCWA{M}{S}{q} \isdef \cert(q,\CI)$,
    where $\CI$ is the set of the GCWA-solutions for $S$ under $M$,
    the \emph{GCWA-answers to $q(\tup{x})$ on $M$ and $S$}.
  \end{enumerate}
\end{defi}

\noindent We have the following characterization of GCWA-solutions:

\begin{prop}
  \label{prop:GCWA-solutions}
  Let $M$ be a schema mapping, and let $S$ be a source instance for $M$.
  Then a solution $T$ for $S$ under $M$
  is an GCWA-solution for $S$ under $M$
  if and only if there is a set $\CT$ of minimal ground solutions
  for $S$ under $M$
  such that $T \subseteq \bigunion \CT$.
\end{prop}

\begin{proof}
  Suppose $T$ is a GCWA-solution for $S$ under $M$.
  Then for each atom $A$ in $T$,
  there is a minimal ground solution $T_A$ for $S$ under $M$
  with $A \in T_A$
  (otherwise, $\lnot A \in \overline{\overline{D_{M,S}}}$,
   so that $T$ would not be a GCWA-solution for $S$ under $M$).
  But then we have $T \subseteq \bigunion_{A \in T} T_A$.

  Suppose now that $\CT$ is a set of minimal ground solutions for $S$ under $M$
  such that $T \subseteq \bigunion \CT$.
  Then, for all atoms $A \in T$
  we have $\lnot A \notin \overline{\overline{D_{M,S}}}$,
  and it follows that $\CT$ is a GCWA-solution for $S$ under $M$.
\end{proof}

Similar to the RCWA-answers semantics,
it can be shown that $\fcertGCWA$ coincides with $\fcertCWA$
on schema mappings defined by full st-tgds.
Moreover, $\fcertGCWA$ leads to the desired answers to the query
in Example~\ref{ex:CWA-failure}:

\begin{exa}
  \label{ex:CWA-failure/GCWA-fix}
  Recall the schema mapping $M$, the source instance $S$, and the query $q$
  from Example~\ref{ex:CWA-failure}.
  We now have
  \begin{align*}
    \overline{\overline{\DDB_{M,S}}} =
    \set{\lnot P(b) \mid b \in \Const,\, b \neq a}
    \union
    \set{\lnot E(b,c) \mid b, c \in \Const,\, b \neq a},
  \end{align*}
  because each atom of the form $E(a,c)$
  is true in some minimal model of $\DDB_{M,S}$,
  and each atom of the form $E(b,c)$ with $b \neq a$
  is false in all minimal models of $\DDB_{M,S}$.
  Therefore, the GCWA-solutions for $S$ under $M$
  are precisely the target instances $T$ for $M$
  for which there is a finite nonempty set $B \subseteq \Const$ with $T = T_B$,
  where $E^{T_B} = \set{(a,b) \mid b \in B}$.
  It follows that $\certGCWA{M}{S}{q} = \set{a}$, as desired.
\end{exa}

Nevertheless, there are cases where the GCWA is still quite unsatisfactory,
as shown by the following example:

\begin{exa}
  \label{ex:GCWA-problem}
  Consider a slight extension of the schema mapping from
  Example~\ref{ex:CWA-failure},
  namely $M = (\set{P},\set{E,F},\Sigma)$,
  where $\Sigma$ consists of the st-tgd
  \begin{align*}
    \theta \,\isdef\,
    \forall x\, \Bigl(
      P(x) \limplies \exists z_1 \exists z_2\,
      \bigl(E(x,z_1) \land F(z_1,z_2)\bigr)
    \Bigr).
  \end{align*}
  Let $S$ be the source instance for $M$ with $P^S = \set{a}$.
  Then,
  \begin{align*}
    \overline{\overline{\DDB_{M,S}}} =
    \set{\lnot P(b) \mid b \in \Const,\, b \neq a} \union
    \set{
      \lnot E(b,c) \mid b, c \in \Const,\, b \neq a
    }.
  \end{align*}
  Note that for all $b,c \in \Const$
  we have $\lnot F(b,c) \notin \overline{\overline{\DDB_{M,S}}}$,
  since the target instance $T$ for $M$ with $P^T = \set{a}$,
  $E^T = \set{(a,b)}$ and $F^T = \set{(b,c)}$
  is a minimal model of $\DDB_{M,S}$.
  So, the GCWA-solutions for $S$ under $M$ are the target instances $T$
  for $M$ for which there is a finite nonempty set $B \subseteq \Const$
  with the following properties:
  (1) $E^T = \set{(a,b) \mid b \in B}$, and
  (2) for at least one $b \in B$
      there is some $c \in \Const$ with $(b,c) \in F^T$.
  In particular, the target instance $T^*$ with $E^{T^*} = \set{(a,b)}$ and
  $F^{T^*} = \set{(b,c),(d,e)}$ is a GCWA-solution for $S$ under $M$.
  For the Boolean query
  \[
    q \,\isdef\,
    \forall z_1 \forall z_2\, \bigl(
      F(z_1,z_2) \limplies \exists x\, E(x,z_1)
    \bigr)
  \]
  we thus have $\certGCWA{M}{S}{q} = \emptyset$.

  So, $\certGCWA{M}{S}{q}$ tells us that it is possible that there is a tuple
  $(b,c)$ in $F$ for which $(a,b)$ is not in $E$.
  However, $\theta$ and $S$ do not ``mention'' this possibility.
  In particular, $\theta$ and $S$ only tell us that there are one or more pairs
  $(b,c) \in \Const^2$ such that $E(a,b)$ and $F(b,c)$ occur together
  in a solution.
  Thus, whenever $E(a,b)$ is present for some $b \in \Const$,
  then $F(b,c)$ should be present for some $c \in \Const$.
  Similarly, whenever $F(b,c)$ is present for some $b,c \in \Const$,
  then $E(a,b)$ should be present.
\end{exa}

\subsection{Extensions of the GCWA}
\label{sec:deductive/other}

Various extensions of the GCWA have been proposed.
One of these extensions is the \emph{extended GCWA (EGCWA)}
by Yahya and Henschen \cite{YH:JAR1-2},
which restricts the set of models of a deductive database $\DDB$
to the minimal models of $\DDB$.
So, given a schema mapping $M = \Des$ and a source instance $S$ for $M$,
an \emph{EGCWA-solution} for $S$ under $M$ can be defined as
a ground minimal solution for $S$ under $M$,
and given a query $q(\tup{x})$ we can define
\[
  \certEGCWA{M}{S}{q} \,\isdef\, \cert(q,\CI),
\]
where $\CI$ is the set of all EGCWA-solutions for $S$ under $M$.
Then, for the schema mapping $M$, the source instance $S$ for $M$,
and the query $q$ in Example~\ref{ex:CWA-failure/GCWA-fix},
$\certEGCWA{M}{S}{q} = \certGCWA{M}{S}{q}$,
and for the schema mapping $M$, the source instance $S$ for $M$,
and the query $q$ in Example~\ref{ex:GCWA-problem},
$\certEGCWA{M}{S}{q} \neq \emptyset$, as desired.
However, the EGCWA seems to be too strong in the sense that it removes
too many solutions from the set of all solutions.
More precisely, it interprets existential quantifiers
(when viewed as disjunctions) exclusively rather than inclusively.
We illustrate this by the following example.%
\footnote{Example~\ref{ex:Libkin-CWA-failure} illustrates this as well,
  but Example~\ref{ex:EGCWA-failure} seems to make it more clear
  why it may be desirable to interpret existential quantifiers inclusively.}

\begin{exa}
  \label{ex:EGCWA-failure}
  Let $M = (\set{P},\set{E},\Sigma)$ be a schema mapping,
  where $\Sigma$ consists of
  \[
    \theta\, =\,
    \forall x\, \bigl(P(x) \limplies \exists^{[2,3]} z\, E(x,z)\bigr),
  \]
  where $\exists^{[2,3]} z\, E(x,z)$ is an abbreviation for ``there exist two
  or three $z$ such that $E(x,z)$''.
  Let $S$ be the source instance for $M$ with $P^S = \set{a}$.
  Then the minimal solutions for $S$ under $M$
  have the form $\set{E(a,b_1),E(a,b_2)}$,
  where $b_1,b_2$ are distinct constants.
  Thus, for
  \begin{align*}
    q(x) \,\isdef\, \exists z_1 \exists z_2\, \Bigl(
      E(x,z_1) \land E(x,z_2) \land
        \forall z_3\, \bigl(
          E(x,z_3) \limplies (z_3 = z_1 \lor z_3 = z_2)
        \bigr)
    \Bigr),
  \end{align*}
  we have $\certEGCWA{M}{S}{q} = \set{a}$.
  In other words, the answer $\certEGCWA{M}{S}{q}$ excludes the possibility
  that there are three distinct values $b_1,b_2,b_3$ with $E(a,b_i)$
  for each $i \in \set{1,2,3}$.
  But $\theta$ and $S$ explicitly mention this possibility.
  Thus, intuitively, $\fcertEGCWA$ is inconsistent with $M$ and $S$.
\end{exa}

To conclude this section,
let us consider the \emph{possible worlds semantics (PWS)} by Chan
\cite{Chan:TLDE5-2}.
A natural translation of the PWS for the case of schema mappings defined
by st-tgds is as follows:
Let $M = \Des$ be a schema mapping, where $\Sigma$ is a set of st-tgds,
and let $S$ be a source instance for $M$.
The definition of a PWS-solution for $S$ under $M$ can be given in terms of
\emph{justifications}, as in \cite{HLS11}.
Given a target instance $T$ for $M$ and an atom $R(\tup{t}) \in T$,
we say that $R(\tup{t})$ is \emph{justified} in $T$ under $M$ and $S$
if and only if there is a st-tgd $\tgd$ in $\Sigma$,
tuples $\tup{a},\tup{b}$ over $\dom(S)$ with $S \models \phi(\tup{a},\tup{b})$,
and a tuple $\tup{u}$ over $\dom(T)$ such that
$T \models \psi(\tup{a},\tup{u})$, and $R(\tup{t})$ is one of the atoms
in $\psi(\tup{a},\tup{u})$.
A PWS-solution for $S$ under $M$ is then a ground solution $T$ for $S$
under $M$ such that all atoms in $T$ are justified in $T$ under $M$ and $S$.
For a query $q$ over $\tau$, we let
\[
  \certPWS{M}{S}{q} \,\isdef\, cert(q,\CI),
\]
where $\CI$ is the set of all PWS-solutions for $S$ under $M$.
However, $\fcertPWS$ does not respect logical equivalence of schema mappings
which can be easily verified using the schema mapping, the source instance
and the query from Example~\ref{ex:logical-equivalence}.

\section{The \texorpdfstring{\XGCWA}{GCWA*}-Semantics}
\label{sec:XGCWA}

We now introduce the \XGCWA-semantics,
and argue that it has the desired properties --
being invariant under logically equivalent schema mappings,
and being just ``open enough'' to interpret existential quantifiers
inclusively.

Among the semantics considered in the previous sections,
the GCWA-semantics is closest to the desired semantics.
For instance, consider the schema mapping $M = (\set{P},\set{E},\Sigma)$,
where $\Sigma$ consists of
\(
  \forall x\, (P(x) \limplies \exists z\, E(x,z)),
\)
and the source instance $S$ for $M$ with $P^S = \set{a}$
from Example~\ref{ex:CWA-failure}.
Let $\CT$ be the set of all GCWA-solutions for $S$ under $M$.
As shown in Example~\ref{ex:CWA-failure/GCWA-fix},
$\CT$ consists of all target instances $T$ for $M$
such that there is a nonempty finite set $B \subseteq \Const$ with $T = T_B$,
where
\(
  E^{T_B} = \set{(a,b) \mid b \in B}.
\)
The set $\CT$ is precisely as we would like the set of solutions to be.
Intuitively, it precisely captures what is expressed by $M$ and $S$:
there is one $b \in \Const$ satisfying $E(a,b)$,
or there are two distinct $b_1,b_2 \in \Const$
satisfying $E(a,b_1)$ and $E(a,b_2)$,
or there are three distinct $b_1,b_2,b_3 \in \Const$
satisfying $E(a,b_1)$, $E(a,b_2)$ and $E(a,b_3)$,
and so on.
The case that there are $n$ distinct $b_1,\dotsc,b_n \in \Const$
such that $E(a,b_i)$ holds for each $i \in \set{1,\dotsc,n}$
is captured precisely by $T_B$, where $B \isdef \set{b_1,\dotsc,b_n}$.
However, as we have argued in Example~\ref{ex:GCWA-problem},
with respect to other schema mappings the GCWA is still ``too open''.

Note that the set $\CT$ in the above example
is the set of all ground solutions for $S$ under $M$
that are unions of minimal solutions.
Indeed, it seems to be a good idea to use the set of all such solutions
as the set of ``valid solutions''.
As we have done in Remark~\ref{rem:inclusive},
we can express the existential quantifier in the st-tgd
from Example~\ref{ex:CWA-failure}
equivalently by an infinite disjunction,
resulting in the following $\LInf$-sentence:
\begin{align*}
  \theta' \,\isdef\,
  \forall x \left(P(x) \limplies
    \biglor_{c \in \Const} E(x,c)\right).
\end{align*}
Then the ground minimal solutions for a source instance $S$
under $M' = (\set{P},\set{E},\set{\theta'})$
correspond to the disjuncts of the disjunction in $\theta'$,
and an inclusive interpretation of this disjunction is guaranteed
by taking all ground solutions that are unions of minimal solutions
as ``valid solutions''.

This can be generalized to other schema mappings defined by st-tgds.
For instance, recall the schema mapping $M$,
and the source instance $S$ for $M$ from Example~\ref{ex:GCWA-problem},
where the GCWA is ``too open''.
Again, we can express the st-tgd $\theta$ that defines $M$
by an $\LInf$-sentence where the existential quantifier in $\theta$
is replaced by an infinite disjunction:
\begin{align*}
  \theta' \,\isdef\,
  \forall x\, \left(
    P(x) \limplies \biglor_{c_1,c_2 \in \Const}
    \bigl(E(x,c_1) \land F(c_1,c_2)\bigr)
  \right).
\end{align*}
Then the ground minimal solutions for $S$ under $M$
correspond to the disjuncts of the disjunction in $\theta'$,
and an inclusive interpretation of this disjunction is guaranteed
by taking the set $\CT$ of all ground solutions for $S$ under $M$
that are unions of minimal solutions as ``valid solutions''.
That is, we take the set $\CT$ of all ground target instances $T$ for $M$
such that
\(
  E^T = \set{(a,b) \mid \text{$(b,c) \in F^T$ for some $c \in \Const$}}
\)
and
\(
  F^T \neq \emptyset.
\)
Indeed, the set of the certain answers to the query $q$
from Example~\ref{ex:GCWA-problem} on $\CT$ is nonempty, as desired.

The preceding two examples suggest
to answer queries by the certain answers on the set of all ground solutions
that are unions of minimal solutions.
Let us call such solutions \emph{\XGCWA-solutions} for the moment:

\begin{defi}[working definition]
  Let $M = \Des$ be a schema mapping, let $S$ be a source instance for $M$,
  and let $q(\tup{x})$ be a query over $\tau$.
  \begin{enumerate}[(1)]
  \item
    A \emph{\XGCWA-solution} for $S$ under $M$
    is a ground solution for $S$ under $M$
    that is a union of minimal solutions for $S$ under $M$.
  \item
    We call $\certXGCWA{M}{S}{q} \isdef \cert(q,\CI)$,
    where $\CI$ is the set of the \XGCWA-solutions for $S$ under $M$,
    the \emph{\XGCWA-answers to $q(\tup{x})$ on $M$ and $S$}.
  \end{enumerate}
\end{defi}

\noindent The definition of \XGCWA-solutions and \XGCWA-answers
already seems to be a good approximation to the concept of solutions,
and query answering semantics, respectively,
that we would like to have.
Immediately from the definitions,
we obtain that the \XGCWA-answers are invariant
under logically equivalent schema mappings:

\begin{prop}
  If $M_1 = (\sigma,\tau,\Sigma_1)$ and $M_2 = (\sigma,\tau,\Sigma_2)$
  are logically equivalent schema mappings,
  $S$ is a source instance for $M_1$ and $M_2$, respectively,
  and $q(\tup{x})$ is a query over $\tau$,
  then $\certXGCWA{M_1}{S}{q} = \certXGCWA{M_2}{S}{q}$.
\end{prop}

Furthermore, let us generalize the discussion from the beginning of this section
to argue that \XGCWA-solutions and the \XGCWA-answers as defined above
are suitable for schema mappings defined by st-tgds.
Let $M = \Des$ be such a schema mapping.
Given a source instance $S$ for $M$, let
\begin{align*}
  \Psi_{M,S} \isdef
  \{
    \exists \tup{z}\, \psi(\tup{u},\tup{z}) \mid\
    & \text{there are a tgd $\tgd$ in $\Sigma$ and tuples} \\
    & \text{$\tup{u} \in \Const^{\length{\tup{x}}},
       \tup{v} \in \Const^{\length{\tup{y}}}$
       such that $S \models \phi(\tup{u},\tup{v})$}
  \}.
\end{align*}
For each ground target instance $T$ for $M$,
it holds that $T$ is a solution for $S$ under $M$ if and only if
$T$ satisfies all sentences in $\Psi_{M,S}$.
Let $\CT_0$ be the set of all ground minimal solutions for $S$ under $M$.
Since all sentences in $\Psi_{M,S}$ are monotonic,
$\Psi_{M,S}$ is logically equivalent
(on the set of all ground instances over $\tau$)
to the sentence
\begin{align*}
  \psi_{M,S} \isdef
  \biglor_{T_0 \in \CT_0} \bigland_{R(\tup{t}) \in T_0} R(\tup{t}),
\end{align*}
that is, for all ground instances $T$ over $\tau$,
we have $T \models \psi_{M,S}$ if and only if
$T$ satisfies all sentences in $\Psi_{M,S}$.
Now, $\psi_{M,S}$ tells us that
there is one $T_0 \in \CT_0$ such that all $R(\tup{t}) \in T_0$ are satisfied,
or there are two $T_0 \in \CT_0$ such that all $R(\tup{t}) \in T_0$
are satisfied, and so on.
So, intuitively,
the set of all solutions that are unions of solutions from $\CT_0$
(namely, the set of all $\XGCWA$-solutions for $S$ under $M$)
captures what is expressed by $M$ and $S$
in the sense as explained in the two motivating examples
from the beginning of this section.

\begin{rem}
  The above argumentation can be generalized to more general classes
  of schema mappings.
  For example, let us consider schema mappings
  defined by a certain kind of $\LInf$ sentences.
  Let $M = \Des$ be a schema mapping,
  where $\Sigma$ consists of
  \emph{right-monotonic $\LInf$-st-tgds},
  which are $\LInf$ sentences of the form
  \begin{align*}
    \theta \isdef \forall \tup{x} (\phi(\tup{x}) \limplies \psi(\tup{x})),
  \end{align*}
  where $\phi$ is a $\LInf$ formula over $\sigma$,
  and $\psi$ is a \emph{monotonic} $\LInf$ formula over $\tau$.
  We assume that for each instance $S$ over $\sigma$,
  and for each instance $T$ over $\tau$,
  we have $S \union T \models \theta$ if and only if
  for all $\tup{a} \in (\dom(S) \union \dom(\theta))^{\length{\tup{x}}}$,
  where $\dom(\theta)$ is the set of all constants that occur in $\theta$,
  $S \models \phi(\tup{a})$ implies $T \models \psi(\tup{a})$.
  This can be enforced, for example, by relativizing the universal quantifiers,
  and the quantifiers in $\phi$ to the
  active domain over $\sigma$,
  and by relativizing the quantifiers in $\psi$ to the active domain
  over $\tau$.
  Note that right-monotonic $\LInf$-st-tgds capture st-tgds.

  Now, the above argumentation for schema mappings defined by st-tgds
  goes through for $M$.
  The only difference is that,
  given a source instance $S$ for $M$,
  we let
  \begin{align*}
    \Psi_{M,S} \isdef
    \{
      \psi(\tup{a}) \mid\
      & \text{there are
        $\forall \tup{x} (\phi(\tup{x}) \limplies \psi(\tup{x}))$
        in $\Sigma$ and $\tup{a} \in \Const^{\length{\tup{x}}}$
        with $S \models \phi(\tup{a})$}
      \}.
  \end{align*}
  The remaining part goes through unchanged.
\end{rem}

The following example shows that the \XGCWA-answers
can be appropriate beyond schema mappings
defined by right-monotonic $\LInf$-st-tgds.

\begin{exa}
  Recall the schema mapping $M$, the source instance $S$ for $M$,
  and the query $q$ from Example~\ref{ex:EGCWA-failure}.
  For each ground target instance $T$ for $M$
  that is the union of minimal solutions for $S$ under $M$,
  there exists a nonempty finite set $C \subseteq \Const$ with
  $E^T = \set{(a,b) \mid b \in C}$.
  $T$ is a \XGCWA-solution for $S$ under $M$ if and only if
  $2 \leq \card{C} \leq 3$, as desired.
  Note that the \XGCWA-answers to $q$ on $M$ and $S$ are empty,
  as intuitively expected.
\end{exa}

However, let $M = \Des$ be a schema mapping,
where $\Sigma$ does \emph{not} entirely consist of right-monotonic
$\LInf$-st-tgds,
and let $S$ be a source instance for $M$.
Then the set of the \XGCWA-solutions for $S$ under $M$
may suppress information that should intuitively be taken into account
when answering queries:

\begin{exa}
  \label{ex:XGCWA-answers/motivation}
  Consider the schema mapping $M = (\set{P},\set{E,F},\set{\theta_1,\theta_2})$,
  where
  \begin{align*}
    \theta_1 & \,\isdef\,
      \forall x\, \bigl(P(x) \limplies \exists z \exists z'\, E(z,z')\bigr), \\
    \theta_2 & \,\isdef\,
      \forall x \forall y\, \bigl(E(x,y) \land E(x',y) \limplies F(x,x')\bigr),
  \end{align*}
  and let $S$ be the source instance for $M$ with $P^S = \set{a}$.
  Furthermore, let $c_1,c_2,c$ be constants with $c_1 \neq c_2$.
  Then the target instance
  \[
    T \,\isdef\,
    \set{E(c_1,c),E(c_2,c)} \union \set{F(c_i,c_j) \mid 1 \leq i,j \leq 2}
  \]
  for $M$ is a solution for $S$ under $M$.
  However, it is not a $\XGCWA$-solution for $S$ under $M$,
  since every ground minimal solution for $S$ under $M$
  has the form $\set{E(d,d'),F(d,d)}$ for $d,d' \in \Const$.

  Nevertheless,
  it seems natural to take into account $T$ when answering queries.
  Intuitively,
  under an inclusive interpretation of the existential quantifiers
  in $\theta_1$,
  the st-tgd $\theta_1$ and the atom $P(a)$ in $S$
  tell us that it is possible that both $E(c_1,c)$ and $E(c_2,c)$ hold.
  In combination with $\theta_2$,
  this tells us that it is possible that a solution
  contains $E(c_1,c)$, $E(c_2,c)$ and $F(c_i,c_j)$ for $i,j \in \set{1,2}$.
  Therefore, $T$ should be a possible solution.

  We now extend the set of \XGCWA-solutions
  so that solutions like $T$ are included as well.
  We do this using the following closure operation.
  Let $\CT$ be a nonempty finite set of ground minimal solutions
  for $S$ under $M$,
  say $\CT$ consists of instances $T_i = \set{E(d_i,e_i),F(d_i,d_i)}$
  for $i = 1,\dotsc,n$.
  $\CT$ represents the information that $\bigland_{i=1}^n E(d_i,e_i)$ holds,
  and $S$ and $\theta_1$ intuitively tell us that this is possible.
  In general, $T_0 \isdef \bigunion \CT$ is not a solution for $S$ under $M$,
  since in general it does not satisfy $\theta_2$
  (it does if the constants $e_1,\dotsc,e_n$ are distinct).
  However, we can extend $T_0$ to a solution $T'_0$
  by adding atoms to $T_0$ so that $\theta_2$ is satisfied.
  We pick the minimal set of such atoms,
  namely $\set{F(d_i,d_j) \mid 1 \leq i,j \leq n,\, e_i = e_j}$,
  since we do not want to add any atoms that are not needed
  to satisfy $\theta_2$.
  Note that $T'_0$ is minimal among all ground solutions $T'$ for $S$ under $M$
  with $T' \supseteq T_0$.
  We add $T'_0$ to the set of ``valid solutions''.

  The set of solutions that results from applying the closure operator
  contains all solutions $T$ for $S$ under $M$ of the form
  \[
    T'_0 \,=\,
    \set{E(d_i,e_i) \mid 1 \leq i \leq n}
    \union
    \set{F(d_i,d_j) \mid 1 \leq i,j \leq n,\, e_i = e_j},
  \]
  where $n \geq 1$ and $d_1,\dotsc,d_n,e_1,\dotsc,e_n$ are arbitrary constants.
  Intuitively, this set precisely captures what is expressed by $M$ and $S$.
\end{exa}

In general, we iterate the (appropriately generalized) closure operator
until a fixed point is reached.
We start with the set
\[
  \CT_{M,S}^0 \,\isdef\, \Set{T \mid
    \text{$T$ is a ground minimal solution for $S$ under $M$}}.
\]
For each set $\CI$ of instances, let
\[
  \langle \CI\rangle \,\isdef\,
  \Set{\, \bigunion \CI' \mid
    \text{$\CI'$ is a nonempty finite subset of $\CI$}},
\]
where $\bigunion \CI'$ denotes the union of all instances in $\CI'$.
For every $i \geq 0$, let
\begin{align*}
  \CT_{M,S}^{i+1} \,\isdef\, \CT_{M,S}^i \,\union\,
  \bigl\{
    T_0' \mid\
    & \text{$T_0' \notin \langle \CT_{M,S}^i\rangle$,
      and there is a $T_0 \in \langle \CT_{M,S}^i\rangle$
      such that $T_0'$ is minimal} \\
    & \text{among all ground solutions $T'$ for $S$ under $M$
      with $T_0 \subseteq T'$}
  \bigr\}.
\end{align*}
Intuitively, each instance $T_0' \in \CT_{M,S}^{i+1} \setminus \CT_{M,S}^i$
is a ``minimal consequence'' of some ``fact''
$T_0 \in \langle \CT_{M,S}^i\rangle$ mentioned by $M$ and $S$.
In Example~\ref{ex:XGCWA-answers/motivation},
the instance $T$ belongs to $\CT_{M,S}^1 \setminus \CT_{M,S}^0$.

Note that if $\Sigma$ contains only st-tgds,
or more generally, right-monotonic $\LInf$-st-tgds,
we have $\CT_{M,S}^0 = \CT_{M,S}^i$ for all $i \geq 0$,
and $\langle \CT_{M,S}^0\rangle$ is precisely the set of all
\XGCWA-solutions for $S$ under $M$.
So, for schema mappings defined by st-tgds or right-monotonic $\LInf$-st-tgds,
we have to take into account only the \XGCWA-solutions as defined earlier.
For more general schema mappings,
we take into account all solutions for $S$ under $M$
that are unions of one or more instances in
\[
  \CT_{M,S}^* \,\isdef\, \bigunion_{i \geq 0} \CT_{M,S}^i.
\]

\begin{defi}[\XGCWA-solution, \XGCWA-answers]
  \label{def:XGCWA-solutions}%
  Let $M = \Des$ be a schema mapping, let $S$ be a source instance for $M$,
  and let $q$ be a query over $\tau$.
  \begin{enumerate}[(1)]
  \item
    A \emph{\XGCWA-solution} for $S$ under $M$
    is a ground solution $T$ for $S$ under $M$
    that is the union of one or more instances in $\CT^*_{M,S}$.
  \item
    We call $\certXGCWA{M}{S}{q} \isdef \cert(q,\CI)$,
    where $\CI$ is the set of the \XGCWA-solutions for $S$ under $M$,
    the \emph{\XGCWA-answers to $q$ on $M$ and $S$}.
  \end{enumerate}
\end{defi}

\noindent As before,
immediately from the definitions,
we obtain that the \XGCWA-answers are invariant
under logically equivalent schema mappings:

\begin{prop}
  If $M_1 = (\sigma,\tau,\Sigma_1)$ and $M_2 = (\sigma,\tau,\Sigma_2)$
  are logically equivalent schema mappings,
  $S$ is a source instance for $M_1$ and $M_2$, respectively,
  and $q$ is a query over $\tau$,
  then $\certXGCWA{M_1}{S}{q} = \certXGCWA{M_2}{S}{q}$.
\end{prop}

Furthermore, it is easy to prove:

\begin{prop}
  \label{prop:XGCWA-solutions}
  Let $M = \Des$ be a schema mapping,
  where $\Sigma$ consists of st-tgds and egds,
  and let $S$ be a source instance for $M$.
  Then a target instance $T$ for $M$ is a \XGCWA-solution for $S$ under $M$
  if and only if
  $T$ is the union of one or more ground minimal solutions for $S$ under $M$,
  and $T$ satisfies all egds in $\Sigma$.
\end{prop}

To conclude this section,
we show that, with respect to schema mappings defined by st-tgds and egds,
\XGCWA-solutions can be defined in a way similar to the definition
of GCWA-solutions.
This characterization also shows that,
with respect to schema mappings defined by st-tgds and egds,
\XGCWA-solutions are special GCWA-solutions.

\begin{defi}
  \label{def:D*}
  For every schema mapping $M = \Des$ and every source instance $S$ for $M$,
  define the following set of $\LInf$ sentences over $\sigma \union \tau$:
  \begin{align*}
    \DDB_{M,S}^* \isdef
    \{
      R(\tup{t}) \limplies \phi
      \mid\
      & \text{$R \in \sigma \union \tau$,
        $\tup{t} \in \Const^{\arity(R)}$, and $\phi$ is a monotonic $\LInf$
        sentence} \\
      & \text{over $\sigma \union \tau$ that is satisfied in every minimal
        model $I$ of $\DDB_{M,S}$} \\
      & \text{with $\tup{t} \in R^I$}
    \}.
  \end{align*}
\end{defi}

\begin{prop}
  \label{prop:XGCWA-characterization}
  Let $M = \Des$ be a schema mapping,
  where $\Sigma$ is a set of st-tgds and egds,
  and let $S$ be a source instance for $M$.
  Then for all ground target instances $T$ for $M$,
  the following statements are equivalent:
  \begin{enumerate}[\em(1)]
  \item\label{prop:XGCWA-characterization/XGCWA}
    $T$ is a \XGCWA-solution for $S$ under $M$.
  \item\label{prop:XGCWA-characterization/model}
    $S \union T$ is a model of $\DDB_{M,S} \union \DDB_{M,S}^*$.
  \end{enumerate}
\end{prop}

\begin{proof}
  \fromto{\ref{prop:XGCWA-characterization/XGCWA}}
    {\ref{prop:XGCWA-characterization/model}}
  Suppose that $T$ is a \XGCWA-solution for $S$ under $M$.
  By Proposition~\ref{prop:XGCWA-solutions},
  $T$ is a ground solution for $S$ under $M$,
  and there is a set $\CT_0$ of minimal solutions for $S$ under $M$
  such that
  \(
    T = \bigunion \CT_0.
  \)
  We have to show that
  \(
    I \isdef S \union T
  \)
  satisfies $\DDB_{M,S} \union \DDB_{M,S}^*$.

  Since $T$ is a solution for $S$ under $M$,
  we have $I \models \DDB_{M,S}$.
  Thus, it remains to show that $I \models \DDB_{M,S}^*$.

  To this end, consider an arbitrary sentence
  \(
    \psi \isdef R(\tup{t}) \limplies \phi
  \)
  in $\DDB_{M,S}^*$, and assume that
  \(
    I \models R(\tup{t}).
  \)
  Since $I = S \union T$ and $T = \bigunion \CT_0$,
  there is some $T_0 \in \CT_0$ with $\tup{t} \in R^{S \union T_0}$.
  Note that $I_0 \isdef S \union T_0$ is a minimal model of $\DDB_{M,S}$.
  By Definition~\ref{def:D*}, we thus have $I_0 \models \phi$.
  Since $I_0 \subseteq I$ and $\phi$ is monotonic,
  it follows that $I \models \phi$.
  Consequently, $I$ satisfies $\psi$.

  \fromto{\ref{prop:XGCWA-characterization/model}}
    {\ref{prop:XGCWA-characterization/XGCWA}}
  Suppose that
  \(
    I \isdef S \union T
  \)
  is a model of $\DDB_{M,S} \union \DDB_{M,S}^*$.
  Since models are ground instances by definition,
  it follows that $T$ is a ground solution for $S$ under $M$.
  To show that $T$ is a \XGCWA-solution for $S$ under $M$,
  it remains to construct, by Proposition~\ref{prop:XGCWA-solutions},
  a set $\CT_0$ of minimal solutions
  for $S$ under $M$ such that $T = \bigunion \CT_0$.

  Let $\CT_0$ be the set of all minimal solutions $T_0$
  for $S$ under $M$ with $T_0 \subseteq T$.
  We claim that $T = \bigunion \CT_0$.
  By construction, we have $\bigunion \CT_0 \subseteq T$.
  Thus it remains to show that it is not the case that
  $\bigunion \CT_0 \subsetneq T$.

  Suppose, to the contrary, that $\bigunion \CT_0 \subsetneq T$.
  Then there are $R \in \tau$ and $\tup{t} \in \Const^{\arity(R)}$ such that
  \begin{equation}
    \label{prop:XGCWA-characterization/prop1}
    \tup{t} \in R^T
    \quad\text{and}\quad
    \tup{t} \notin R^{T_0}\ \text{for all $T_0 \in \CT_0$}.
  \end{equation}
  On the other hand, there is at least one minimal model $I_0$
  of $\DDB_{M,S}$ with $\tup{t} \in R^{I_0}$.
  Otherwise, $R(\tup{t}) \limplies \biglor \emptyset$,
  which is equivalent to $\lnot R(\tup{t})$,
  would be in $\DDB_{M,S}^*$, so that $\tup{t} \notin R^I \supseteq R^T$
  would contradict \eqref{prop:XGCWA-characterization/prop1}.
  Let
  \begin{equation*}
    \CI_0 \,\isdef\,
    \set{
      I_0
      \mid
      \text{$I_0$ is a minimal model of $\DDB_{M,S}$
        with $\tup{t} \in R^{I_0}$}
    }.
  \end{equation*}
  Then,
  \begin{align*}
    \psi \isdef R(\tup{t}) \limplies \phi
    \quad \text{with} \quad
    \phi \isdef \biglor_{I_0 \in \CI_0} \bigland_{R'(\tup{t}') \in I_0} R'(\tup{t}')
  \end{align*}
  is satisfied in \emph{every} minimal model of $\DDB_{M,S}$.
  Since $\phi$ is monotonic, we thus have $\psi \in \DDB_{M,S}^*$.
  Furthermore, since $I \models \DDB_{M,S}^*$ and $I \models R(\tup{t})$,
  it follows that $I \models \phi$.
  In particular, there must be some $I_0 \in \CI_0$ such that
  $I \models \bigland_{R'(\tup{t}') \in I_0} R'(\tup{t}')$,
  and thus, $I_0 \subseteq I$.
  Note that $I_0 = S \union T_0$ for some $T_0 \in \CT_0$.
  Together with $\tup{t} \in R^{I_0}$ and $R \in \tau$,
  this implies that $\tup{t} \in R^{T_0}$.
  However, this contradicts \eqref{prop:XGCWA-characterization/prop1}.
  Consequently, $\bigunion \CT_0 = T$.
\end{proof}

Moreover, the following result translates \cite[Theorem~5]{Minker:CADE82}
from GCWA-solutions to \XGCWA-solutions,
and shows that for a given schema mapping $M$ and a source instance $S$ for $M$,
the set $\DDB_{M,S} \union \DDB_{M,S}^*$
is maximally consistent in the sense that the addition of any sentence $\psi$
of the form $R(\tup{t}) \limplies \phi$,
where $\phi$ is a monotonic $\LInf$ sentence and
$\DDB_{M,S} \union \DDB_{M,S}^* \not\models \psi$,
leads to a set of formulas that is inconsistent with
$\DDB_{M,S} \union \DDB_{M,S}^*$.

\begin{prop}
  \label{prop:maximal-consistency}
  Let $M = \Des$ be a schema mapping,
  let $S$ be a nonempty source instance for $M$,
  let $\DDB \isdef \DDB_{M,S}$ and let $\DDB' \isdef \DDB \union \DDB^*$.
  \begin{enumerate}[\em(1)]
  \item\label{prop:maximal-consistency/consistency}
    For all monotonic $\LInf$-sentences $\phi$ over $\sigma \union \tau$,
    we have $\DDB \models \phi$ if and only if $\DDB' \models \phi$.
  \item\label{prop:maximal-consistency/maximality}
    For all $\psi \isdef R(\tup{t}) \limplies \phi$,
    where $R(\tup{t})$ is a ground atom over $\sigma \union \tau$,
    $\phi$ is a $\LInf$ sentence over $\sigma \union \tau$,
    and $\DDB' \not\models \psi$:
    \begin{enumerate}[\em(a)]
    \item
      $\DDB' \union \set{\psi}$ has no model, or
    \item
      there is a monotonic $\LInf$-sentence $\chi$ over $\sigma \union \tau$
      such that $\DDB' \union \set{\psi} \models \chi$,
      but $\DDB' \not\models \chi$.
    \end{enumerate}
  \end{enumerate}
\end{prop}

\begin{proof}
  Statement \ref{prop:maximal-consistency/consistency} is obvious,
  so in the following we prove \ref{prop:maximal-consistency/maximality}.
  Let $\psi$ be given.
  If $\DDB' \union \set{\psi}$ has no model, then we are done.
  So assume that $\DDB' \union \set{\psi}$ has a model.
  Let $\CI_0$ be the set of all minimal models
  of $\DDB' \union \set{\psi}$, and consider the monotonic $\LInf$ sentence
  \begin{align*}
    \chi \isdef
    \biglor_{I_0 \in \CI_0} \bigland_{R'(\tup{t}') \in I_0} R'(\tup{t}').
  \end{align*}
  Clearly, we have $\DDB' \union \set{\psi} \models \chi$.
  Indeed, if $I$ is a model of $\DDB' \union \set{\psi}$,
  let $I_0 \in \CI_0$ be such that $I_0 \subseteq I$.
  Then, $I_0 \models \bigland_{R'(\tup{t}') \in I_0} R'(\tup{t}')$,
  and therefore, $I_0 \models \chi$.
  Since $\chi$ is monotonic and $I_0 \subseteq I$,
  this leads to $I \models \chi$.

  Furthermore, we have $\DDB' \not\models \chi$.
  For a contradiction suppose that $\DDB' \models \chi$.
  Note that there is a minimal model $I_0$ of $\DDB$
  with $I_0 \not\models \psi$.
  (This follows immediately
   from $\DDB^* \subseteq \DDB'$ and $\DDB' \not\models \psi$,
   which imply $\psi \notin \DDB^*$.)
  Since $I_0$ is a minimal model of $\DDB'$ as well,
  and $\DDB' \models \chi$,
  we have $I_0 \models \chi$.
  Thus, there is some $I_0' \in \CI_0$ such that
  $I_0 \models \bigland_{R'(\tup{t}') \in I_0'} R'(\tup{t}')$.
  In other words, $I_0' \subseteq I_0$, which implies $I_0' = I_0$,
  because $I_0$ is a minimal model of $\DDB'$,
  and $I_0' \models \DDB'$.
  But this is impossible, since $I_0' \models \psi$ and
  $I_0 \not\models \psi$.
  Consequently, $\DDB' \not\models \chi$.
\end{proof}

\section{Data Complexity of Query Evaluation
  under the \texorpdfstring{\XGCWA}{GCWA*}-Semantics}
\label{sec:complexity}

In this section, we study the data complexity of computing \XGCWA-answers,
where data complexity means that the schema mapping and the query
to be evaluated are fixed (i.e., not part of the input).
We concentrate on schema mappings defined by st-tgds only.

Since in data exchange,
the goal is to answer queries based on some materialized solution,
given a schema mapping $M = \Des$ and a query language $L$,
we are particularly interested in whether there are algorithms
$\mathbb{A}_1,\mathbb{A}_2$ with the following properties:
\begin{enumerate}[(1)]
\item\label{alg1}
  $\mathbb{A}_1$ takes a source instance $S$ for $M$ as input
  and computes a solution $T$ for $S$ under $M$, and
\item\label{alg2}
  $\mathbb{A}_2$ takes a solution $T$ computed by $\mathbb{A}_1$
  and a query $q \in L$ over $\tau$ as input
  and computes $\certXGCWA{M}{S}{q}$.
\end{enumerate}
In particular, $\mathbb{A}_1$ takes care of the actual data exchange
(dependent on the query language, but independent of any concrete query),
while $\mathbb{A}_2$ answers queries based on some materialized solution.
At best,
both $\mathbb{A}_1$, and $\mathbb{A}_2$ for fixed $q \in L$,
run in polynomial time.

For proving complexity lower bounds, we consider,
for fixed schema mappings $M = \Des$ and queries $q(\tup{x})$ over $\tau$,
the decision problem
\bigskip
\problem{$\Eval(M,q)$}
  {a source instance $S$ for $M$,
   and a tuple $\tup{t} \in \Const^{\length{\tup{x}}}$}
  {Is $\tup{t} \in \certXGCWA{M}{S}{q}$?}
\bigskip
The complexity of this problem can be seen as a lower bound
on the joint complexity of finding a solution $T$ as in step~\ref{alg1} above,
and obtaining $\certXGCWA{M}{S}{q}$ from $T$ as in step~\ref{alg2}.
If, for example, $\Eval(M,q)$ is $\co\NP$-complete,
then finding $T$ is intractable, or computing $\certXGCWA{M}{S}{q}$
from $T$ is intractable.

We first consider the complexity of computing the $\XGCWA$-answers
to monotonic queries and existential queries
in Sections~\ref{sec:complexity/monotonic} and \ref{sec:complexity/existential},
and deal with the present section's main result concerning universal queries
in Section~\ref{sec:complexity/universal}.

\subsection{Monotonic Queries}
\label{sec:complexity/monotonic}

For monotonic queries, all results obtained for the certain answers semantics
(see, e.g.,
 \cite{FKMP:TCS336-1,Madry:IPL94,ABFL:PODS04,Kolaitis:PODS05,Libkin:PODS06,
       DNR:PODS08,ABR:ICDT09,Barcelo:SR38-1})
carry over to the \XGCWA-answers semantics:

\begin{prop}
  \label{prop:certXGCWA-for-monotonic-queries}
  Let $M = \Des$ be a schema mapping, let $S$ be a source instance for $M$,
  and let $q(\tup{x})$ be a monotonic query over $\tau$.
  Then, $\certXGCWA{M}{S}{q} = \certOWA{M}{S}{q}$.
\end{prop}

\begin{proof}
  Since every \XGCWA-solution for $S$ under $M$
  is a solution for $S$ under $M$,
  we have $\certOWA{M}{S}{q} \subseteq \certXGCWA{M}{S}{q}$.
  To show $\certXGCWA{M}{S}{q} \subseteq \certOWA{M}{S}{q}$,
  consider a tuple $\tup{t} \in \certXGCWA{M}{S}{q}$.
  We have to show that $\tup{t} \in q(T)$ for all solutions $T$ for $S$
  under $M$.
  To this end, it suffices to show that $\tup{t} \in q(T)$
  for all \emph{ground} solutions $T$ for $S$ under $M$,
  since nulls can be seen as special constants.
  Let $T$ be a ground solution for $S$ under $M$,
  and let $T_0$ be a minimal solution for $S$ under $M$ with $T_0 \subseteq T$.
  By Definition~\ref{def:XGCWA-solutions},
  $T_0$ is a \XGCWA-solution for $S$ under $M$,
  and since $\tup{t} \in \certXGCWA{M}{S}{q}$, we have $\tup{t} \in q(T_0)$.
  Since $q$ is monotonic and $T_0 \subseteq T$,
  we conclude $\tup{t} \in q(T)$.
\end{proof}

In particular,
if $M$ is a schema mapping defined by st-tgds,
and $q(\tup{x})$ is a union of conjunctive queries over $M$'s target schema,
then Proposition~\ref{prop:OWA-UCQ} implies
that there is a polynomial time algorithm
that takes a universal solution for a source instance $S$ for $M$ as input
and outputs the \XGCWA-answers to $q(\tup{x})$ on $M$ and $S$.
Note that by Theorem~\ref{thm:core-algorithm},
a universal solution can be computed in polynomial time
(for fixed $M$) from a given source instance for $M$.

\subsection{Existential Queries and Beyond}
\label{sec:complexity/existential}

We now turn to \emph{existential queries},
which are FO queries of the form
$q(\tup{x}) = \exists \tup{y}\, \phi(\tup{x},\tup{y})$,
where $\phi$ is quantifier-free.
A particular class of existential queries
are \emph{conjunctive queries with negation} ($\CQneg$ queries, for short),
which are queries of the form
\(
  q(\tup{x}) = \exists \tup{y}\, (L_1 \land \dotsb \land L_k)
\)
where each $L_i$ is either a relational atom $R(\tup{u})$
or the negation of a relational atom.
A simple reduction from the $\CLIQUE$ problem \cite{GJ79}
shows that $\Eval(M,q)$ can be $\co\NP$-hard for schema mappings $M$
defined by LAV tgds and $\CQneg$ queries with only one negated atom:

\begin{prop}
  \label{prop:CQ-with-neg-coNP-complete}
  There exists a schema mapping $M = \Des$,
  where $\Sigma$ consists of two LAV tgds,
  and a Boolean $\CQneg$ query $q$ over $\tau$ with one negated atomic formula
  such that $\Eval(M,q)$ is $\co\NP$-complete.
\end{prop}

\begin{proof}
  Let $M = \Des$, where $\sigma$ consists of binary relation symbols $E_0,C_0$,
  $\tau$ consists of binary relation symbols $E,C,A$,
  and $\Sigma$ consists of the following st-tgds:
  \begin{align*}
    \theta_1 &  \,\isdef\, \forall x \forall y\, \bigl(
        E_0(x,y) \limplies E(x,y)
      \bigr), \\
    \theta_2 &  \,\isdef\, \forall x \forall y\, \bigl(
        C_0(x,y) \limplies \exists z_1 \exists z_2\, (
          C(x,y) \land A(x,z_1) \land A(y,z_2)
        )
      \bigr).
  \end{align*}
  Furthermore, let
  \begin{align*}
    q \,\isdef\,
    \exists x \exists y \exists z_1 \exists z_2\, \bigl(
      C(x,y) \land A(x,z_1) \land A(y,z_2) \land \lnot E(z_1,z_2)
    \bigr).
  \end{align*}
  We show that $\Eval(M,q)$ is $\co\NP$-complete
  by showing that the complement of $\Eval(M,q)$ is $\NP$-complete.

  \proofpart{Membership}
  The complement of $\Eval(M,q)$ is solved by
  a nondeterministic Turing machine as follows.
  Given a source instance $S$ for $M$,
  the machine needs to decide whether $\certXGCWA{M}{S}{q} = \emptyset$,
  that is, whether there is a \XGCWA-solution $T$ for $S$ under $M$
  such that $T \models \lnot q$.

  Note that
  \begin{align}
    \label{eq:CQ-with-neg-coNP-complete/negation}
    \lnot q \,\equiv\,
    \forall x \forall y \forall z_1 \forall z_2\, \bigl(
      C(x,y) \land A(x,z_1) \land A(y,z_2) \limplies E(z_1,z_2)
    \bigr),
  \end{align}
  and that the minimal ground solutions for $S$ under $M$
  are all solutions of the form
  \[
    T_f\, \isdef\,
    \Set{E(a,b) \mid (a,b) \in E_0^S}
    \union
    \Set{C(c,c') \mid (c,c') \in C_0^S}
    \union
    \Set{A(c,f(c)) \mid c \in \dom(C_0^S)},
  \]
  for some mapping $f\colon \dom(C_0^S) \to \Const$.
  Hence, if $T$ is a \XGCWA-solution for $S$ under $M$ with $T \models \lnot q$,
  there is a minimal ground solution $T' \subseteq T$ for $S$ under $M$
  with $T' \models \lnot q$.
  In particular,
  it suffices to decide whether there is a minimal ground solution $T$
  for $S$ under $M$ such that $T \models \lnot q$.

  By \eqref{eq:CQ-with-neg-coNP-complete/negation},
  if $T_f \models \lnot q$ for some $f\colon \dom(C_0^S) \to \Const$,
  then for all $c \in \dom(C_0^S)$ we have $f(c) \in \dom(E_0^S)$.
  Hence, in order to check whether there is a minimal ground solution $T$
  for $S$ under $M$ such that $T \models \lnot q$,
  it suffices to guess a mapping $f\colon \dom(C_0^S) \to \dom(S)$,
  and to check whether $T_f \models \lnot q$.
  Clearly, this can be done by a nondeterministic Turing machine
  in time polynomial in the size of $S$.

  \proofpart{Hardness}
  To show that the complement of $\Eval(M,q)$ is $\NP$-hard,
  we present a reduction from the $\NP$-complete $\CLIQUE$ problem \cite{GJ79}.
  The $\CLIQUE$ problem is to decide,
  given an undirected graph $G = (V,E)$ without loops and
  a positive integer $k$,
  whether $G$ contains a \emph{clique} of size $k$.
  Here, a clique in $G$ is a set $C \subseteq V$ such that
  for all $u,v \in C$ with $u \neq v$ we have $\set{u,v} \in E$.

  Let $G = (V,E)$ be an undirected graph without loops,
  and let $k \geq 1$ be an integer.
  If $k = 1$, then $G$ has a clique of size $k$ if and only if $V$ is nonempty,
  and we can reduce $(G,k)$ to some predefined fixed source instance $S$
  for $M$ such that $\certXGCWA{M}{S}{q} = \emptyset$ if and only if
  $V$ is nonempty
  (i.e., a source instance $S$ with $C_0^S = \emptyset$ if $V$ is nonempty,
   and a source instance $S$ with $C_0^S = \set{(c,c)}$ for some $c \in \Const$
   if $V$ is empty).

  If $k \geq 2$, we reduce $(G,k)$ to the source instance $S$ for $M$ with $E_0^S = E$
  and $C_0^S = \set{(c_i,c_j) \mid 1 \leq i,j \leq k,\, i \neq j}$,
  where $c_1,\dotsc,c_k$ are pairwise distinct constants
  that do not occur in $V$.
  We claim that $G$ has a clique of size $k$ if and only if
  $\certXGCWA{M}{S}{q} = \emptyset$.

  \onlyifdir\
  Let $C = \set{v_1,\dotsc,v_k}$ be a clique of size $k$ in $G$,
  and let $T$ be the target instance for $M$ with $E^T = E$, $C^T = C_0^S$
  and $A^T = \set{(c_i,v_i) \mid 1 \leq i \leq k}$.
  Then $T$ is a minimal solution for $S$ under $M$,
  and, by Definition~\ref{def:XGCWA-solutions},
  a \XGCWA-solution for $S$ under $M$.
  Furthermore, we have $T \not\models q$.
  To see this, note that for all $u,v,w_1,w_2 \in \dom(T)$ with
  $T \models C(u,v) \land A(u,w_1) \land A(v,w_2)$,
  there are distinct $i,j \in \set{1,\dotsc,k}$ with $u = c_i$ and $v = c_j$,
  so that $w_1 = v_i$ and $w_2 = v_j$.
  Since $v_i,v_j \in C$ and $E^T = E$, we thus have $T \models E(w_1,w_2)$
  for all such $u,v,w_1,w_2$.
  Since $T$ is a \XGCWA-solution for $S$ under $M$, and $T \not\models q$,
  we have $\certXGCWA{M}{S}{q} = \emptyset$.

  \ifdir\
  Suppose that $\certXGCWA{M}{S}{q} = \emptyset$.
  Then there is a \XGCWA-solution $T$ for $S$ under $M$
  with $T \not\models q$.
  For all $i \in \set{1,\dotsc,k}$, let
  \begin{align*}
    V_i \,\isdef\, \Set{v \in \dom(T) \mid (c_i,v) \in A^T}.
  \end{align*}
  Since $S \union T \models \theta_2$,
  each $V_i$ is nonempty.
  Thus, there is a set $C = \set{v_1,\dotsc,v_k}$ such that $v_i \in V_i$
  for each $i \in \set{1,\dotsc,k}$.
  Moreover, for all $i,j \in \set{1,\dotsc,k}$ with $i \neq j$,
  we have $(v_i,v_j) \in E$.
  To see this,
  observe that $T \models C(c_i,c_j) \land A(c_i,v_i) \land A(c_j,v_j)$,
  so that $T \not\models q$ implies $T \models E(v_i,v_j)$.
  It follows that $C$ is a clique in $G$ of size $k$
  (since $k \geq 2$ and $G$ has no loops).
\end{proof}

Adding only one universal quantifier can make the problem undecidable.
Specifically, let us consider $\exists^* \forall$ FO queries,
which are FO queries of the form
\(
  \exists x_1 \dotsb \exists x_k \forall y\, \phi,
\)
where $\phi$ is quantifier-free.
Then we have:

\begin{prop}
  \label{prop:EX*ALL-undecidable}
  There exists a schema mapping $M = \Des$,
  where $\Sigma$ consists of two LAV tgds,
  and a Boolean $\exists^* \forall$ FO query $q$ over $\tau$
  such that $\Eval(M,q)$ is undecidable.
\end{prop}

\begin{proof}
  Let $M = (\set{R},\set{R_p,R_f},\Sigma)$,
  where $R,R_p,R_f$ are ternary relation symbols,
  and $\Sigma$ consists of the st-tgds
  $\theta_c \isdef \forall \tup{x} (R(\tup{x}) \limplies R_p(\tup{x}))$ and
  \(
    \theta_d \isdef
    \forall \tup{x} (R(\tup{x}) \limplies \exists \tup{y}\, R_f(\tup{y})).
  \)
  Let $\hat{q}$ be a FO query that is true in a target instance $T$ for $M$
  precisely if $R_p^T \subseteq R_f^T$,
  and $R_f^T$ encodes the graph of a total associative function
  $f\colon B \times B \to B$ for some set $B$:
  \begin{align*}
    \hat{q}
    \ \, \isdef\, \ \ & 
        \forall \tup{x}\, \bigl(R_p(\tup{x}) \limplies R_f(\tup{x})\bigr)
    \\\allowdisplaybreaks
    \land\ & 
        \forall \tup{x}\, \forall y_1 \forall y_2 \bigl(
          R_f(\tup{x},y_1) \land R_f(\tup{x},y_2) \limplies y_1 = y_2
        \bigr)
          \\
    \allowdisplaybreaks
    \land\ & 
        \forall x \forall y\, \bigl(
          \phi_{\dom}(x) \land \phi_{\dom}(y) \limplies \exists z\, R_f(x,y,z)
        \bigr)
        \\
    \allowdisplaybreaks
    \land\ & 
      \forall x \forall y \forall z \forall u \forall v \forall w\, \bigl(
          R_f(x,y,u) \land R_f(u,z,v) \land R_f(y,z,w) \limplies R_f(x,w,v)
      \bigr)
        ,
  \end{align*}
  where, in a target instance $T$ for $M$,
  \begin{align*}
    \phi_{\dom}(x) \,\isdef\, \exists x_1 \exists x_2 \exists x_3\, \left(
      R_f(x_1,x_2,x_3) \land \biglor_{i=1}^3 x = x_i
    \right)
  \end{align*}
  defines the set of all values that occur in $R_f^T$.
  Note that the last three lines in the definition of $\hat{q}$
  are essentially the target constraints of the schema mapping
  in \cite[Theorem~3.6]{KPT:PODS06}.
  Note also that the negation of $\hat{q}$ is equivalent to a Boolean
  $\exists^* \forall$ FO query.
  Let $q$ be this $\exists^* \forall$ FO query.

  To show that $\Eval(M,q)$ is undecidable, we reduce the
  \problem{\ProblemName{Embedding problem for finite semigroups}}
    {a partial function $p\colon A^2 \to A$, where $A$ is a finite set}
    {Is there a finite set $B \supseteq A$ and a total function
     $f\colon B^2 \to B$
     such that $f$ is associative, and $f$ extends $p$
     (i.e., $p(x,y)$ defined implies $f(x,y) = p(x,y)$)?}
  that is known to be undecidable \cite{KPT:PODS06}, to $\Eval(M,q)$.
  Let $p\colon A^2 \to A$ be a partial function, where $A$ is a finite set.
  Construct the source instance $S$ for $M$,
  where $R^S$ is the graph of $p$,
  that is, $R^S = \set{(a,b,c) \mid p(a,b) = c}$.
  We claim that $\certXGCWA{M}{S}{q} = \emptyset$ if and only if
  $p$ is a ``yes''-instance of the embedding problem for finite semigroups.

  Note that the \XGCWA-solutions for $S$ under $M$
  are all the target instances $T$ for $M$ such that $R_p^T = R^S$,
  and either
  (1) $R^S = R_p^T = R_f^T = \emptyset$, or
  (2) $R_f^T$ is a nonempty finite subset of $\Const^3$.
  Therefore, $\certXGCWA{M}{S}{q} = \emptyset$ if and only if
  there is a \XGCWA-solution $T$ for $S$ under $M$ such that
  $R_f^T$ is the graph of a total function $f\colon \dom(T)^2 \to \dom(T)$
  that is associative and extends $p$.
  This is the case precisely if
  $p$ is a ``yes''-instance of the embedding problem for finite semigroups.
\end{proof}

\subsection{Universal Queries}
\label{sec:complexity/universal}

As we have seen in Section~\ref{sec:complexity/existential},
computing \XGCWA-answers to existential queries may be a difficult task,
and even more difficult (if possible at all)
if the query additionally contains universal quantifiers.

We now turn to \emph{universal queries},
which are FO queries of the form
$q(\tup{x}) = \forall \tup{y}\, \phi(\tup{x},\tup{y})$,
where $\phi$ is quantifier-free.
As a general upper bound for such queries with respect to schema mappings
defined by st-tgds we obtain:

\begin{prop}
  \label{prop:universal-queries-and-st-tgds}
  Let $M = \Des$ be a schema mapping, where $\Sigma$ consists of st-tgds,
  and let $q(\tup{x})$ be a universal query over $\tau$.
  Then, $\Eval(M,q)$ is in $\co\NP$.
\end{prop}

The proof of Proposition~\ref{prop:universal-queries-and-st-tgds}
uses basic ideas from the proof of this section's main result,
Theorem~\ref{thm:universal-queries},
and is deferred to Section~\ref{sec:XGCWA-answers/universal/coNP}.

In what follows,
we prove that for schema mappings defined by st-tgds
which are \emph{packed} as defined below,
the \XGCWA-answers to universal queries can even be computed in polynomial time.

\begin{defi}[packed st-tgd]
  \label{def:packed-tgd}
  An st-tgd $\tgd$ is \emph{packed}
  if for all distinct atoms $R_1(\tup{u}_1),R_2(\tup{u}_2)$ in $\psi$,
  there is a variable in $\tup{z}$ that occurs both in $\tup{u}_1$
  and in $\tup{u}_2$.
\end{defi}

Notice that the schema mapping defined in the proof
of Proposition~\ref{prop:EX*ALL-undecidable} is defined by packed st-tgds.

Although schema mappings defined by packed st-tgds
are not as expressive as schema mappings defined by st-tgds,
they seem to form an interesting class of schema mappings.
Packed st-tgds still allow for non-trivial use of existential quantifiers
in the heads of st-tgds.
For example, consider a schema mapping $M$ defined by st-tgds $\tgd$,
where $\psi$ contains at most two atoms that contain variables from $\tup{z}$.
Then $M$ is logically equivalent to a schema mapping defined by packed st-tgds.
To see this, let
\[
  \theta \,\isdef\, \tgd
\]
be an st-tgd in $M$,
and let $G$ be the graph whose vertices are the atoms in $\psi$,
and which has an edge between two distinct atoms
if they share a variable from $\tup{z}$.
Let $C_1,\dotsc,C_k$ be the connected components of $G$,
and for every $i \in \set{1,\dotsc,k}$ let
\[
  \theta_i \,\isdef\,
  \forall \tup{x} \forall \tup{y}\, \bigl(
    \phi(\tup{x},\tup{y})
    \limplies
    \exists \tup{z}\, \psi_i
  \bigr),
\]
where $\psi_i$ is the conjunction of all atoms in $C_i$.
Then $\theta$ is logically equivalent to $\set{\theta_1,\dotsc,\theta_k}$.
Using that $\psi$ contains at most two atoms with variables from $\tup{z}$,
it is easy to see that each $\theta_i$ is a packed st-tgd.
As a special case, it follows that each full st-tgd is equivalent
to a set of packed st-tgds.
An example of a st-tgd that is \emph{not} packed is
\(
  \forall x (P(x) \limplies \exists z_1 \exists z_2 \exists z_3 (
    E(x,z_1) \land E(z_1,z_2) \land E(z_2,z_3)
  )).
\)

We are now ready to state this section's main result:

\begin{thm}
  \label{thm:universal-queries}
  Let $M = \Des$ be a schema mapping, where $\Sigma$ consists of packed st-tgds,
  and let $q(\tup{x})$ be a universal query over $\tau$.
  Then there is a polynomial time algorithm that,
  given $\core(M,S)$ for some source instance $S$ for $M$,
  outputs $\certXGCWA{M}{S}{q}$.
\end{thm}

Note that Theorem~\ref{thm:universal-queries}
and Theorem~\ref{thm:core-algorithm}
immediately imply that for every schema mapping $M$ specified by packed st-tgds,
and for every universal query $q(\tup{x})$ over $M$'s target schema,
there is a polynomial time algorithm that takes a source instance $S$ for $M$
as input,
and outputs $\certXGCWA{M}{S}{q}$.
In particular:

\begin{cor}
  If $M$ is a schema mapping defined by packed st-tgds,
  and $q(\tup{x})$ is a universal query over $M$'s target schema,
  then $\Eval(M,q)$ is in $\PTIME$.
\end{cor}

An interesting consequence of Theorem~\ref{thm:universal-queries}
is the following.
Let $M$ be a schema mapping defined by packed st-tgds,
and let $S$ be a source instance for $M$.
Recall from Section~\ref{sec:basics}
that the OWA-answers to unions of conjunctive queries on $M$ and $S$
can be computed in polynomial time from $\core(M,S)$
(assuming $M$ and the query are fixed).
In other words,
we only need to compute $\core(M,S)$
in order to answer both unions of conjunctive queries,
and universal queries.
As mentioned above,
$\core(M,S)$ can be computed in polynomial time if $M$ is fixed.

Let us now turn to the proof of Theorem~\ref{thm:universal-queries}.
Observe that Theorem~\ref{thm:universal-queries}
is an immediate consequence of:

\begin{thm}
  \label{thm:universal-queries/bool}
  Let $M = \Des$ be a schema mapping,
  where $\Sigma$ consists of packed st-tgds,
  and let $q(\tup{x})$ be a universal query over $\tau$.
  Then there is a polynomial time algorithm that,
  given $\core(M,S)$ for some source instance $S$ for $M$,
  and a tuple $\tup{t} \in \Const^{\length{\tup{x}}}$,
  decides whether $\tup{t} \in \certXGCWA{M}{S}{q}$.
\end{thm}

The remaining part of this section is devoted to the proof
of Theorem~\ref{thm:universal-queries/bool}.

\subsubsection{\texorpdfstring{\XGCWA}{GCWA*}-Answers and the Core}
\label{sec:XGCWA-answers/universal/core}

Let us first see how we can decide membership of tuples
in $\certXGCWA{M}{S}{q}$ using $\core(M,S)$.
Consider a schema mapping $M = \Des$, where $\Sigma$ is a set of packed st-tgds,
and let $q(\tup{x})$ be a universal query over $\tau$.
Given $\core(M,S)$ and an $\length{\tup{x}}$-tuple $\tup{t}$ of constants,
how can we decide whether $\tup{t} \in \certXGCWA{M}{S}{q}$?

First observe that if $\tup{t}$ is not a tuple
over $\const(\core(M,S)) \union \dom(q)$,
then by the definition of $\certXGCWA{M}{S}{q}$
we have $\tup{t} \notin \certXGCWA{M}{S}{q}$.
Therefore, in the following we assume that $\tup{t}$ is a tuple
over $\const(\core(M,S)) \union \dom(q)$.
In this case, we have $\tup{t} \notin \certXGCWA{M}{S}{q}$ if and only if
there is a \XGCWA-solution $\tilde{T}$ for $S$ under $M$
such that $\tilde{T} \models \lnot q(\tup{t})$.
By the definition of \XGCWA-solution
and the fact that $\Sigma$ consists of st-tgds,
the latter is the case precisely
if there is a nonempty finite set $\CT$ of ground minimal solutions
for $S$ under $M$
with $\bigunion \CT \models \lnot q(\tup{t})$.
Using the following lemma, we can reformulate the last condition
in terms of $\core(M,S)$.

Recall the definition of a valuation of an instance $T$,
and the definition of $\rep(T)$ from Section~\ref{sec:problems}.
Then:

\begin{lem}
  \label{lemma:minimal-from-core}
  Let $M = \Des$ be a schema mapping,
  where $\Sigma$ consists of st-tgds,
  and let $S$ be a source instance for $M$.
  Then the set of all ground minimal solutions for $S$ under $M$
  is precisely the set of all minimal instances
  in $\rep(\core(M,S))$.
\end{lem}

\begin{proof}
  Let $T \isdef \core(M,S)$.
  We first show that every instance in $\rep(T)$
  is a ground solution for $S$ under $M$.
  Let $\hat{T}$ be an instance in $\rep(T)$.
  Then there is a valuation $v$ of $T$ with $v(T) = \hat{T}$.
  This shows that $\hat{T}$ is ground.
  To see that $\hat{T}$ satisfies all st-tgds in $\Sigma$,
  let
  \(
    \theta \isdef
    \forall \tup{x} \forall \tup{y} (
      \phi(\tup{x},\tup{y}) \limplies \exists \tup{z} \psi(\tup{x},\tup{z})
    )
  \)
  be a st-tgd in $\Sigma$, and let $\tup{a},\tup{b}$ be tuples
  with $S \models \phi(\tup{a},\tup{b})$.
  Since $S \union T \models \theta$,
  there is a tuple $\tup{t}$ with $T \models \psi(\tup{a},\tup{t})$,
  and thus $\hat{T} \models \psi(\tup{a},v(\tup{t}))$.
  Altogether, $\hat{T}$ is a ground solution for $S$ under $M$.

  It remains to show that every ground minimal solution for $S$ under $M$
  is in $\rep(T)$.
  Let $T_0$ be a ground minimal solution for $S$ under $M$.
  It is not hard to verify that there is a valuation $v_0$
  of $T^* \isdef \cansol(M,S)$ with $v_0(T^*) = T_0$
  (see also \cite{Libkin:PODS06}).
  Since $T^*$ is a universal solution for $S$ under $M$,
  we have $T = \core(T^*)$,
  and thus $\iota(T) \subseteq T^*$ for some injective mapping
  $\iota\colon \dom(T) \to \dom(T)$ that is legal for $T$.
  Let $v \isdef \compose{\iota}{v_0}$.
  Then,
  \begin{align}
    \label{eq:minimal-from-core}
    v(T) \,=\, v_0(\iota(T)) \,\subseteq\, v_0(T^*) \,=\, T_0.
  \end{align}
  Note that $v(T) \in \rep(T)$.
  Therefore, as shown above, $v(T)$ is a solution for $S$ under $M$.
  Since $T_0$ is a minimal solution for $S$ under $M$,
  \eqref{eq:minimal-from-core} implies $v(T) = T_0$.
  Thus, $v$ is a valuation of $T$ with $v(T) = T_0$,
  which proves that $T_0 \in \rep(T)$.
\end{proof}

Given $\core(M,S)$ and $\tup{t}$,
it remains to decide whether there is a nonempty finite set $\CT$
of minimal instances in $\rep(\core(M,S))$
such that $\bigunion \CT \models \lnot q(\tup{t})$.
Note that, since $q$ is a universal query,
$\lnot q$ is logically equivalent to a query
of the form $\exists \tup{y}\, \phi(\tup{x},\tup{y})$.
Before we consider the general case
(where $\phi$ is an arbitrary quantifier-free query)
in Section~\ref{sec:XGCWA-answers/universal/proof},
the following section deals with the case
that $\tup{y}$ contains no variable and $\phi$ consists of a single atom
$R(\tup{u})$, where $\tup{u}$ is a tuple of constants.
In this case, the problem simplifies to:
Is there a minimal instance in $\rep(\core(M,S))$
that contains $R(\tup{u})$?

\subsubsection{Finding Atoms in Minimal Instances}
\label{sec:XGCWA-answers/universal/atoms}

Let $M = \Des$ be a schema mapping, where $\Sigma$ consists of packed st-tgds,
let $T \isdef \core(M,S)$ for some source instance $S$ for $M$,
and let $R(\tup{t})$ be an atom over $\tau$.
In the following,
we consider the problem of testing whether there is a minimal instance
$T_0$ in $\rep(T)$ with $R(\tup{t}) \in T_0$.
We will often state results in a more general form than necessary,
so that we can apply those results later in the more general setting
considered in Section~\ref{sec:XGCWA-answers/universal/proof}.

First note that there may be infinitely many minimal instances in $\rep(T)$,
so that it is impossible to check out all these instances.
However, it suffices to consider representatives of the minimal instances
in $\rep(T)$,
where constants that do not occur in $T$ or $R(\tup{t})$
are represented by nulls in $T$.
Denoting by $C$ the set of all constants in $\tup{t}$,
the set $\minrep_C(T)$ of all such representatives is formally defined
as follows:

\begin{defi}[$\val_C(T)$, $\minrep_C(T)$]
  \label{def:minrep}
  Let $T$ be an instance, and let $C \subseteq \Const$.
  \begin{enumerate}[(1)]
  \item
    We write $\val_C(T)$ for the set of all mappings
    $f\colon \dom(T) \to \dom(T) \union C$ that are legal for $T$.
  \item
    Let $\minrep_C(T)$ be the set of all instances $\hat{T}$
    for which there is some $f \in \val_C(T)$ with $\hat{T} = f(T)$,
    and there is no $f' \in \val_C(T)$ with $f'(T) \subsetneq \hat{T}$.
  \end{enumerate}
\end{defi}

Throughout this section,
$C$ will usually be the set of constants in $\tup{t}$.

\begin{prop}
  \label{prop:minrep-props}
  Let $T$ be an instance, and let $C \subseteq \Const$.
  \begin{enumerate}[\em(1)]
  \item\label{prop:minrep-props/capture}
    For each $T_0 \in \rep(T)$, the following are equivalent:
    \begin{enumerate}
    \item\label{prop:minrep-props/capture/1}
      $T_0$ is a minimal instance in $\rep(T)$.
    \item\label{prop:minrep-props/capture/3}
      There is an instance $T_0' \in \minrep_C(T)$
      and an injective valuation $v$ of $T'_0$
      such that $v(T_0') = T_0$,
      and $v^{-1}(c) = c$ for all $c \in \dom(T_0) \intersection C$.
    \end{enumerate}
  \item\label{prop:minrep-props/core-is-minimal}
    If $T$ is a core, then $T \in \minrep_C(T)$.
  \item\label{prop:minrep-props/core}
    Each instance in $\minrep_C(T)$ is a core.
  \end{enumerate}
\end{prop}

\begin{proof}
  \ad{\ref{prop:minrep-props/capture}}
  We first prove that \ref{prop:minrep-props/capture/1} implies
  \ref{prop:minrep-props/capture/3}.
  Suppose that $T_0$ is a minimal instance in $\rep(T)$,
  and let $v_0$ be a valuation of $T$ with $v_0(T) = T_0$.
  Furthermore, let $\bar{v}\colon \dom(T_0) \to \dom(T) \union C$
  be an injective mapping with
  \begin{align}
    \label{eq:bar-v}
    \bar{v}(c) = c
    \quad
    \text{for each $c \in \dom(T_0) \intersection (\const(T) \union C)$},
  \end{align}
  and
  \begin{align}
    \label{eq:bar-v-2}
    \bar{v}(c) \in \nulls(T)
    \quad
    \text{for each $c \in \dom(T_0) \setminus (\const(T) \union C)$}.
  \end{align}
  Then $f \isdef \compose{v_0}{\bar{v}} \in \val_C(T)$,
  and
  \begin{align}
    \label{eq:T_0'}
    T_0' \isdef f(T) = \bar{v}(v_0(T)) = \bar{v}(T_0).
  \end{align}
  Let $v$ be the inverse of $\bar{v}$ on $\dom(T_0')$.
  Then, by \eqref{eq:bar-v}--\eqref{eq:T_0'},
  $v$ is an injective valuation of $T_0'$ such that $v(T_0') = T_0$,
  and $v^{-1}(c) = \bar{v}(c) = c$ for every $c \in \dom(T_0) \intersection C$.

  It remains to show that $T_0' \in \minrep_C(T)$.
  By \eqref{eq:T_0'}, we have $T'_0 = f(T)$, where $f \in \val_C(T)$.
  Suppose, for a contradiction, that there is an $f' \in \val_C(T)$
  with $f'(T) \subsetneq T_0'$.
  Since $v$ is injective and $v(T_0') = T_0$, we then have
  \(
    v(f'(T)) \subsetneq v(T_0') = T_0,
  \)
  which is impossible, since $v(f'(T)) \in \rep(T)$,
  and $T_0$ is a minimal instance in $\rep(T)$.

  We next prove that \ref{prop:minrep-props/capture/3} implies
  \ref{prop:minrep-props/capture/1}.
  Suppose that $T'_0 \in \minrep_C(T)$
  and that $v$ is an injective valuation of $T'_0$ with $v(T'_0) = T_0$
  (we will not need the restriction that $v^{-1}(c) = c$
   for all $c \in \dom(T_0) \intersection C$).
  We show that $T_0$ is a minimal instance in $\rep(T)$.
  To this end, let $f \in \val_C(T)$ be such that $f(T) = T'_0$.
  Then $v_0 \isdef \compose{f}{v}$ is a valuation of $T$,
  so that
  \[
    T_0 = v(T'_0) = v(f(T)) = v_0(T) \in \rep(T).
  \]
  It remains, therefore, to show that there is no $\tilde{T}_0 \in \rep(T)$
  with $\tilde{T}_0 \subsetneq T_0$.

  Suppose, to the contrary, that there is such a $\tilde{T}_0$.
  Let $\tilde{v}_0$ be a valuation of $T$ with $\tilde{v}_0(T) = \tilde{T}_0$,
  and let $\tilde{f} \isdef \compose{\tilde{v}_0}{v^{-1}}$,
  where $v^{-1}$ is the inverse of $v$ on $\dom(T_0)$.
  Since $v^{-1}$ is an injective mapping on $\dom(T_0)$, we have
  \(
    \tilde{f}(T)
    =
    v^{-1}(\tilde{v}_0(T))
    =
    v^{-1}(\tilde{T}_0)
    \subsetneq
    v^{-1}(T_0)
    =
    T'_0,
  \)
  which is impossible,
  since $\tilde{f} \in \val_C(T)$ and $T'_0 \in \minrep_C(T)$.

  \ad{\ref{prop:minrep-props/core-is-minimal}}
  Clearly, the identity $f$ on $\dom(T)$ belongs to $\val_C(T)$
  and satisfies $f(T) = T$.
  Let $f' \in \val_C(T)$ be such that $f'(T) \subseteq T$.
  Then $f'$ is a homomorphism from $T$ to $T$, and since $T$ is a core,
  we cannot have $f'(T) \subsetneq T$.

  \ad{\ref{prop:minrep-props/core}}
  Let $f \in \val_C(T)$ be such that $T_0 \isdef f(T) \in \minrep_C(T)$.
  For a contradiction, suppose that $T_0$ is not a core.
  Let $h$ be a homomorphism from $T_0$ to $T_0$ such that
  $h(T_0)$ is a core of $T_0$.
  Since $T_0$ is not a core, we have $h(T_0) \subsetneq T_0$.
  Thus, for $f' \isdef \compose{f}{h}$,
  we have $f'(T) = h(f(T)) = h(T_0) \subsetneq T_0$,
  which contradicts $T_0 \in \minrep_C(T)$.
  Hence, $T_0$ is a core.
\end{proof}

The converse of Proposition~\ref{prop:minrep-props}%
(\ref{prop:minrep-props/core})
is not true, as shown by the following example:

\begin{exa}
  Let $T$ be an instance over $\sigma = \set{E,P}$,
  where $E^T = \set{(a,\bot),(\bot,\bot')}$ and $P^T = \set{b}$.
  The mapping $f \in \val_C(T)$ with $f(\bot) = a$ and $f(\bot') = b$
  then yields the instance $f(T)$,
  where $E^{f(T)} = \set{(a,a),(a,b)}$ and $P^{f(T)} = \set{b}$.
  Hence, $f(T)$ is a core.
  However, $f(T)$ does not belong to $\minrep_C(T)$,
  since the mapping $f' \in \val_C(T)$ with $f'(\bot) = f'(\bot') = a$
  yields the instance $f'(T)$ with
  $E^{f'(T)} = \set{(a,a)}$ and $P^{f'(T)} = \set{b}$,
  which is a proper subinstance of $f(T)$.
\end{exa}

Note that the size of $\minrep_C(T)$ can be exponential in the size of $T$,
so that it is not possible to enumerate all instances in $\minrep_C(T)$
in polynomial time, given $T$ and $R(\tup{t})$ as input.
To tackle this problem, we take advantage of a nice structural property
of $T$ that can be described in terms of \emph{atom blocks}:

\begin{defi}[atom block \cite{GN:JACM08}]
  \label{def:atom-block}
  Let $T$ be an instance.
  \begin{iteMize}{$\bullet$}
  \item
    The \emph{Gaifman graph of the atoms of $T$}
    is the undirected graph whose vertices are the atoms of $T$,
    and which has an edge between two atoms $A,A' \in T$
    if and only if $A \neq A'$,
    and there is a null that occurs both in $A$ and $A'$.
  \item
    An \emph{atom block} of $T$ is the set of atoms in a connected component
    of the Gaifman graph of the atoms of $T$.
  \end{iteMize}
\end{defi}

\noindent Note that each atom block of $T$ is a subinstance of $T$.
Furthermore, for each atom block $B$ of $T$ that contains at least one null,
$\nulls(B)$ is a \emph{block} as considered in \cite{FKP:TODS30-1}.
The crucial property of $T$ is:

\begin{lem}[\cite{FKP:TODS30-1}]
  \label{lemma:bounded-blocksize}
  For every schema mapping $M = \Des$, where $\Sigma$ consists of st-tgds,
  there is a positive integer $\bs$ such that
  if $S$ is a source instance for $M$,
  and $B$ is an atom block of $\core(M,S)$,
  then $\card{\nulls(B)} \leq \bs$.
\end{lem}

Let us come back to our initial problem --
to decide whether there is a minimal instance in $\rep(\core(M,S))$
that contains the ground atom $R(\tup{t})$.
Let $T \isdef \core(M,S)$, and let $C$ be the set of constants in $\tup{t}$.
By Proposition~\ref{prop:minrep-props}(\ref{prop:minrep-props/capture})
it is enough to decide whether there is a $T_0 \in \minrep_C(T)$
with $R(\tup{t}) \in T_0$.
The following algorithm seems to accomplish this task:
\begin{enumerate}[(1)]
\item
  Compute the atom blocks of $T$.
\item
  Consider the atom blocks $B$ of $T$ in turn, and
\item\label{step:XGCWA-answers/universal/atoms/algo1/if}
  if there is an instance $B_0 \in \minrep_C(B)$ with $R(\tup{t}) \in B_0$,
  accept the input; \\
  otherwise reject it.
\end{enumerate}
Since, by Lemma~\ref{lemma:bounded-blocksize},
there is a constant $\bs$ with $\card{\nulls(B)} \leq \bs$
for each atom block $B$ of $T$,
we have to consider at most $\card{\val_C(B)} = \card{\dom(B) \union C}^\bs$
mappings in step~\ref{step:XGCWA-answers/universal/atoms/algo1/if}
to find all the instances $B_0 \in \minrep_C(B)$.
Thus, the whole algorithm runs in polynomial time.

Example~\ref{ex:naive-algorithm-fails} below shows that
this algorithm is incorrect.
In particular, the example exhibits an instance $T$ that is a core,
and an atom block $B$ of $T$
such that there is an atom $A$ of some minimal instance $B_0 \in \rep(B)$
that is not an atom of any minimal instance in $\rep(T)$.
Letting $C$ be the set of all constants in $A$,
this implies that there is an atom of some instance $B_0 \in \minrep_C(B)$
that is not an atom of any instance in $\minrep_C(T)$.

\begin{exa}
  \label{ex:naive-algorithm-fails}
  Let $T$ be the instance over $\set{E}$ with
  \[
    E^T = \set{(a,b),(a,\bot),(b,\bot),(b,\bot'),(b,\bot''),(\bot',\bot'')},
  \]
  and consider the atom block
  \[
    B = \set{E(b,\bot'),E(b,\bot''),E(\bot',\bot'')}
  \]
  of $T$;
  see Figure~\ref{fig:naive-algorithm-fails} for a graph representation
  of $T$ and $B$.
  \begin{figure}[ht]
    \centering
    \includeTikZ{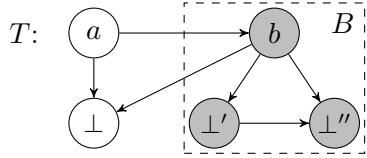}
    \caption{The instance $T$.
      The subinstance induced by the gray vertices is $B$.}
    \label{fig:naive-algorithm-fails}
  \end{figure}
  Note that $T$ is a core.
  It is not hard to see that every minimal instance in $\rep(B)$
  has one of the following forms:
  \begin{enumerate}[(1)]
  \item
    $\set{E(b,b)}$,
  \item
    $\set{E(b,c),E(c,c)}$ with $c \in \Const \setminus \set{b}$, or
  \item\label{third-form}
    $\set{E(b,c),E(b,c'),E(c,c')}$ with $c,c' \in \Const \setminus \set{b}$
    and $c \neq c'$.
  \end{enumerate}
  Thus, there is a minimal instance in $\rep(B)$
  of the third form that contains $E(c,a)$
  for some $c \in \Const \setminus \set{b}$
  (replace $c'$ in \ref{third-form} with $a$).

  However, there is no minimal instance in $\rep(T)$
  that contains $E(c,a)$:
  Such an instance must be obtained from $T$ by a valuation $v$ of $T$
  with $v(\bot') = c$ and $v(\bot'') = a$,
  since $E(\bot',\bot'')$ is the only atom in $T$ that could be the preimage
  of $E(c,a)$ -- all other atoms either have $a$ or $b$ as their first value.
  However,
  let $v$ be a valuation of $T$ with $v(\bot') = c$ and $v(\bot'') = a$,
  and let $f\colon \dom(T) \to \dom(T)$ be such that $f(a) = a$, $f(b) = b$,
  $f(\bot') = a$ and $f(\bot) = f(\bot'') = \bot$.
  Then, for $v' \isdef \compose{f}{v}$, we have
  \begin{align*}
    v'(T)
    & = \set{E(a,b),E(b,a),E(a,v(\bot)),E(b,v(\bot))} \\
    & \subsetneq \set{E(a,b),E(b,a),E(a,v(\bot)),E(b,v(\bot)),E(b,c),E(c,a)}
      = v(T).
  \end{align*}
  Thus, $v(T)$ is not minimal in $\rep(T)$.
\end{exa}

It is nevertheless possible to solve our initial problem
using the following approach.
Let $T = \core(M,S)$, let $R(\tup{t})$ be a ground atom,
and let $C$ be the set of all constants in $\tup{t}$.
Our goal is to decide whether there is an instance $T_0 \in \minrep_C(T)$
with $R(\tup{t}) \in T_0$.
To this end,
we identify a set $\CS \subseteq \min_C(T)$
of size polynomial in the size of $T$
such that $R(\tup{t})$ occurs in an instance in $\minrep_C(T)$
if and only if $R(\tup{t})$ occurs in an instance in $\CS$.
Furthermore, we ensure that $\CS$ can be computed in polynomial time
from $T$ and $C$.
To define $\CS$, we need a few definitions.

In the following, we fix, for each instance $I$, a core $\core(I) \subseteq I$,
namely the output of the algorithm provided by the following lemma:

\begin{lem}[implicit in \cite{FKP:TODS30-1}]
  \label{lem:block-core}
  There is an algorithm that takes an instance $I$ as input,
  and outputs a core $J \subseteq I$ of $I$ in time $O(n^{b+3})$,
  where $n$ is the size of $I$ and $b$ is the maximum number of nulls
  in an atom block of $I$.
\end{lem}

\begin{proof}
  Just omit the first step of the \emph{blocks algorithm}
  from \cite{FKP:TODS30-1}.
  That is, given an instance $I$, proceed as follows:
  \begin{enumerate}[(1)]
  \item
    Compute a list $B_1,\dotsc,B_m$ of all atom blocks of $I$,
    and initialize $J$ to be $I$.
  \item\label{step:find-h}
    Check whether there is a homomorphism $h$ from $J$ to $J$
    such that $h$ is not injective,
    and there is some $i \in \set{1,\dotsc,m}$ such that $h(u) = u$
    for each $u \in \dom(J) \setminus \nulls(B_i)$.
  \item
    If such a $h$ exists, replace $J$ by $h(J)$,
    and go to step~\ref{step:find-h}.
  \item
    Output $J$.
  \end{enumerate}
  Now the lemma follows from the proof of \cite[Theorem~5.9]{FKP:TODS30-1}.
\end{proof}

Given $T = \core(M,S)$ and $C$ as above,
we define the set $\CS$ to be the union of the following sets $\minrep_C(T,B)$
over all atom blocks $B$ of $T$.

\begin{defi}[$\minval_C(T,B)$, $\minrep_C(T,B)$]
  \label{def:minval_B}
  Let $T$ be an instance, let $B$ be an atom block of $T$,
  let $\compl{B} \isdef T \setminus B$,
  and let $C \subseteq \Const$.
  \begin{enumerate}[(1)]
  \item
    Let $\val_C(T,B)$ be the set of all mappings $f \in \val_C(T)$
    such that
    \begin{itemize}
    \item
      $f(\bot) = \bot$ for all $\bot \in \nulls(\compl{B})$, and
    \item
      all nulls that occur in $f(B) \setminus \compl{B}$ belong to $\nulls(B)$.
    \end{itemize}
  \item
    Let $\minval_C(T,B)$ be the set of all $f \in \val_C(T,B)$
    such that there is no $f' \in \val_C(T,B)$ with $f'(T) \subsetneq f(T)$.
  \item
    Let $\minrep_C(T,B) \isdef \set{\core(f(T)) \mid f \in \minval_C(T,B)}$.
  \end{enumerate}
\end{defi}

\noindent Using Lemma~\ref{lem:block-core}, we obtain:

\begin{prop}
  \label{prop:minrep-computation}
  For each positive integer $\bs$,
  there is a polynomial time algorithm that,
  given an instance $T$ such that the number of nulls in each atom block
  of $T$ is at most $\bs$,
  and a set $C \subseteq \Dom$,
  outputs a list of all instances that occur in $\minrep_C(T,B)$
  for some atom block $B$ of $T$.
\end{prop}

\begin{proof}
  The algorithm is as follows:
  Given $T$ and $C$, first compute a list $f_1,\dotsc,f_m$
  of all mappings $f$ such that there is an atom block $B$ of $T$
  with $f \in \minval_C(T,B)$.
  This can be done in time polynomial in the size of $T$.
  Then, compute and output $\core(f_i(T))$
  for each $i \in \set{1,\dotsc,m}$.
  By Lemma~\ref{lem:block-core}, this can be done in time $O(n^{\bs+3})$,
  where $n$ is the size of $T$
  (note that the number of nulls in each atom block of $f_i(T)$
   is bounded by $\bs$).
\end{proof}

The following Lemma~\ref{lemma:minval} tells us that the instances
in $\minrep_C(T,B)$ indeed belong to $\minrep_C(T)$.
Before stating the lemma, let us introduce retractions.
Given an instance $I$, a \emph{retraction} of $I$ is a homomorphism $h$
from $I$ to $I$ such that $h(u) = u$ for all elements $u$ in the range of $h$.
In particular, for all atoms $A \in h(I)$, we have $A \in I$ and $h(A) = A$.
It is known that a core of $I$ is an instance $J$ for which there is
a retraction $h$ of $I$ with $h(I) = J$,
and there is no retraction of $J$ to a proper subinstance of $J$
(\cite{HN:DM109}).
A \emph{retraction of $I$ over a set $X \subseteq \Dom$}
is a retraction $h$ of $I$ such that $h(u) = u$
for each $u \in X \intersection \dom(I)$.

\begin{lem}
  \label{lemma:minval}
  Let $T$ be an instance, let $B$ be an atom block of $T$,
  let $\compl{B} \isdef T \setminus B$, and let $C \subseteq \Const$.
  Then for each $f \in \minval_C(T,B)$,
  there is a retraction $h$ of $\hat{T} \isdef f(T)$
  over the set of the nulls of $f(B) \setminus \compl{B}$ such that
  \begin{enumerate}[\em(1)]
  \item
    $h(\hat{T})$ is a core of $\hat{T}$, and
  \item
    $h(\hat{T}) \in \minrep_C(T)$.
  \end{enumerate}
  In particular, $\minrep_C(T,B) \subseteq \minrep_C(T)$.
\end{lem}

\begin{proof}
  Let $\CA \isdef f(B) \setminus \compl{B}$,
  and let $h$ be a retraction of $\hat{T}$ over $\nulls(\CA)$
  such that for
  \begin{align}
    \label{eq:minval/T_0}
    \hat{T}_0 \isdef h(\hat{T}) = h(f(T))
  \end{align}
  we have:
  \begin{align}
    \label{eq:minval/h}
    \text{There is no retraction $h'$ of $\hat{T}_0$ over $\nulls(\CA)$
      with $h'(\hat{T}_0) \subsetneq \hat{T}_0$.}
  \end{align}
  We show that $\hat{T}_0$ is a core of $\hat{T}$,
  and that $\hat{T}_0 \in \minrep_C(T)$.

  \step{1}{$\hat{T}_0$ is a core of $\hat{T}$.}
  Suppose, for a contradiction, that $\hat{T}_0$ is not a core of $\hat{T}$.
  Then there is a retraction $h'$ of $\hat{T}_0$
  with $h'(\hat{T}_0) \subsetneq \hat{T}_0$.
  By \eqref{eq:minval/h},
  there is some $\bot \in \nulls(\CA)$ with $h'(\bot) \neq \bot$.
  Let $A$ be an atom in $\CA$ that contains $\bot$.
  Since $h'(\bot) \neq \bot$ and $h'$ is a retraction,
  $\bot$ does not occur in the range of $h'$,
  and therefore $A$ does not occur in $h'(\CA)$.
  Together with $h'(\CA) \subseteq \CA \union \compl{B}$
  and $A \in \CA$,
  this implies
  \begin{align}
    \label{eq:minval/retr}
    h'(\CA) \setminus \compl{B} \,\subsetneq\, \CA.
  \end{align}

  Consider the mapping $f'\colon \dom(T) \to \dom(T) \union C$
  defined for each $u \in \dom(T)$ by
  \begin{align*}
    f'(u) \isdef
    \begin{cases}
      h'\bigl(h(f(u))\bigr), & \text{if $u \in \nulls(B)$}, \\
      u, & \text{otherwise}.
    \end{cases}
  \end{align*}
  Then $f' \in \val_C(T,B)$,
  since $f \in \val_C(T,B)$ and $h,h'$ are retractions.
  Moreover,
  \begin{align}
    \label{eq:minval/f'}
    f'(B)
    \subseteq h'(\hat{T}_0)
    = h'(\hat{T}_0 \setminus \compl{B}) \union h'(\compl{B})
    \subseteq h'(\hat{T}_0 \setminus \compl{B}) \union \compl{B},
  \end{align}
  where the first inclusion holds due to
  $f'(B) = h'(h(f(B))) \subseteq h'(h(f(T)))$ and \eqref{eq:minval/T_0},
  and the last inclusion holds,
  because $h'$ is a retraction of $\hat{T}_0$.

  Observe also that
  \begin{align}
    \label{eq:minval/A}
    \hat{T}_0 \setminus \compl{B} \subseteq \CA.
  \end{align}
  Indeed, let $A$ be an atom of $\hat{T}_0$ with $A \notin \compl{B}$.
  By \eqref{eq:minval/T_0}, we have $\hat{T}_0 = h(f(T))$.
  Together with $f(T) = f(B) \union f(\compl{B})$,
  and the fact that $h$ is a retraction of $f(T)$,
  this implies that $A \in f(B)$ or $A \in f(\compl{B})$.
  Note that $A \notin f(\compl{B})$,
  because $f(\compl{B}) = \compl{B}$ and $A \notin \compl{B}$.
  Hence, $A \in f(B)$,
  which, together with $A \notin \compl{B}$, implies that $A \in \CA$.

  Consequently, we have
  \begin{align*}
    f'(T)
    & = f'(B) \union \compl{B}
      && \text{since $T = B \union \compl{B}$ and $f' \in \val_C(T,B)$} \\
    & \subseteq h'(\CA) \union \compl{B}
      && \text{by \eqref{eq:minval/f'} and \eqref{eq:minval/A}} \\
    & \subsetneq \CA \union \compl{B}
      && \text{by \eqref{eq:minval/retr}}\\
    & = f(B) \union f(\compl{B})
      && \text{by definition of $\CA$, and $f \in \val_C(T,B)$}\\
    & = f(T),
  \end{align*}
  which contradicts $f \in \minval_C(T,B)$.
  Thus, $\hat{T}_0$ is a core of $\hat{T}$.

  \step{2}{$\hat{T}_0 \in \minrep_C(T)$.}
  First observe that $\hat{T}_0 = f_0(T)$,
  where $f_0 \isdef \compose{f}{h}$ and $f_0 \in \val_C(T)$.
  It remains, therefore, to show that there is no $f' \in \val_C(T)$
  with $f'(T) \subsetneq \hat{T}_0$.

  Suppose, for a contradiction, that there is such a mapping $f'$.
  Without loss of generality, we may assume that $f'(T) \in \minrep_C(T)$.
  Moreover,
  \begin{align}
    \label{eq:minval/f'->f}
    f'(T)
    \subseteq \hat{T}_0
    \stackrel{\eqref{eq:minval/T_0}}{=} h(f(T))
    \subseteq f(T)
    = f(B) \union \compl{B}.
  \end{align}
  Next observe that
  \begin{align}
    \label{eq:minval/B-ident}
    f'(B) \setminus \compl{B} = f(B) \setminus \compl{B}.
  \end{align}
  Otherwise, we could use $f'$ to construct a mapping $f'' \in \val_C(T,B)$
  such that $f''(T) \subsetneq f(T)$,
  which is impossible, because $f \in \minval_C(T,B)$.
  Indeed, assume that \eqref{eq:minval/B-ident} is not true.
  Then,
  \begin{align}
    \label{eq:minval/not-B-ident}
    f'(B) \setminus \compl{B} \subsetneq f(B) \setminus \compl{B},
  \end{align}
  since by \eqref{eq:minval/f'->f}
  we have $f'(B) \subseteq f(B) \union \compl{B}$.
  Let $f''\colon \dom(T) \to \dom(T) \union C$
  be such that for each $u \in \dom(T)$,
  $f''(u) = f'(u)$ if $u \in \nulls(B)$,
  and $f''(u) = u$ otherwise.
  Then it is not hard to see that $f'' \in \val_C(T,B)$.
  Moreover, we have
  \begin{align*}
    f''(T)
    =
    (f'(B) \setminus \compl{B}) \union \compl{B}
    \stackrel{\eqref{eq:minval/not-B-ident}}{\subsetneq}
    (f(B) \setminus \compl{B}) \union \compl{B}
    =
    f(T),
  \end{align*}
  as claimed.

  Now, \eqref{eq:minval/B-ident} implies that the mapping
  $h'\colon \dom(T) \to \dom(T) \union C$
  that is defined for each $u \in \dom(T)$ by
  $h'(u) = u$ if $u \in \nulls(B)$, and $h'(u) = f'(u)$ otherwise,
  is a homomorphism from $f(T) = f(B) \union \compl{B}$ to $f'(T)$.
  Furthermore, by \eqref{eq:minval/f'->f} we have $f'(T) \subseteq f(T)$,
  and therefore,
  the identity on $\dom(f'(T))$ is a homomorphism from $f'(T)$ to $f(T)$.
  This implies that $f'(T)$ and $f(T)$ are homomorphically equivalent,
  and therefore, their cores are isomorphic.
  Since $\hat{T}_0$ is a core of $f(T)$ as shown in Step~1,
  and $f'(T)$ is a core
  by Proposition~\ref{prop:minrep-props}(\ref{prop:minrep-props/core}),
  we have $\hat{T}_0 \isomorphic f'(T)$.
  However, this is a contradiction to our earlier assumption
  that $f'(T) \subsetneq \hat{T}_0$.
\end{proof}

Clearly, the union of the sets $\minrep_C(T,B)$ over all atom blocks $B$ of $T$
does not cover the whole set of instances in $\minrep_C(T)$.
However, Lemma~\ref{lemma:universal-queries/basis} below tells us
that for each atom $A$ of some instance $T_0 \in \minrep_C(T)$
there is an atom block $B$ of $T$ and an instance $T_B \in \minrep_C(T,B)$
that contains an atom $A'$ isomorphic to $A$ in the following sense:

\begin{notation}
  We say that two atoms $A_1,A_2$ are \emph{isomorphic},
  and we write $A_1 \isomorphic A_2$,
  if the instances $\set{A_1}$ and $\set{A_2}$ are isomorphic.
\end{notation}

Note that $R(u_1,\dotsc,u_r)$ and $R'(u'_1,\dotsc,u'_{r'})$
are isomorphic if and only if $R = R'$, $r = r'$,
and for all $i,j \in \set{1,\dotsc,r}$,
$u_i \in \Const$ if and only if $u_i' \in \Const$,
$u_i \in \Const$ implies $u_i = u_i'$,
and $u_i = u_j$ if and only if $u'_i = u'_j$.

Lemma~\ref{lemma:universal-queries/basis} is based on the following notion
of a \emph{packed} atom block:

\begin{defi}[packed atom block]
  An atom block $B$ of an instance is called \emph{packed}
  if for all atoms $A,A' \in B$ with $A \neq A'$,
  there is a null that occurs both in $A$ and $A'$.
\end{defi}

Immediately from the definitions, we obtain:

\begin{prop}
  \label{prop:packed}
  If $M = \Des$ is a schema mapping, where $\Sigma$ consists of packed st-tgds,
  and $S$ is a source instance for $M$,
  then each atom block of $\core(M,S)$ is packed.
\end{prop}

We are now ready to state the main result
of the present Section~\ref{sec:XGCWA-answers/universal/atoms}:

\begin{lem}
  \label{lemma:universal-queries/basis}
  Let $T$ be an instance such that $T$ is a core,
  and each atom block of $T$ is packed.
  Let $C \subseteq \Const$, and let $T_0 \in \minrep_C(T)$.
  Then for each atom $A \in T_0$, there are
  \begin{enumerate}[\em(1)]
  \item
    an atom block $B$ of $T$,
  \item
    an instance $T_B \in \minrep_C(T,B)$,
  \item
    an atom $A' \in T_B$ with $A' \isomorphic A$, and
  \item
    a homomorphism $h$ from $T_B$ to $T_0$ with $h(T_B) = T_0$ and $h(A') = A$.
  \end{enumerate}
\end{lem}

\begin{proof}
  Let $T$ be an instance such that $T$ is a core
  and each atom block of $T$ is packed.
  Let $C \subseteq \Const$,
  and let $T_0 \in \minrep_C(T)$.
  Furthermore,
  let $B_1,\dotsc,B_n$ be an enumeration of all the atom blocks of $T$.
  We prove the following stronger statement:
  \begin{align}
    \tag{$\star$}
    \label{eq:universal-queries/basis/main-claim}
    \begin{minipage}{0.85\linewidth}
      Let $i \in \set{1,\dotsc,n}$.
      Then there is an instance $T_i \in \minrep_C(T,B_i)$
      and a homomorphism $h_i$ from $T_i$ to $T_0$ with $h_i(T_i) = T_0$
      such that the following is true:
      For each atom $A \in T_0$,
      there is an index $j \in \set{1,\dotsc,n}$ and an atom $A' \in T_j$
      with $h_j(A') = A$ and $A' \isomorphic A$.
    \end{minipage}
  \end{align}

  \smallskip

  \proofpartnl{Idea of the construction.}
  We start with an $f \in \val_C(T)$ such that $f(T) = T_0$.
  The first step is to find mappings
  $f_1 \in \val_C(T,B_1),\dotsc,f_n \in \val_C(T,B_n)$
  such that each $f_i(B_i)$ is isomorphic to $f(B_i)$.
  This is easy since we can assign an ``unused'' null in $B_i$
  to any null in $B_i$ that is mapped by $f$ to a null outside $B_i$.

  We then ``minimize'' each $f_i(T)$ by picking a $g_i \in \minval_C(T,B_i)$
  with $g_i(T) \subseteq f_i(T)$.
  The instance $T_i$ is then defined to be the core of $g_i(T)$
  (that contains $g_i(B_i) \setminus (T \setminus B_i)$).
  It is then not hard to define a homomorphism $h_i$ from $T_i$ to $T_0$
  with $h_i(T_i) = T_0$;
  see Figure~\ref{fig:main-lemma-mappings} for an illustration.
  \begin{figure}
    \centering
    \includeTikZ{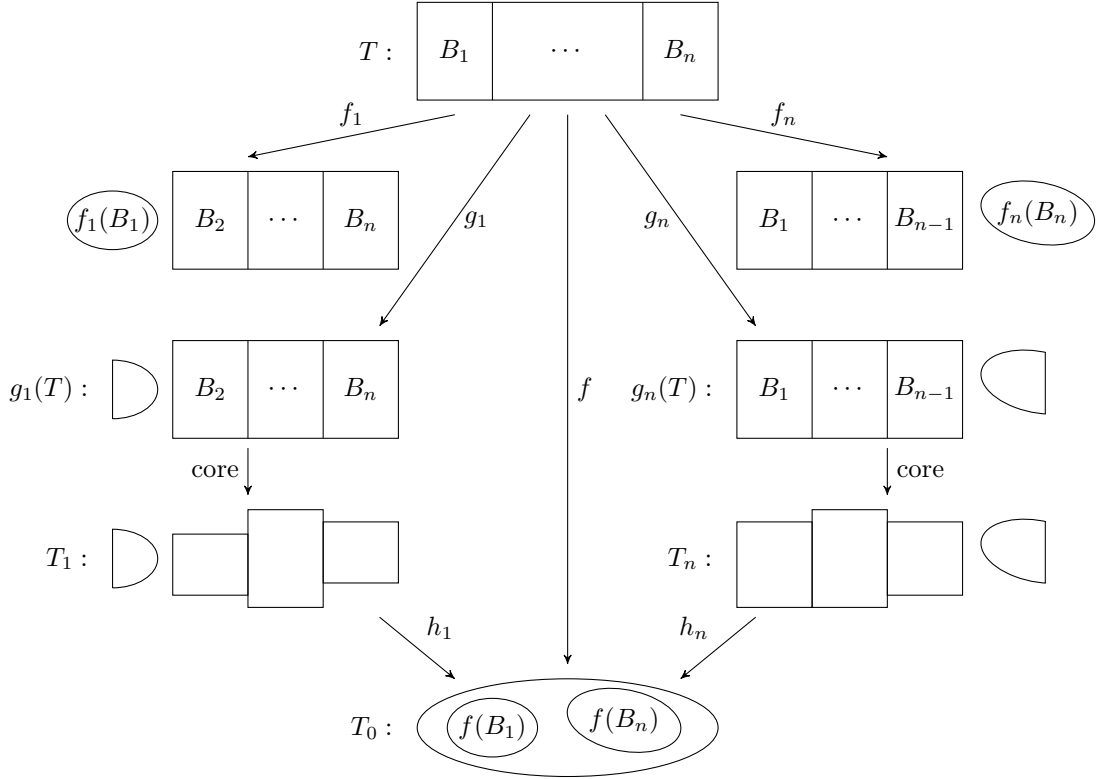}
    \caption{The mappings $f,f_i,g_i,h_i$ and their relations.}
    \label{fig:main-lemma-mappings}
  \end{figure}

  Finally, we have to show that for each atom $A \in T_0$,
  there is a $j \in \set{1,\dotsc,n}$ and an atom $A' \in T_j$
  with $h_j(A') = A$ and $A' \isomorphic A$.
  This is not an immediate consequence of the construction of $T_i$ and $g_i$.
  To explain the problem, let us pick an atom $A \in T_0$.
  This atom must occur in some $f(B_j)$,
  possibly in more than one such $f(B_j)$.
  Let $j \in \set{1,\dotsc,n}$ be such that $A \in f(B_j)$.
  It will be clear from the construction of the $f_j$ and $h_j$
  that there is an atom $A' \in f_j(B_j)$ with $A' \isomorphic A$.
  If $A' \in g_j(B_j)$ and $A' \notin T \setminus B_j$,
  we are done:
  in this case we will have $A' \in T_j$ and $h_j(A') = A$.
  However, if $A' \notin g_j(B_j)$ or $A' \in T \setminus B_j$,
  it may be that $T_j$ contains no atom isomorphic to $A$,
  or -- in the case that $T$ contains such an atom, say $A''$ --
  that $h_j$ does not map $A''$ to $A$.

  The extreme case is that for all $j \in \set{1,\dotsc,n}$ with $A \in f(B_j)$,
  there is no atom $A' \in g_j(B_j)$ with $A' \notin T \setminus B_j$
  and $A' \isomorphic A$.
  Using the property that all of the atom blocks $B_1,\dotsc,B_n$ are packed,
  we show that this case does not occur,
  whereby proving \eqref{eq:universal-queries/basis/main-claim}.
  (Example~\ref{ex:universal-queries/basis/failure} below
   shows that it may occur if the atom blocks are not packed.)

  It is helpful here to think in terms of the following graph $G$:%
  \label{def:universal-queries/basis/G}
  the nodes of $G$ are the atoms of $T$,
  and there is an edge from a node $A' \in B_j$ to a node $A''$
  if $g_j(A') \in T \setminus B_j$ and $g_j(A') = A''$.
  The core of the proof can then be summarized as follows:
  We first show that any path in $G$
  must eventually reach a node $A' \in B_j$ for some $j \in \set{1,\dotsc,n}$
  such that $g_j(A') \notin T \setminus B_j$.
  Otherwise, there would be a cycle in $G$ containing a node $A'$.
  This would imply that $A'' \isdef g_j(A')$ is isomorphic to $A'$,
  and that $A'' \notin B_j$.
  But since $B_j$ is packed,
  this implies that $g_j$ is actually a homomorphism from $B_j$
  to $T \setminus B_j$,
  which is impossible since $T$ is a core.
  Using this property,
  we then construct --
  basically by repeated application of the mappings $g_1,\dotsc,g_n$
  followed by a renaming of the nulls --
  a mapping $f' \in \val_C(T)$ such that $T' \isdef f'(T) \subseteq T_0$.
  If there are no $j \in \set{1,\dotsc,n}$ and $A' \in T_j$
  with $h_j(A') = A$ and $A' \isomorphic A$,
  we will have $A \notin T'$,
  which would mean that $T' \subsetneq T_0$
  and contradict $T_0 \in \minrep_C(T)$.
  Consequently, there is a $j \in \set{1,\dotsc,n}$ and an atom $A' \in T_j$
  with $h_j(A') = A$ and $A' \isomorphic A$,
  which proves \eqref{eq:universal-queries/basis/main-claim}.

  \smallskip

  \proofpartnl{The details.}
  Let $f \in \val_C(T)$ be such that $f(T) = T_0$.
  The construction of the instances $T_i$ and the homomorphisms $h_i$
  proceeds in three steps.
  First, we ``split'' $f$ into mappings
  $f_1 \in \val_C(T,B_1),\dotsc,f_n \in \val_C(T,B_n)$
  such that each $f_i(B_i)$ is isomorphic to $f(B_i)$.
  Second, we use these mappings to construct the instances $T_i$
  and the homomorphisms $h_i$.
  Third, we show that for each atom $A \in T_0$,
  there is a $j \in \set{1,\dotsc,n}$ and an atom $A' \in T_j$
  with $h_j(A') = A$ and $A' \isomorphic A$.

  \step{1}{Construction of $f_1,\dotsc,f_n$.}
  Let $i \in \set{1,\dotsc,n}$.
  We construct a mapping $f_i \in \val_C(T,B_i)$
  and an injective homomorphism $r_i$ from $f_i(B_i)$ to $f(B_i)$
  with $r_i(f_i(B_i)) = f(B_i)$ as follows.
  Pick an injective mapping
  \[
    \bar{r}_i\colon \dom(f(B_i)) \to \const(f(B_i)) \union \nulls(B_i)
  \]
  such that $\bar{r}_i(c) = c$ for each $c \in \const(f(B_i))$,
  and $\bar{r}_i(\bot) \in \nulls(B_i)$ for each $\bot \in \nulls(f(B_i))$.
  Then define $f_i\colon \dom(T) \to \dom(T) \union C$
  such that for each $u \in \dom(T)$,
  \[
    f_i(u) \,\isdef\,
    \begin{cases}
      \bar{r}_i(f(u)), & \text{if $u \in \nulls(B_i)$} \\
      u, & \text{otherwise}.
    \end{cases}
  \]
  By construction, we have $f_i \in \val_C(T,B_i)$.
  Furthermore, for each atom $A$ of $f(B_i)$,
  we have $\bar{r_i}(A) \isomorphic A$.
  In particular, each atom of $f(B_i)$ is isomorphic to an atom of $f_i(B_i)$,
  and vice versa.
  Let $r_i$ be the inverse of $\bar{r}_i$ on $\dom(f_i(B_i))$.
  Then, $r_i$ is an injective homomorphism from $f_i(B_i)$ to $f(B_i)$ with
  \begin{align}
    \label{eq:universal-queries/basis/proj}
    r_i(f_i(B_i))
    \,=\,
    r_i(\bar{r}_i(f(B_i)))
    \,=\,
    f(B_i).
  \end{align}
  In particular,
  \begin{align}
    \label{eq:universal-queries/basis/iso}
    r_i(A) \,\isomorphic\, A
    \quad \text{for all atoms $A \in f_i(B_i)$}.
  \end{align}

  \smallskip

  \step{2}{Construction of the instances $T_i$ and the homomorphisms $h_i$.}
  Let $i \in \set{1,\dotsc,n}$,
  and pick $g_i \in \minval_C(T,B_i)$ with $g_i(T) \subseteq f_i(T)$.
  By Lemma~\ref{lemma:minval},
  there is a retraction $h'_i$ of $g_i(T)$
  over the set of the nulls of $g_i(B_i) \setminus (T \setminus B_i)$
  such that
  \begin{align}
    \label{eq:T_i}
    T_i \,\isdef\, h'_i(g_i(T)) \,\in\, \minrep_C(T).
  \end{align}
  Define $h_i\colon \dom(g_i(T)) \to \dom(T_0)$
  such that for each $u \in \dom(g_i(T))$,
  \begin{align*}
    h_i(u) \,\isdef\,
    \begin{cases}
      r_i(u), & \text{if $u \in \dom(g_i(T) \setminus (T \setminus B_i))$} \\
      f(u),   & \text{otherwise.}
    \end{cases}
  \end{align*}
  Note that $r_i$ is defined for all values that occur in
  $\dom(g_i(T) \setminus (T \setminus B_i))$,
  since $g_i(T) \subseteq f_i(T) = f_i(B_i) \union (T \setminus B_i)$,
  and therefore,
  \begin{align}
    \label{eq:universal-queries/basis/g_i_T}
    g_i(T) \setminus (T \setminus B_i) \,\subseteq\, f_i(B_i).
  \end{align}
  Furthermore,
  \begin{align*}
    h_i(g_i(T) \setminus (T \setminus B_i))
    & \,=\,
    r_i(g_i(T) \setminus (T \setminus B_i))
    \,\stackrel{\eqref{eq:universal-queries/basis/g_i_T}}{\subseteq}\,
    r_i(f_i(B_i))
    \,\stackrel{\eqref{eq:universal-queries/basis/proj}}{=}\,
    f(B_i)
    \,\subseteq\,
    T_0, \\
    \intertext{and}
    h_i(T \setminus B_i)
    & \,=\, f(T \setminus B_i) \,\subseteq\, f(T) \,=\, T_0,
  \end{align*}
  which yields $h_i(g_i(T)) \subseteq T_0$.
  In particular,
  \[
    h_i(h'_i(g_i(T))) \,\subseteq\, h_i(g_i(T)) \,\subseteq\, T_0.
  \]
  Since $\compose{g_i}{\compose{h'_i}{h_i}} \in \val_C(T)$
  and $T_0 \in \minrep_C(T)$,
  we have $h_i(h'_i(g_i(T))) = T_0$, and hence,
  \[
    h_i(T_i) \,\stackrel{\eqref{eq:T_i}}{=}\, h_i(h'_i(g_i(T)))\, =\, T_0.
  \]
  Let $\tilde{h}_i$ be the restriction of $h_i$ to $\dom(T_i)$.
  Then, clearly, $\tilde{h}_i$ is a homomorphism
  from $T_i$ to $T_0$ with $\tilde{h}_i(T_i) = T_0$.

  \step{3}
    {For each $A \in T_0$ there are $j \in \set{1,\dotsc,n}$ and $A' \in T_j$
      with $h_j(A') = A$ and $A' \isomorphic A$.}
  Let
  \[
    T^* \,\isdef\,
    \bigunion_{i=1}^n\, \bigl(g_i(B_i) \setminus (T \setminus B_i)\bigr).
  \]
  To prove \eqref{eq:universal-queries/basis/main-claim},
  it suffices to show that there is a mapping $r\colon \dom(T^*) \to \dom(T_0)$
  with
  \begin{enumerate}[(1)]
  \item\label{cond:universal-queries/basis/image}
    $r(T^*) = T_0$,
  \item\label{cond:universal-queries/basis/iso}
    $r(A') \isomorphic A'$ for each $A' \in T^*$, and
  \item\label{cond:universal-queries/basis/h_i}
    $r(A') = h_i(A')$ for each $i \in \set{1,\dotsc,n}$
    and each $A' \in g_i(B_i) \setminus (T \setminus B_i)$.
  \end{enumerate}
  Indeed, let $A \in T_0$.
  Since $r(T^*) = T_0$ by condition~\ref{cond:universal-queries/basis/image},
  there is some $A' \in T^*$ with $r(A') = A$.
  So, by the construction of $T^*$, there is an $i \in \set{1,\dotsc,n}$
  such that $A' \in g_i(B_i) \setminus (T \setminus B_i) \subseteq T_i$.
  Condition~\ref{cond:universal-queries/basis/h_i} then yields
  $h_i(A') = r(A') = A$,
  and since $r(A') \isomorphic A'$
  by condition~\ref{cond:universal-queries/basis/iso},
  we have $A' \isomorphic A$.

  Define $r\colon \dom(\bigunion_{i=1}^n f_i(B_i)) \to \dom(T_0)$ such that
  \[
    r(u) = r_i(u)
    \quad
    \text{for all $i \in \set{1,\dotsc,n}$ and $u \in \dom(f_i(B_i))$}.
  \]
  This is well-defined,
  since $\nulls(f_i(B_i)) \intersection \nulls(f_j(B_j)) = \emptyset$
  for all distinct $i,j \in \set{1,\dotsc,n}$,
  and each $r_i$ is the identity on constants.
  We claim that $r$ satisfies
  conditions~\ref{cond:universal-queries/basis/image}--%
  \ref{cond:universal-queries/basis/h_i} above.

  To see that $r$ satisfies condition~\ref{cond:universal-queries/basis/iso},
  let $A' \in T^*$.
  Then there is some $i \in \set{1,\dotsc,n}$
  with $A' \in g_i(B_i) \setminus (T \setminus B_i)$.
  By \eqref{eq:universal-queries/basis/iso}
  and \eqref{eq:universal-queries/basis/g_i_T},
  we thus have $r(A') = r_i(A') \isomorphic A'$.

  To see that $r$ satisfies condition~\ref{cond:universal-queries/basis/h_i},
  let $i \in \set{1,\dotsc,n}$ and
  $A' \in g_i(B_i) \setminus (T \setminus B_i)$.
  Then, $r(A') = r_i(A') = h_i(A')$,
  where the last equality follows from the construction of $h_i$.

  It thus remains to show that
  $r$ satisfies condition~\ref{cond:universal-queries/basis/image},
  that is, $r(T^*) = T_0$.
  Note that, by \eqref{eq:universal-queries/basis/g_i_T}, we have
  \begin{align*}
    T^* \,\subseteq\, \bigunion_{i=1}^n f_i(B_i).
  \end{align*}
  Hence,
  \begin{align*}
    r(T^*)
    \,\subseteq\,
    r\left(\bigunion_{i=1}^n f_i(B_i)\right)
    \,=\,
    \bigunion_{i=1}^n r(f_i(B_i))
    \,=\,
    \bigunion_{i=1}^n r_i(f_i(B_i))
    \,\stackrel{\eqref{eq:universal-queries/basis/proj}}{=}\,
    \bigunion_{i=1}^n f(B_i)
    \,=\,
    T_0.
  \end{align*}
  To show that $r(T^*) = T_0$,
  we show that there is some $f^* \in \val_C(T)$ with $f^*(T) = T^*$.
  Then, $f' \isdef \compose{f^*}{r} \in \val_C(T)$.
  Since $T_0 \in \minrep_C(T)$ and $f'(T) = r(f^*(T)) = r(T^*) \subseteq T_0$,
  this implies $r(T^*) = T_0$, and the proof is complete.

  Thus, it remains to show that there is a mapping $f^* \in \val_C(T)$
  with $f^*(T) = T^*$.
  Basically, $f^*$ is obtained by repeated application of the mappings
  $g_1,\dotsc,g_n$.

  Let us first modify $g_1,\dotsc,g_n$ as follows.
  Choose an arbitrary ``renaming'' of the nulls of $T$.
  That is, pick an injective mapping
  $\rho\colon \dom(T) \to \const(T) \union (\Nulls \setminus \nulls(T))$
  such that $\rho(c) = c$ for each constant $c \in \const(T)$.
  Note that $\rho$ maps each null of $T$ to a unique null that does not
  occur in $T$.
  Let
  \[
    \CX \,\isdef\, \rho(\nulls(T)).
  \]
  For each $i \in \set{1,\dotsc,n}$ then define
  \(
    \hat{g}_i\colon
    \dom(T) \union C \union \CX \to \dom(T) \union C \union \CX
  \)
  such that for each $u \in \dom(T) \union C \union \CX$,
  \begin{align*}
    \hat{g}_i(u)
    \,\isdef\,
    \begin{cases}
      g_i(\rho^{-1}(u)),
      & \text{if $u \in \rho(\nulls(B_i))$ and
        $g_i(\rho^{-1}(u)) \in \nulls(B_i)$} \\
      \rho(g_i(\rho^{-1}(u)),
      & \text{if $u \in \rho(\nulls(B_i))$ and
        $g_i(\rho^{-1}(u)) \notin \nulls(B_i)$} \\
      u,
      & \text{otherwise.}
    \end{cases}
  \end{align*}
  Note that for each $i \in \set{1,\dotsc,n}$,
  \begin{align}
    \label{eq:g_i}
    \hat{g}_i(\rho(B_i)) \setminus \rho(T)
    \,=\,
    g_i(B_i) \setminus (T \setminus B_i).
  \end{align}

  Now let
  \[
    \hat{g} \,\isdef\,
    \compose{\hat{g}_1}{\compose{\hat{g}_2}{\compose{\dotsb}{\hat{g}_n}}}.
  \]
  Recall the graph $G$ mentioned at the beginning of the proof,
  on page~\pageref{def:universal-queries/basis/G}.
  Then an application of $\hat{g}$ to an atom $\rho(A)$ with $A \in B_{i_1}$
  corresponds to following the maximal path in $G$
  that starts in $A$ and proceeds
  to atoms $A' \in B_{i_2},A'' \in B_{i_3},\dotsc$
  with $i_1 < i_2 < i_3 < \dotsb$.
  If $A''' \in B_j$ is the endpoint of this path,
  then either
  $g_j(A''') \notin T \setminus B_j$
  and $\hat{g}(\rho(A)) = \hat{g}_j(\rho(A''')) = g_j(A''')$,
  or
  $g_j(A''') \in T \setminus B_j$
  and $\hat{g}(\rho(A)) = \rho(A''')$.

  For each $s \geq 0$ let
  \[
    \hat{g}^s \,\isdef\,
    \begin{cases}
      \rho, & \text{if $s = 0$}, \\
      \compose{\hat{g}^{s-1}}{\hat{g}}, & \text{if $s \geq 1$.}
    \end{cases}
  \]
  We show by induction that
  \begin{align}
    \label{eq:universal-queries/basis/descending}
    \hat{g}^1(T) \,\supseteq\, \hat{g}^2(T) \,\supseteq\,
    \hat{g}^3(T) \,\supseteq\, \dotsb.
  \end{align}
  To prove $\hat{g}^1(T) \supseteq \hat{g}^2(T)$,
  let $A \in \hat{g}^2(T)$.
  If $A \in T^*$, then by \eqref{eq:g_i}, we have $A \in \hat{g}^1(T)$.
  Otherwise, if $A \in \rho(T)$,
  there is an $A' \in \hat{g}^1(T)$ with $\hat{g}(A') = A$ and $A' \in \rho(T)$.
  Since $A' \in \rho(T)$, we have $A \in \hat{g}(\rho(T)) = \hat{g}^1(T)$,
  as desired.
  To prove $\hat{g}^{i+2}(T) \supseteq \hat{g}^{i+1}(T)$ for $i \geq 1$,
  let $A \in \hat{g}^{i+2}(T)$.
  Then there is an $A' \in \hat{g}^{i+1}(T)$ with $\hat{g}(A') = A$.
  Since $\hat{g}^{i+1}(T) \subseteq \hat{g}^i(T)$ by the induction hypothesis,
  we have $A \in \hat{g}(\hat{g}^i(T)) = \hat{g}^{i+1}(T)$,
  as desired.

  By \eqref{eq:universal-queries/basis/descending}
  and since $\hat{g}^1(T)$ is finite,
  there is an $s_0 \geq 1$ such that
  $\hat{g}^{s_0}(T) = \hat{g}^s(T)$ for each $s \geq s_0$.
  Let $f^* \isdef \hat{g}^{s_0}$.
  We show that $f^*(T) = T^*$.

  First observe that
  \[
    T^*
    \,=\,
    \bigunion_{i=1}^n\, \bigl(g_i(B_i) \setminus (T \setminus B_i)\bigr)
    \,=\,
    \hat{g}^1(T) \setminus \rho(T)
    \,=\,
    f^*(T) \setminus \rho(T).
  \]
  To see that $T^*$ is not a proper subinstance of $f^*(T)$,
  we show that $f^*(T)$ contains no atoms from $\rho(T)$.

  For a contradiction, suppose that $f^*(T)$ contains an atom $A \in \rho(T)$.
  Then, $A \in \rho(B_i)$ for some $i \in \set{1,\dotsc,n}$.
  Since $\hat{g}(f^*(T)) = f^*(T)$,
  we know that $\hat{g}$ is a bijection on $\dom(f^*(T))$.
  Furthermore, since $\hat{g}$ is the identity on $\dom(T) \union C$,
  we have $\hat{g}(\bot) \in \CX$ for each $\bot \in \CX$.
  It follows that
  \begin{align*}
    \hat{g}_i(A) \isomorphic A,
    \quad \text{and} \quad
    \hat{g}_i(A) \in \rho(B_j)
    \quad \text{for some $j \in \set{1,\dotsc,n} \setminus \set{i}$}.
  \end{align*}
  Let $A' \isdef \rho^{-1}(A)$.
  By the construction of $\hat{g}_i$, we have
  \begin{align*}
    g_i(A') \isomorphic A',
    \quad \text{and} \quad
    g_i(A') \in B_j
    \quad \text{for some $j \in \set{1,\dotsc,n} \setminus \set{i}$}.
  \end{align*}
  Since $B_i$ is packed and $g_i$ maps each null in $A'$ to a null in $B_j$,
  each atom in $g_i(B_i)$ contains a null from $B_j$.
  Together with $g_i \in \val_C(T,B_i)$, this implies $g_i(B_i) \subseteq B_j$.
  In other words, $g_i$ is a homomorphism from $T$ to $T \setminus B_i$,
  which contradicts the fact that $T$ is a core.
  Consequently, we must have $f^*(T) = T^*$.
\end{proof}

\begin{cor}
  Let $T$ be an instance such that $T$ is a core,
  and each atom block of $T$ is packed.
  Let $C \subseteq \Const$, and let $A$ be an atom.
  Then the following statements are equivalent:
  \begin{iteMize}{$\bullet$}
  \item
    There is an instance in $\minrep_C(T)$
    that contains an atom isomorphic to $A$.
  \item
    There is an atom block $B$ of $T$
    such that some instance in $\minrep_C(T,B)$
    contains an atom isomorphic to $A$.
  \end{iteMize}
\end{cor}

\noindent The following polynomial time algorithm for deciding
whether $R(\tup{t})$ occurs in some minimal instance in $\rep(T)$
immediately suggests itself.
Let $C$ be the set of all constants in $\tup{t}$.
Consider each atom block $B$ of $T$, and each $T_0 \in \minrep_C(T,B)$ in turn,
and accept the input if and only if $R(\tup{t}) \in T_0$ for some $T_0$.
By Proposition~\ref{prop:minrep-computation},
the instances $T_0$ can be computed in polynomial time.

The following example shows that the proof
of Lemma~\ref{lemma:universal-queries/basis} fails
if $T$ contains atom blocks that are not packed.

\begin{exa}
  \label{ex:universal-queries/basis/failure}
  Let $E$ be a binary relation symbol,
  and consider the instance $T$ over $\set{E}$ with
  \begin{align*}
    E^T \,=\, \{
      (\bot_1,a), (\bot_1,b), (\bot_1,\bot'_1), (\bot'_1,c),
      (\bot_2,a), (\bot_2,b), (\bot_2,\bot'_2), (c,\bot'_2)
    \}.
  \end{align*}
  Note that $T$ is a core, and that $T$ has the two atom blocks
  \begin{align*}
    B_1 & \,=\, \set{E(\bot_1,a), E(\bot_1,b), E(\bot_1,\bot'_1), E(\bot'_1,c)},
        \ \text{and} \\
    B_2 & \,=\, \set{E(\bot_2,a), E(\bot_2,b), E(\bot_2,\bot'_2), E(c,\bot'_2)}.
  \end{align*}
  See Figure~\ref{fig:universal-queries/basis/failure/graph}
  for a graph representation of $T$ and its two atom blocks $B_1$ and $B_2$.
  \begin{figure}[ht]
    \centering
    \includeTikZ{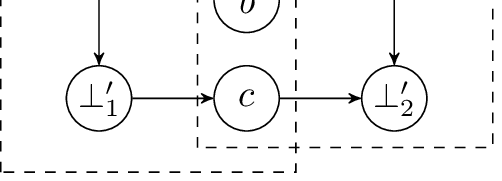}
    \caption{The instance $T$,
      and the two atom blocks $B_1$ and $B_2$ of $T$,
      which are the subinstances induced by the vertices
      in the corresponding dashed rectangles.}
    \label{fig:universal-queries/basis/failure/graph}
  \end{figure}
  Note also that neither $B_1$ nor $B_2$ is packed.

  Consider $f \in \val_\emptyset(T)$
  with $f(\bot_1) = f(\bot_2) = a$ and $f(\bot'_1) = f(\bot'_2) = b$.
  Then it is not hard to see that
  \[
    f(T) \,=\, \set{E(a,a),E(a,b),E(b,c),E(c,b)} \in \minrep_\emptyset(T).
  \]
  Furthermore, for the mappings $f_i$ created in the proof
  of Lemma~\ref{lemma:universal-queries/basis}, we have
  \begin{align*}
    f_1(T) \,=\, \{
    & E(a,a), E(a,b), E(b,c),
      E(\bot_2,a), E(\bot_2,b), E(\bot_2,\bot'_2), E(c,\bot'_2)
    \}
    \intertext{and}
    f_2(T) \,=\, \{
    & E(a,a), E(a,b), E(c,b),
      E(\bot_1,a), E(\bot_1,b),E(\bot_1,\bot'_1),E(\bot'_1,c)
    \}.
  \end{align*}

  For $g_i \in \val_\emptyset(T,B_i)$ with $g_i(\bot_i) = \bot_{3-i}$ and
  $g_i(\bot'_i) = b$, it holds that $g_i \in \minval_\emptyset(T,B_i)$,
  and moreover,
  \begin{align*}
    g_1(T)
    & \,=\,
      \{E(\bot_2,a), E(\bot_2,b), E(\bot_2,\bot'_2), E(c,\bot'_2), E(b,c)\}
      \in \minrep_\emptyset(T), \\
    g_2(T)
    & \,=\,
      \{E(\bot_1,a), E(\bot_1,b), E(\bot_1,\bot'_1), E(\bot'_1,c), E(c,b)\}
      \in \minrep_\emptyset(T).
  \end{align*}
  Note that $E(a,a)$ and $E(a,b)$ occur in $f_i(B_i)$,
  but neither $g_1(T)$ nor $g_2(T)$ contains $E(a,a)$ or $E(a,b)$.
\end{exa}

\subsubsection{Proof of Theorem~\ref{thm:universal-queries/bool}}
\label{sec:XGCWA-answers/universal/proof}

This section finally proves Theorem~\ref{thm:universal-queries/bool}.
Let $M = \Des$ be a schema mapping, where $\Sigma$ consists of packed st-tgds,
and let $q(\tup{x})$ be a universal query over $\tau$.
We show that there is a polynomial time algorithm that,
given as input an instance $T \isdef \core(M,S)$
for some source instance $S$ for $M$,
and a tuple $\tup{t} \in \Const^{\length{\tup{x}}}$,
decides whether $\tup{t} \in \certXGCWA{M}{S}{q}$.

As shown in Section~\ref{sec:XGCWA-answers/universal/core},
we can assume that $\tup{t}$ is a tuple over $\const(T) \union \dom(q)$,
and that in this case we have $\tup{t} \notin \certXGCWA{M}{S}{q}$
if and only if there is a nonempty finite set $\CT$
of minimal instances in $\rep(T)$
such that $\bigunion \CT \models \lnot q(\tup{t})$.
Now observe that $\lnot q$ is logically equivalent to a query $\bar{q}$
of the form
\begin{align*}
  \bar{q}(\tup{x}) = \biglor_{i=1}^m  q_i(\tup{x}),
\end{align*}
where each $q_i$ is an existential query of the form
\begin{align*}
  q_i(\tup{x}) = \exists \tup{y}_i \bigland_{j=1}^{n_i} \phi_{i,j},
\end{align*}
and each $\phi_{i,j}$ is an atomic FO formula or the negation
of an atomic FO formula.
Indeed, since $q$ is a universal query, we have
\(
  \lnot q \equiv \exists \tup{y}\, \phi(\tup{x},\tup{y}),
\)
where $\phi$ is quantifier-free.
By transforming $\phi$ into ``disjunctive normal form'',
we obtain a query of the form
\(
  \exists \tup{y} \biglor_{i=1}^m \bigland_{j=1}^{n_i} \phi_{i,j},
\)
where each $\phi_{i,j}$ is an atomic FO formula
or the negation of an atomic FO formula.
By moving existential quantifiers inwards, we finally obtain $\bar{q}$.
It remains therefore to decide whether there is some $i \in \set{1,\dotsc,m}$
and a nonempty finite set $\CT$ of minimal instances in $\rep(T)$
such that $\bigunion \CT \models q_i(\tup{t})$.

Fix some constant $\bs$ as in Lemma~\ref{lemma:bounded-blocksize}.
Then, for each atom block $B$ of $\core(M,S)$,
we have $\card{\nulls(B)} \leq \bs$.
Furthermore,
Proposition~\ref{prop:packed} tells us that each atom block of $\core(M,S)$
is packed.
We can now use the following algorithm to decide,
given as input an instance $T \isdef \core(M,S)$
for some source instance $S$ for $M$,
and a tuple $\tup{t} \in \Const^{\length{\tup{x}}}$,
whether $\tup{t} \in \certXGCWA{M}{S}{q}$:
\begin{enumerate}[(1)]
\item\label{algo:universal/check-blocks}
  Determine the atom blocks of $T$
  and check whether each atom block $B$ of $T$ is packed and satisfies
  $\card{\nulls(B)} \leq \bs$;
  if not, reject the input.
\item\label{algo:universal/check-core}
  Check whether $T$ is a core;
  if not, reject the input.
\item\label{algo:universal/loop}
  For each $i \in \set{1,\dotsc,m}$:
  \begin{enumerate}[(a)]
  \item\label{algo:universal/find}
    Check whether there is a nonempty finite set $\CT$ of minimal instances
    in $\rep(T)$ such that $\bigunion \CT \models q_i(\tup{t})$.
  \item\label{algo:universal/not-found}
    If such a $\CT$ exists, reject the input.
  \end{enumerate}
\item\label{algo:universal/accept}
  Accept the input.
\end{enumerate}
Step~\ref{algo:universal/check-blocks} clearly runs in polynomial time,
and step~\ref{algo:universal/check-core} can be implemented in polynomial time
using the algorithm from Lemma~\ref{lem:block-core}
(that algorithm outputs $T$ if and only if $T$ is a core).
Lemma~\ref{lemma:universal-queries/combine} below
 tells us that step~\ref{algo:universal/find} can be implemented
in polynomial time as well.
Thus, once Lemma~\ref{lemma:universal-queries/combine} is proved,
the proof of Theorem~\ref{thm:universal-queries/bool} is complete.

\begin{lem}
  \label{lemma:universal-queries/combine}
  Let $q(\tup{x}) = \exists \tup{y}\, \phi(\tup{x},\tup{y})$
  be a FO query over $\tau$,
  where $\phi = \bigland_{i=1}^p \phi_i$,
  and each $\phi_i$ is an atomic FO formula
  or the negation of an atomic FO formula.
  For each positive integer $\bs$,
  there is a polynomial time algorithm that decides:
  \problem{$\CoreEval_{\tau,\bs}$}
    {an instance $T$ over $\tau$
     such that $T$ is a core
     and each atom block of $T$ is packed and contains at most $\bs$ nulls;
     and a tuple $\tup{t} \in \Const^{\length{\tup{x}}}$}
    {Is there a nonempty finite set $\CT$ of minimal instances
     in $\rep(T)$ such that $\bigunion \CT \models q(\tup{t})$?}
\end{lem}

The remaining part of this section is devoted to a proof of
Lemma~\ref{lemma:universal-queries/combine}.

\medskip

Let $q(\tup{x}) = \exists \tup{y}\, \phi(\tup{x},\tup{y})$
be as in the hypothesis of Lemma~\ref{lemma:universal-queries/combine}.
Without loss of generality, there is no variable that occurs both
in $\tup{x}$ and in $\tup{y}$,
and $\phi$ has the form $\bigland_{i=1}^p \phi_i$,
where each $\phi_i$ is a relational atomic FO formula,
the negation of a relational atomic FO formula,
or the negation of an equality.
Let $\bs$ be a positive integer.

Suppose we are given an instance $T$ over $\tau$,
where $T$ is a core, each atom block of $T$ is packed,
and each atom block of $T$ contains at most $\bs$ nulls,
and a tuple $\tup{t} \in \Const^{\length{\tup{x}}}$.
In a first step, we rewrite $\phi$ to a formula $\psi$
by replacing each variable $x$ in $\tup{x}$
with the corresponding constant assigned to $x$ by $\tup{t}$.
That is, if $\tup{x} = (x_1,\dotsc,x_k)$ and $\tup{t} = (t_1,\dotsc,t_k)$,
then $\psi$ is obtained from $\phi$ by replacing,
for each $i \in \set{1,\dotsc,k}$,
each occurrence of the variable $x_i$ in $\phi$ by $t_i$.
Let
\[
  \tilde{q} \isdef \exists \tup{y}\, \psi(\tup{y}).
\]
To check whether there is a nonempty finite set $\CT$
of minimal instances in $\rep(T)$
such that $\bigunion \CT \models q(\tup{t})$,
it suffices to check whether there is a nonempty finite set $\CT$
of minimal instances in $\rep(T)$
such that $\bigunion \CT \models \tilde{q}$.

Suppose that $\psi$ has the form
\begin{align*}
  \psi(\tup{y}) =
  \bigland_{i=1}^k R_i(\tup{x}_i) \land
  \bigland_{i=1}^l \lnot Q_i(\tup{w}_i) \land
  \bigland_{i=1}^m \lnot v_i = v'_i.
\end{align*}
Let $C$ be the set of constants that occur in $\psi$,
and for each $i \in \set{1,\dotsc,k}$,
let $X_i$ be the set of all variables in $\tup{x}_i$.
Given an assignment $\alpha$ for a set $X$ of variables,
and a tuple $\tup{t}$ over $X \union \Const$,
we sloppily write $\alpha(\tup{t})$ for the tuple obtained from $\tup{t}$
by replacing each occurrence of each variable $x \in X$ in $\tup{t}$
with $\alpha(x)$.

The idea for finding a nonempty finite set $\CT$
of minimal instances in $\rep(T)$ with $\bigunion \CT \models \tilde{q}$
is as follows.
In the first step, we compute, for each $i \in \set{1,\dotsc,k}$,
the set of all pairs $(T_i,\alpha_i)$
such that $T_i \in \minrep_C(T,B)$ for some atom block $B$ of $T$,
and $\alpha_i(\tup{x}_i) \in R_i^{T_i}$.
Thus, modulo renaming of values that do not occur in $\const(T) \union C$,
we enumerate the possible assignments $\alpha_i$ of $\tup{x}_i$
under which $R_i(\tup{x}_i)$ is satisfied in some minimal instance
in $\rep(T)$;
the instance $T_i$ can then be considered as a witness to this fact.
In the second step,
we try to join the pairs $(T_1,\alpha_1),\dotsc,(T_k,\alpha_k)$,
where each $(T_i,\alpha_i)$ is a pair computed for $i$ in the first step,
to a single pair $(\tilde{T},\tilde{\alpha})$
such that $\tilde{T}$ satisfies each $R_i(\tup{x}_i)$
under the assignment $\tilde{\alpha}$.
The instance $\tilde{T}$ will actually be the union of isomorphic copies
$\rho_1(T_1),\dotsc,\rho_k(T_k)$ of the instances $T_1,\dotsc,T_k$.
In particular,
the set $\set{\rho_1(T_1),\dotsc,\rho_k(T_k)}$ is already close
to the desired set $\CT$:
it is a finite set of instances from $\min_C(T)$,
and its union satisfies the subformula $\bigland_{i=1}^k R_i(\tup{x}_i)$
of $\tilde{q}$ under $\tilde{\alpha}$.
Not all pairs $(T_1,\alpha_1),\dotsc,(T_k,\alpha_k)$ can be joined together.
We will join only pairs that are \emph{compatible}
in the sense of Definition~\ref{def:compatible} below.
By taking care in how those pairs are joined together,
and using Lemma~\ref{lemma:universal-queries/basis},
we can show that the desired set $\CT$ exists
if and only if the instance obtained from $\tilde{T}$
by adding a large enough, but constant, number of isomorphic copies of $T$
to $\tilde{T}$ satisfies $\tilde{q}$.

More precisely, the algorithm proceeds as follows.
Fix the constant
\[
  s \isdef k + \sum_{i=1}^l\, \length{\tup{w}_i} + 2 \cdot m.
\]
In the above description,
$s-k$ is the number of isomorphic copies of $T$
that will be added to $\tilde{T}$.
For each $i \in \set{1,\dotsc,s}$,
pick an injective mapping
\[
  \rho_i\colon \dom(T) \union C \to \Dom
\]
such that $\rho_i(c) = c$ for each $c \in \const(T) \union C$,
$\rho_i(\bot) \in \Nulls$ for each $\bot \in \nulls(T)$,
and such that for all distinct $i,j \in \set{1,\dotsc,s}$,
we have $\nulls(\rho_i(T)) \intersection \nulls(\rho_j(T)) = \emptyset$.
Then compute, for each $i \in \set{1,\dotsc,k}$, the set
\begin{align}
  \label{eq:Xi-def}
  \begin{aligned}
  \CX_i \,\isdef\,
  \Bigl\{
  (T_0,\alpha) \mid\
  & \text{there is some $T_0' \in \minrep_C(T,B)$
    and an atom block $B$ of $T$} \\
  & \text{such that $T_0 = \rho_i(T_0')$,
    $\alpha\colon X_i \to \dom(T_0)$, and $\alpha(\tup{x}_i) \in R_i^{T_0}$}
  \Bigr\}.
  \end{aligned}
\end{align}
Now we would like to join pairs
$(T_1,\alpha_1) \in \CX_1,\dotsc,(T_k,\alpha_k) \in \CX_k$
into single pairs $(\tilde{T},\tilde{\alpha})$
according to the above description.
However, we would like to do this only if the pairs
$(T_1,\alpha_1),\dotsc,(T_k,\alpha_k)$ are \emph{compatible}
in the following sense.
Intuitively, $(T_1,\alpha_1),\dotsc,(T_k,\alpha_k)$ are compatible
if the nulls in the image of each $\alpha_i$ can be consistently renamed
such that the resulting mappings $\tilde{\alpha}_i$ agree on common variables.

\begin{defi}[compatible]
  \label{def:compatible}
  We say that $(T_1,\alpha_1) \in \CX_1,\dotsc,(T_k,\alpha_k) \in \CX_k$
  are \emph{compatible}
  if there is an equivalence relation $\sim$
  on $D \isdef \bigunion_{i=1}^k \alpha_i(X_i)$
  such that
  \begin{enumerate}[(1)]
  \item\label{def:compatible/cons}
    for all $i,j \in \set{1,\dotsc,k}$ and $x \in X_i \intersection X_j$,
    we have $\alpha_i(x) \sim \alpha_j(x)$,
  \item\label{def:compatible/const}
    for all $u,u' \in D$, if $u \sim u'$ and $u \in \Const$, then $u = u'$,
    and
  \item\label{def:compatible/eq}
    for all $i \in \set{1,\dotsc,k}$ and $x,x' \in X_i$,
    we have $\alpha_i(x) \sim \alpha_i(x')$ if and only if
    $\alpha_i(x) = \alpha_i(x')$.
  \end{enumerate}
\end{defi}

\begin{prop}
  \label{prop:XGCWA-answers/universal/composition/compat}
  There is an algorithm that,
  given $(T_1,\alpha_1) \in \CX_1,\dotsc,(T_k,\alpha_k) \in \CX_k$ as input,
  decides in time linear in the size of $T$
  whether $(T_1,\alpha_1),\dotsc,(T_k,\alpha_k)$ are compatible,
  and if so, outputs an equivalence relation $\sim$ on
  $D \isdef \bigunion_{i=1}^k \alpha_i(X_i)$
  that satisfies conditions~\ref{def:compatible/cons}--\ref{def:compatible/eq}
  of Definition~\ref{def:compatible}.
  In fact, $\sim$ is the smallest such equivalence relation
  (with respect to set inclusion).
\end{prop}

\begin{proof}
  Given $(T_1,\alpha_1) \in \CX_1,\dotsc,(T_k,\alpha_k) \in \CX_k$,
  the following algorithm computes the desired relation $\sim$
  if it exists:
  \begin{enumerate}[(1)]
  \item\label{alg:XGCWA-answers/universal/composition/compat/init}
    Initialize $\sim$ to be
    $\set{(u,u) \mid
      \text{$u \in \alpha_i(X_i)$ for some $i \in \set{1,\dotsc,k}$}}$.
  \item\label{alg:XGCWA-answers/universal/composition/compat/eq1}
    For all $i,j \in \set{1,\dotsc,k}$ and $x \in X_i \intersection X_j$,
    add $(\alpha_i(x),\alpha_j(x))$ to $\sim$.
  \item\label{alg:XGCWA-answers/universal/composition/compat/eq2}
    For all $i \in \set{1,\dotsc,k}$ and $x,x' \in X_i$
    with $\alpha_i(x) = \alpha_i(x')$,
    add $(\alpha_i(x),\alpha_i(x'))$ to $\sim$.
  \item\label{alg:XGCWA-answers/universal/composition/compat/closure}
    Update $\sim$ to be the symmetric and transitive closure of $\sim$.
  \item\label{alg:XGCWA-answers/universal/composition/compat/check}
    If $\sim$ satisfies conditions \ref{def:compatible/const} and
    \ref{def:compatible/eq} of Definition~\ref{def:compatible},
    then output $\sim$;
    otherwise output ``not compatible''.
  \end{enumerate}
  Since $k$ and $X_1,\dotsc,X_k$ are constant,
  it should be clear that each of the steps
  \ref{alg:XGCWA-answers/universal/composition/compat/init}--%
  \ref{alg:XGCWA-answers/universal/composition/compat/check}
  can be accomplished in constant time,
  after building the necessary data structures from the input
  in time linear in the size of $T$
  (note that each $T_i$ is at most as large as $T$,
   so that the length of the input is linear in the size of $T$).

  It is now not hard to see that if the algorithm outputs a relation $\sim$,
  then $\sim$ is an equivalence relation
  on $D \isdef \bigunion_{i=1}^k \alpha_i(X_i)$
  that satisfies conditions~\ref{def:compatible/cons}--\ref{def:compatible/eq}
  of Definition~\ref{def:compatible}.
  In particular, $(T_1,\alpha_1),\dotsc,(T_k,\alpha_k)$ are compatible.
  Even more, $\sim$ is the smallest such equivalence relation,
  since every equivalence relation $\sim^*$ on $D$
  that satisfies conditions~\ref{def:compatible/cons}--\ref{def:compatible/eq}
  of Definition~\ref{def:compatible}
  must contain the pairs put into $\sim$
  in steps~\ref{alg:XGCWA-answers/universal/composition/compat/init}--%
  \ref{alg:XGCWA-answers/universal/composition/compat/closure}
  of the algorithm.
  The same argument shows that the algorithm outputs a relation $\sim$
  if there is an equivalence relation $\sim^*$ on $D$
  that satisfies conditions~\ref{def:compatible/cons}--\ref{def:compatible/eq}
  of Definition~\ref{def:compatible},
  that is, if $(T_1,\alpha_1),\dotsc,(T_k,\alpha_k)$ are compatible.
\end{proof}

We now define the join of compatible pairs
$(T_1,\alpha_1) \in \CX_1,\dotsc,(T_k,\alpha_k) \in \CX_k$.
Given an equivalence relation $\sim$
on $D \isdef \bigunion_{i=1}^k \alpha_i(X_i)$
as in Definition~\ref{def:compatible},
the idea of the join is to identify values $u,u' \in D$ with $u \sim u'$,
and to ``glue''
the resulting instances $\tilde{T}_i$ and assignments $\tilde{\alpha}_i$
together to a single instance $\tilde{T}$ and assignment $\tilde{\alpha}$.

\begin{defi}[Join]
  \label{def:join}
  Let $(T_1,\alpha_1) \in \CX_1,\dotsc,(T_k,\alpha_k) \in \CX_k$
  be compatible,
  and let $\sim$ be the smallest equivalence relation
  on $D \isdef \bigunion_{i=1}^k \alpha_i(X_i)$
  that satisfies conditions~\ref{def:compatible/cons}--\ref{def:compatible/eq}
  of Definition~\ref{def:compatible}.
  Pick some linear order $\preceq$ on the elements of $D$,
  and for each $u \in D$,
  let $\hat{u}$ be the minimal element in
  $[u] \isdef \set{u' \in D \mid u' \sim u}$
  with respect to $\preceq$.
  For each $i \in \set{1,\dotsc,k}$,
  define $r_i\colon \dom(T_i) \to \Dom$ such that for each $u \in \dom(T_i)$,
  \begin{align*}
    r_i(u) \isdef
    \begin{cases}
      \hat{u}, & \text{if $u \in \alpha_i(X_i)$}, \\
      u, & \text{otherwise}.
    \end{cases}
  \end{align*}
  Then \emph{the join of $(T_1,\alpha_1),\dotsc,(T_k,\alpha_k)$}
  is the pair $(\tilde{T},\tilde{\alpha})$,
  where $\tilde{T} \isdef \bigunion_{i=1}^k r_i(T_i)$,
  and $\tilde{\alpha}\colon \bigunion_{i=1}^k X_i \to \Dom$ is such that
  for each $x \in \bigunion_{i=1}^k X_i$,
  \begin{align*}
    \tilde{\alpha}(x) \isdef
    \begin{cases}
      r_1(\alpha_1(x)), & \text{if $x \in X_1$} \\
      \quad\ \ \vdots &  \\
      r_k(\alpha_k(x)), & \text{if $x \in X_k$}.
    \end{cases}
  \end{align*}
\end{defi}

Note that different choices of $\preceq$ yield different joins.
For definiteness, we can generate $\preceq$ as follows.
Initialize $\preceq$ to be the empty relation.
For increasing $i = 1,2,\dotsc,k$, consider the variables $x \in X_i$
in some predefined fixed order,
and if $u \isdef \alpha_i(x)$ does not already occur in $\preceq$,
add $u$ as the new maximal element to $\preceq$.
This takes constant time, since $k$ and $X_1,\dotsc,X_k$ are fixed.
For the following construction,
it is not important that the join always yields the same result --
the join resulting from any linear ordering $\preceq$ on $D$ is fine.
What is important are the properties summarized
in Proposition~\ref{prop:XGCWA-answers/universal/composition/join} below.
Note also that modulo the choice of $\preceq$,
$\tilde{\alpha}$ is well-defined
by the construction of $\sim$ and $r_1,\dotsc,r_k$:
if $x \in X_i \intersection X_j$, then $\alpha_i(x) \sim \alpha_j(x)$,
and thus, $r_i(\alpha_i(x)) = r_j(\alpha_j(x))$.

\begin{prop}
  \label{prop:XGCWA-answers/universal/composition/join}
  The join $(\tilde{T},\tilde{\alpha})$ of compatible pairs
  $(T_1,\alpha_1) \in \CX_1,\dotsc,(T_k,\alpha_k) \in \CX_k$
  can be computed in time linear in the size of $T$
  and has the following properties.
  Let $r_1,\dotsc,r_k$ be the mappings used in the construction
  of $(\tilde{T},\tilde{\alpha})$.
  Then for all $i,j \in \set{1,\dotsc,k}$:
  \begin{enumerate}[\em(1)]
  \item
    \label{prop:XGCWA-answers/universal/composition/join/r/r}
    For all $c \in \const(T_i)$ and $\bot \in \nulls(T_i)$,
    we have $r_i(c) = c$ and $r_i(\bot) \in \Nulls$.
  \item\label{prop:XGCWA-answers/universal/composition/join/r/rel2}
    Let $\sim$ be an equivalence relation on
    $D \isdef \bigunion_{i=1}^k \alpha_i(X_i)$ that satisfies
    conditions~\ref{def:compatible/cons}--\ref{def:compatible/eq}
    of Definition~\ref{def:compatible}.
    Then for all $u \in \dom(T_i)$ and $u' \in \dom(T_j)$,
    \begin{align*}
      r_i(u) = r_j(u')\ \Longrightarrow\
      \text{$u = u'$ or:
        $u \in \alpha_i(X_i)$, $u' \in \alpha_j(X_j)$ and $u \sim u'$}.
    \end{align*}
    Furthermore, if $\sim$ is the smallest such equivalence relation,
    then
    \begin{align*}
      r_i(u) = r_j(u') \iff
      \text{$u = u'$ or:
        $u \in \alpha_i(X_i)$, $u' \in \alpha_j(X_j)$ and $u \sim u'$}.
    \end{align*}
  \item\label{prop:XGCWA-answers/universal/composition/join/r/inj}
    $r_i$ is injective.
  \item\label{prop:XGCWA-answers/universal/composition/join/tuples}
    $\tilde{\alpha}(\tup{x}_i) \in R_i^{\tilde{T}}$.
  \end{enumerate}
\end{prop}

\begin{proof}
  Let us first see that $(\tilde{T},\tilde{\alpha})$
  can be computed in time $O(n)$,
  where $n$ is the size of $T$,
  given $(T_1,\alpha_1),\dotsc,(T_k,\alpha_k)$ as input.
  By Proposition~\ref{prop:XGCWA-answers/universal/composition/compat},
  the relation $\sim$ can be computed in time $O(n)$.
  Since $k$ and $X_1,\dotsc,X_k$ are fixed,
  the linear order $\preceq$ can be computed in constant time.
  Furthermore, since $k$ is constant and all $T_i$ have size at most $n$,
  the mappings $r_i$ and the join $(\tilde{T},\tilde{\alpha})$
  can be computed in time $O(n)$.
  We next prove \ref{prop:XGCWA-answers/universal/composition/join/r/r}--%
  \ref{prop:XGCWA-answers/universal/composition/join/tuples}.

  \ad{\ref{prop:XGCWA-answers/universal/composition/join/r/r}}
  Let $c \in \const(T_i)$.
  If $c \notin \alpha_i(X_i)$, then by the construction of $r_i$
  we have $r_i(c) = c$.
  Otherwise, if $c \in \alpha_i(X_i)$,
  there is some $x \in X_i$ with $\alpha_i(x) = c$,
  so that by the construction of $r_i$
  we have $c = \alpha_i(x) \sim r_i(\alpha_i(x)) = r_i(c)$.
  Condition~\ref{def:compatible/const} of Definition~\ref{def:compatible}
  then yields $r_i(c) = c$.
  Next let $\bot \in \nulls(T_i)$.
  As above, if $\bot \notin \alpha_i(X_i)$, then $r_i(\bot) = \bot \in \Nulls$.
  Otherwise, $\bot \sim r_i(\bot)$,
  so that by condition~\ref{def:compatible/const}
  of Definition~\ref{def:compatible},
  $r_i(\bot) \in \Nulls$.

  \ad{\ref{prop:XGCWA-answers/universal/composition/join/r/rel2}}
  Let $\sim$ be an equivalence relation on $D$
  that satisfies conditions~\ref{def:compatible/cons}--\ref{def:compatible/eq}
  of Definition~\ref{def:compatible}.
  We first prove the statement for the case that $\sim$ is the smallest
  such relation.
  That is, we have to show
  \begin{align}
    \label{eq:XGCWA-answers/universal/composition/join/eq}
    r_i(u) = r_j(u') \iff
    \text{$u = u'$ or:
      $u \in \alpha_i(X_i)$, $u' \in \alpha_j(X_j)$ and $u \sim u'$}.
  \end{align}

  We first prove the direction from right to left.
  Suppose that $u = u'$.
  If $i = j$, we have $r_i(u) = r_j(u')$.
  If $i \neq j$, then $\nulls(T_i) \intersection \nulls(T_j) = \emptyset$
  implies that $u$ and $u'$ are constants,
  and therefore $r_i(u) = u = u' = r_j(u')$
  by \ref{prop:XGCWA-answers/universal/composition/join/r/r}.
  Suppose next that $u \in \alpha_i(X_i)$, $u' \in \alpha_j(X_j)$
  and  $u \sim u'$.
  Pick $x \in X_i$ and $x' \in X_j$ such that
  $\alpha_i(x) = u$ and $\alpha_j(x') = u'$.
  Then, $\alpha_i(x) \sim \alpha_j(x')$,
  and the construction of $r_i$ and $r_j$
  immediately implies $r_i(u) = r_j(u')$.

  We next prove the direction from left to right.
  Let $r_i(u) = r_j(u')$.
  We distinguish the following cases:
  \begin{enumerate}[(a)]
  \item
    \label{case:XGCWA-answers/universal/composition/join/1}
    $u \notin \alpha_i(X_i) \intersection \Nulls$ and
    $u' \notin \alpha_j(X_j) \intersection \Nulls$.
  \item
    \label{case:XGCWA-answers/universal/composition/join/2}
    $u \in \alpha_i(X_i) \intersection \Nulls$ or
    $u' \in \alpha_j(X_j) \intersection \Nulls$.
  \end{enumerate}
  In case~(a), by the construction of $r_i,r_j$ and
  by \ref{prop:XGCWA-answers/universal/composition/join/r/r},
  we have $r_i(u) = u$ and $r_j(u') = u'$.
  Since $r_i(u) = r_j(u')$, this implies $u = u'$.

  So assume case~(b).
  By symmetry it suffices to deal with the case that
  $u \in \alpha_i(X_i) \intersection \Nulls$.
  By the construction of $r_i$, we then have
  \[
    u \sim r_i(u) = r_j(u').
  \]
  We claim that $u' \in \alpha_j(X_j)$.
  Suppose, to the contrary, that $u' \notin \alpha_j(X_j)$.
  By the construction of $r_j$, we have
  \[
    u' = r_j(u') = r_i(u) \sim u.
  \]
  Note that $u \in \alpha_i(X_i)$, $r_i(u) = u'$
  and the construction of $r_i$
  imply that $u' \in D$.
  Pick $p \in \set{1,\dotsc,k}$ and $x \in X_p$ with $\alpha_p(x) = u'$.
  By $u \in \Nulls$, $r_i(u) = u'$ and
  \ref{prop:XGCWA-answers/universal/composition/join/r/r},
  we have $u' \in \Nulls$.
  Moreover, since $u' \in \nulls(T_j)$, $u' = \alpha_p(X_p) \in \nulls(T_p)$,
  and $\nulls(T_j) \intersection \nulls(T_p) = \emptyset$ for $j \neq p$,
  we have $p = j$.
  This, however, implies that $u' \in \alpha_j(X_j)$,
  which is a contradiction to our assumption that $u' \notin \alpha_j(X_j)$.
  Hence, $u' \in \alpha_j(X_j)$.
  By the construction of $r_j$, we have
  \(
    u' \sim r_j(u') \sim u.
  \)
  In particular,
  $u \in \alpha_i(X_i)$, $u' \in \alpha_j(X_j)$ and $u \sim u'$,
  as desired.

  Finally, let $\sim^*$ be another equivalence relation on $D$
  that satisfies conditions~\ref{def:compatible/cons}--\ref{def:compatible/eq}
  of Definition~\ref{def:compatible}.
  We show that
  \begin{align*}
    r_i(u) = r_j(u')\ \Longrightarrow\
    \text{$u = u'$ or:
      $u \in \alpha_i(X_i)$, $u' \in \alpha_j(X_j)$ and $u \sim^* u'$}.
  \end{align*}
  Let $r_i(u) = r_j(u')$.
  By \eqref{eq:XGCWA-answers/universal/composition/join/eq},
  we have $u = u'$,
  or: $u \in \alpha_i(X_i)$, $u' \in \alpha_j(X_j)$ and $u \sim u'$.
  By minimality of $\sim$, $u \sim u'$ implies $u \sim^* u'$,
  so that $u = u'$,
  or: $u \in \alpha_i(X_i)$, $u' \in \alpha_j(X_j)$ and $u \sim^* u'$,
  as desired.

  \ad{\ref{prop:XGCWA-answers/universal/composition/join/r/inj}}
  Let $u,u' \in \dom(T_i)$ be such that $r_i(u) = r_i(u')$.
  We have to show that $u = u'$.
  By~\ref{prop:XGCWA-answers/universal/composition/join/r/rel2},
  we have $u = u'$, or: $u,u' \in \alpha_i(X_i)$ and $u \sim u'$.
  If $u = u'$, we are done.
  So assume that $u,u' \in \alpha_i(X_i)$ and $u \sim u'$.
  Let $x,x' \in X_i$ be such that $\alpha_i(x) = u$ and $\alpha_i(x') = u'$.
  Then $\alpha_i(x) \sim \alpha_i(x')$,
  and by condition~\ref{def:compatible/eq} of Definition~\ref{def:compatible},
  we have $u = \alpha_i(x) = \alpha_i(x') = u'$,
  as desired.

  \ad{\ref{prop:XGCWA-answers/universal/composition/join/tuples}}
  This follows immediately from the construction of $\tilde{T}$,
  $\tilde{\alpha}$,
  and \ref{prop:XGCWA-answers/universal/composition/join/r/r}.
\end{proof}

We can now give the algorithm for $\CoreEval_{\tau,\bs}$:

\begin{algo}[Main algorithm]
  \label{alg:XGCWA-answers/universal/composition}\mbox{}\\
  \textit{Input:} \hfill
  \parbox[t]{0.9\linewidth}{%
    an instance $T$ over $\tau$ that is a core
    and each atom block of $T$ is packed and
    contains at most $\bs$ nulls;
    a tuple $\tup{t} \in \Const^{\length{\tup{x}}}$
  }

  \smallskip

  \noindent
  \textit{Output:} \hfill
  \parbox[t]{0.9\linewidth}{%
    ``yes'' if there is a nonempty finite set $\CT$ of minimal
    instances in $\rep(T)$ such that $\bigunion \CT \models q(\tup{t})$;
    otherwise ``no''
  }

  \smallskip

  \begin{enumerate}[(1)]
  \item
    Compute $\tilde{q} = \exists \tup{y}\, \psi(\tup{y})$
    and choose $\rho_1,\dotsc,\rho_s$.
    (Recall that each $\rho_i$ is an injective mapping
    from $\dom(T) \union C$ to $\Dom$
    that is the identity on constants, maps nulls to nulls,
    and that $\nulls(\rho_i(T)) \intersection \nulls(\rho_j(T)) = \emptyset$
    for distinct $i,j$.)
  \item
    Compute the sets $\CX_1,\dotsc,\CX_k$ according to \eqref{eq:Xi-def}.
  \item
    For all $(T_1,\alpha_1) \in \CX_1,\dotsc,(T_k,\alpha_k) \in \CX_k$:
    \begin{enumerate}[(a)]
    \item
      Check whether $(T_1,\alpha_1),\dotsc,(T_k,\alpha_k)$ are compatible;
      \\
      if not, continue with next $(T_1,\alpha_1),\dotsc,(T_k,\alpha_k)$.
    \item
      Let $(\tilde{T},\tilde{\alpha})$ be the join of
      $(T_1,\alpha_1),\dotsc,(T_k,\alpha_k)$.
    \item
      If $\tilde{T} \union \bigunion_{i=k+1}^s \rho_i(T)$
      satisfies $\tilde{q}$,
      output ``yes''.
    \end{enumerate}
  \item
    Output ``no''.
  \end{enumerate}
\end{algo}

\noindent Let us now show that the algorithm decides $\CoreEval_{\tau,\bs}$
in polynomial time.
For a more precise upper bound on the algorithm's running time,
see \cite[Lemma~5.40]{Hernich:Diss10}.

\begin{lem}
  \label{lemma:XGCWA-answers/universal/composition/correctness}
  Algorithm~\ref{alg:XGCWA-answers/universal/composition}
  runs in time polynomial in the size of $T$.
  Furthermore, the following two statements are equivalent:
  \begin{enumerate}[\em(1)]
  \item\label{lemma:XGCWA-answers/universal/composition/correctness/2}
    There is a nonempty finite set $\CT$ of minimal instances
    in $\rep(T)$ such that $\bigunion \CT \models \tilde{q}$.
  \item\label{lemma:XGCWA-answers/universal/composition/correctness/1}
    Algorithm~\ref{alg:XGCWA-answers/universal/composition} outputs ``yes''
    on input $T$ and $\tup{t}$.
  \end{enumerate}
\end{lem}

\begin{proof}
  It is not hard to see that the algorithm runs in time polynomial
  in the size of $T$.
  Indeed, the transformation from $q$ to $\tilde{q}$
  can be accomplished in constant time (since $q$ is fixed),
  and the mappings $\rho_1,\dotsc,\rho_s$ can be generated
  in polynomial time.
  It is also not hard to compute the sets $\CX_1,\dotsc,\CX_k$
  in polynomial time:
  All we need to do in order to compute $\CX_i$
  for $i \in \set{1,\dotsc,k}$
  is to iterate through all $T'_0 \in \minrep_C(T,B)$,
  where $B$ is an atom block of $T$,
  and all assignments $\alpha\colon X_i \to \dom(\rho(T'_0))$,
  and to check whether $R_i(\alpha(\tup{x}_i)) \in \rho(T'_0)$.
  By Proposition~\ref{prop:minrep-computation},
  and since $X_i$ is fixed,
  this can be done in polynomial time.
  Since $k$ is constant,
  all the sets $\CX_1,\dotsc,\CX_k$ can thus be computed in polynomial time.
  In particular, since each of these sets has polynomial size,
  there are at most a polynomial number of iterations
  of the algorithm's main loop.
  Propositions~\ref{prop:XGCWA-answers/universal/composition/compat}
  and \ref{prop:XGCWA-answers/universal/composition/join}
  imply that steps~3(a) and 3(b) of the main loop run in polynomial time.
  Finally, step~3(c) clearly takes only a polynomial number of steps.
  Altogether, the algorithm runs in polynomial time.

  It remains to show that the two statements
  \ref{lemma:XGCWA-answers/universal/composition/correctness/2} and
  \ref{lemma:XGCWA-answers/universal/composition/correctness/1}
  are equivalent.

  \fromto{\ref{lemma:XGCWA-answers/universal/composition/correctness/1}}
         {\ref{lemma:XGCWA-answers/universal/composition/correctness/2}}
  Assume that Algorithm~\ref{alg:XGCWA-answers/universal/composition}
  outputs ``yes'' on input $T$ and $\tup{t}$.
  Then there are compatible
  $(T_1,\alpha_1) \in \CX_1,\dotsc,(T_k,\alpha_k) \in \CX_k$
  such that the join $(\tilde{T},\tilde{\alpha})$
  of $(T_1,\alpha_1),\dotsc,(T_k,\alpha_k)$
  has the following property:
  the instance $T^* \isdef \tilde{T} \union \bigunion_{i=k+1}^s \rho_i(T)$
  satisfies $\tilde{q}$.
  We construct a nonempty finite set $\CT$
  of minimal instances in $\rep(T)$
  with $\bigunion \CT \models \tilde{q}$.

  By Proposition~\ref{prop:XGCWA-answers/universal/composition/join},
  we have
  \(
    \tilde{T} = \bigunion_{i=1}^k r_i(T_i),
  \)
  where each $r_i$ is an injective mapping from $\dom(T_i)$ to $\Dom$
  with $r_i(c) = c$ for each $c \in \const(T_i)$,
  and $r_i(\bot) \in \Nulls$ for each $\bot \in \nulls(T_i)$.
  For each $i \in \set{k+1,\dotsc,s}$,
  let $T_i \isdef \rho_i(T)$,
  and let $r_i$ be the identity mapping on $\dom(T_i)$.
  Then,
  \begin{align}
    \label{eq:XGCWA-answers/universal/composition/1/T*}
    T^* \,=\, \bigunion_{i=1}^s r_i(T_i).
  \end{align}
  Note that each $T_i$ is isomorphic to an instance
  $\hat{T}_i \in \minrep_C(T)$.
  For $i \in \set{1,\dotsc,k}$,
  this follows from Lemma~\ref{lemma:minval}
  and the fact that $T_i$ is isomorphic to an instance in $\minrep_C(T,B)$
  for some atom block $B$ of $T$.
  For $i \in \set{k+1,\dotsc,s}$,
  this follows from the fact that $T_i = \rho_i(T)$, that $T$ is a core, and
  Proposition~\ref{prop:minrep-props}(\ref{prop:minrep-props/core-is-minimal}).
  For each $i \in \set{1,\dotsc,s}$,
  let $f_i$ be an isomorphism from $\hat{T}_i$ to $T_i$.

  Let
  \(
    v\colon \dom(T^*) \to \const(T^*) \union (\Const \setminus C)
  \)
  be an injective valuation of $T^*$,
  and for every $i \in \set{1,\dotsc,s}$
  let
  \[
    v_i \,\isdef\, \compose{f_i}{\compose{r_i}{v}}.
  \]
  Note that $v_i$ is an injective valuation of $\hat{T}_i$.
  To see this,
  note that $f_i$ is an injective mapping from $\dom(\hat{T}_i)$ to $\dom(T_i)$
  that is legal for $\hat{T}_i$,
  that $r_i$ is an injective mapping from $\dom(T_i)$ to $\dom(T^*)$
  that is legal for $T_i$,
  and that $v$ is an injective valuation of $T^*$.
  Furthermore, for each $\bot \in \nulls(\hat{T}_i)$
  we have $v_i(\bot) \notin C$,
  since both $f_i$ and $r_i$ map nulls to nulls,
  and $v$ maps nulls to constants in $\Const \setminus C$.
  In summary, $\hat{T}_i \in \minrep_C(T)$,
  $v_i$ is an injective valuation of $\hat{T}_i$,
  and $v_i^{-1}(c) = c$ for all $c \in \dom(v_i(\hat{T}_i)) \intersection C$.
  Together with
  Proposition~\ref{prop:minrep-props}(\ref{prop:minrep-props/capture}),
  this implies that
  $v_i(\hat{T}_i)$ is a minimal instance in $\rep(T)$.

  So,
  \[
    \CT \,\isdef\, \set{v_i(\hat{T}_i) \mid 1 \leq i \leq s}
  \]
  is a finite nonempty set of minimal instances in $\rep(T)$,
  and
  \[
    \bigunion \CT
    \,=\,
    \bigunion_{i=1}^s v_i(\hat{T}_i)
    \,=\,
    \bigunion_{i=1}^s v(r_i(T_i))
    \,=\,
    v\left(\bigunion_{i=1}^s r_i(T_i)\right)
    \,\stackrel{\eqref{eq:XGCWA-answers/universal/composition/1/T*}}{=}\,
    v(T^*).
  \]
  Since $T^* \models \tilde{q}$, $v$ is injective,
  and $v$ maps nulls in $T^*$ to constants that do not occur in $\tilde{q}$,
  we conclude that
  \(
    \bigunion \CT \models \tilde{q}.
  \)

  \fromto{\ref{lemma:XGCWA-answers/universal/composition/correctness/2}}
         {\ref{lemma:XGCWA-answers/universal/composition/correctness/1}}
  Assume that there is a nonempty finite set $\CT$
  of minimal instances in $\rep(T)$
  such that $\bigunion \CT \models \tilde{q}$.
  We show that Algorithm~\ref{alg:XGCWA-answers/universal/composition}
  outputs ``yes'' on input $T$ and $\tup{t}$.

  Since $\bigunion \CT \models \tilde{q}$,
  there is an assignment
  $\beta\colon \tup{y} \to \dom(\bigunion \CT) \union C$
  with
  \begin{align*}
    \bigunion \CT \,\models\, \psi(\beta).
  \end{align*}
  In particular, we can pick for each $i \in \set{1,\dotsc,k}$
  an instance $\hat{T}_i \in \CT$
  such that
  \[
    \beta(\tup{x}_i) \in R_i^{\hat{T}_i}.
  \]
  Note that there are at most $s-k$ values in $\beta(\tup{y}) \setminus C$
  that do not occur in $\beta(\tup{x}_i)$ for some $i \in \set{1,\dotsc,k}$.
  Thus, we can fix instances $\hat{T}_{k+1},\dotsc,\hat{T}_s \in \CT$
  such that each of the values in $\beta(\tup{y}) \setminus C$
  that does not occur in $\beta(\tup{x}_i)$ for some $i \in \set{1,\dotsc,k}$
  belongs to $\dom(\hat{T}_j)$ for some $j \in \set{k+1,\dotsc,s}$.
  Now $\beta$ is an assignment for $\psi$ with range in
  $\dom(\hat{T}_1 \union \dotsb \union \hat{T}_s) \union C$,
  and we have:
  \begin{align}
    \label{eq:XGCWA-answers/universal/composition/correctness/beta2}
    \bigunion_{i=1}^s \hat{T}_i \,\models\, \psi(\beta).
  \end{align}

  Let $i \in \set{1,\dotsc,k}$.
  By Proposition~\ref{prop:minrep-props}(\ref{prop:minrep-props/capture}),
  there is an instance $\tilde{T}_i \in \minrep_C(T)$
  and an injective valuation $v_i$ of $\tilde{T}_i$
  such that $v_i(\tilde{T}_i) = \hat{T}_i$, and
  $v_i^{-1}(c) = c$ for all $c \in \dom(\hat{T}_i) \intersection C$.
  In particular,
  \[
    A_i \,\isdef\, R_i\bigl(v_i^{-1}(\beta(\tup{x}_i))\bigr)
    \,\in\, \tilde{T}_i.
  \]
  By Lemma~\ref{lemma:universal-queries/basis},
  there is an atom block $B_i$ of $T$, an instance $T_i' \in \minrep_C(T,B_i)$,
  an atom $A''_i \in T_i'$ with $A''_i \isomorphic A_i$,
  and a homomorphism $h_i'$ from $T_i'$ to $\tilde{T}_i$
  such that $h_i'(T_i') = \tilde{T}_i$ and $h_i'(A''_i) = A_i$.
  In particular, $h_i \isdef \compose{\rho_i^{-1}}{h_i'}$
  is a homomorphism from $T_i \isdef \rho_i(T_i')$ to $\tilde{T}_i$ with
  \[
    h_i(T_i) \,=\, \tilde{T}_i
    \qquad \text{and} \qquad
    h_i(A'_i) \,=\, A_i,
  \]
  where $A'_i \isdef \rho_i(A_i'')\ \isomorphic\ A_i$.
  Let $\alpha_i$ be an
  assignment for $X_i$ such that
  \[
    A'_i \,=\, R_i(\alpha_i(\tup{x}_i)).
  \]
  Note that $(T_i,\alpha_i) \in \CX_i$.

  In the following,
  we show that $(T_1,\alpha_1),\dotsc,(T_k,\alpha_k)$ are compatible,
  and if $(\tilde{T},\tilde{\alpha})$ is the join of these pairs,
  then $\tilde{T} \union \bigunion_{i=k+1}^s \rho_i(T)$ satisfies $\tilde{q}$.
  In particular,
  Algorithm~\ref{alg:XGCWA-answers/universal/composition}
  outputs ``yes'' on input $T$ and $\tup{t}$.

  The following properties of the
  assignments $\alpha_i$
  are crucial for showing this:

  \begin{sclaim}
    \label{claim:basics}
    Let $i,j \in \set{1,\dotsc,k}$, $x,x' \in X_i$ and $x'' \in X_j$.
    Then,
    \begin{enumerate}[(1)]
    \item\label{claim:basics/proj}
      $v_i(h_i(\alpha_i(\tup{x}_i))) = \beta(\tup{x}_i)$.
      In particular, $v_i(h_i(\alpha_i(x))) = \beta(x)$.
    \item\label{claim:basics/const}
      If $\beta(x) \in \const(T) \union C$,
      then $\alpha_i(x) = \beta(x)$.
    \item\label{claim:basics/eq}
      $\alpha_i(x) = \alpha_i(x')$ if and only if
      $\beta(x) = \beta(x')$.
    \item\label{claim:basics/beta-eq}
      $\alpha_i(x) = \alpha_j(x'')$ implies $\beta(x) = \beta(x'')$.
    \end{enumerate}
  \end{sclaim}

  \begin{proofofclaim}
    \textit{Ad \ref{claim:basics/proj}:}
    Recall that $h_i(A'_i) = A_i \in \tilde{T}_i$,
    and that $v_i$ is injective on $\dom(\tilde{T}_i)$.
    In particular,
    we have $h_i(\alpha_i(\tup{x}_i)) = v_i^{-1}(\beta(\tup{x}_i))$.
    Applying $v_i$ to both sides yields
    $v_i(h_i(\alpha_i(\tup{x}_i))) = \beta(\tup{x}_i)$.

    \ad{\ref{claim:basics/const}}
    Let $\beta(x) \in \const(T) \union C$.
    By \ref{claim:basics/proj}, we have
    \begin{align}
      \label{claim:basics/const/1}
      v_i(h_i(\alpha_i(x))) \,=\, \beta(x),
    \end{align}
    which implies
    \begin{align}
      \label{claim:basics/const/2}
      h_i(\alpha_i(x)) \,=\, \beta(x).
    \end{align}
    Indeed, if $\beta(x) \in \const(T)$,
    \eqref{claim:basics/const/2} follows immediately
    from \eqref{claim:basics/const/1},
    $\const(T) \subseteq \const(\tilde{T}_i)$,
    and the fact that $v_i$ is an injective mapping from $\dom(\tilde{T}_i)$
    that is the identity on constants.
    On the other hand, if $\beta(x) \in C$,
    then \eqref{claim:basics/const/2} follows immediately
    from \eqref{claim:basics/const/1},
    $\beta(x) \in \dom(\hat{T}_i)$,
    and the fact that $v_i^{-1}(c) = c$
    for all $c \in \dom(\hat{T}_i) \intersection C$.

    Now \eqref{claim:basics/const/2} and $h_i(A'_i) \isomorphic A'_i$
    imply that $\alpha_i(x)$ is a constant,
    and since $h_i$ is the identity on constants,
    we have $\alpha_i(x) = \beta(x)$.

    \ad{\ref{claim:basics/eq}}
    By \ref{claim:basics/proj},
    we have $v_i(h_i(\alpha_i(\tup{x}_i))) = \beta(\tup{x}_i)$.
    Recall also that $v_i$ is injective,
    and that $h_i(A_i') = A_i \isomorphic A_i'$,
    which implies that $h_i$ is injective on $\alpha_i(X_i)$.
    Altogether, $f_i \isdef \compose{h_i}{v_i}$
    is a bijection from $\alpha_i(X_i)$ to $\beta(X_i)$.
    This implies that $\alpha_i(x) = \alpha_i(x')$ if and only if
    $\beta(x) = \beta(x')$.

    \ad{\ref{claim:basics/beta-eq}}
    Let $\alpha_i(x) = \alpha_j(x'')$.
    If $i = j$, then $\beta(x) = \beta(x'')$ follows immediately
    from \ref{claim:basics/eq}.
    So assume that $i \neq j$.
    Since $\alpha_i(x) \in \dom(T_i)$, $\alpha_j(x'') \in \dom(T_j)$
    and $\nulls(T_i) \intersection \nulls(T_j) = \emptyset$,
    $\alpha_i(x)$ and $\alpha_j(x'')$ must be constants.
    By \ref{claim:basics/proj} and the fact that the homomorphisms $h_i,h_j$
    as well as the valuations $v_i,v_j$ are the identity on constants,
    we conclude that $\beta(x) = \alpha_i(x) = \alpha_j(x'') = \beta(x'')$.
    \varqed
  \end{proofofclaim}

  We now show that $(T_1,\alpha_1),\dotsc,(T_k,\alpha_k)$ are compatible.
  To this end, we consider the relation
  \[
    \sim\ \, \isdef\,
    \Set{
      (\alpha_i(x),\alpha_j(x'))
      \mid
      i,j \in \set{1,\dotsc,k},\, x \in X_i,\, x' \in X_j,\,
      \beta(x) = \beta(x')
    }
  \]
  on $D \isdef \bigunion_{i=1}^k \alpha_i(X_i)$.

  \begin{sclaim}
    \label{claim:compatible}\hfill
    \begin{enumerate}[(1)]
    \item\label{claim:compatible/sim}
      For all $i,j \in \set{1,\dotsc,k}$,
      $x \in X_i$ and $x' \in X_j$, we have
      \[
        \alpha_i(x) \sim \alpha_j(x') \iff \beta(x) = \beta(x').
      \]
    \item\label{claim:compatible/comp}
      The relation $\sim$ is an equivalence relation on $D$
      that satisfies conditions~\ref{def:compatible/cons}--%
      \ref{def:compatible/eq} of Definition~\ref{def:compatible}.
    \end{enumerate}
  \end{sclaim}

  \begin{proofofclaim}\mbox{}
    \textit{Ad \ref{claim:compatible/sim}:}
    Let $i,j \in \set{1,\dotsc,k}$, $x \in X_i$ and $x' \in X_j$.
    If $\beta(x) = \beta(x')$,
    then the definition of $\sim$ immediately yields
    $\alpha_i(x) \sim \alpha_j(x')$.

    On the other hand, let $\alpha_i(x) \sim \alpha_j(x')$.
    Then there are $i',j' \in \set{1,\dotsc,k}$,
    $y \in X_{i'}$ and $y' \in X_{j'}$
    such that
    \begin{align}
      \label{claim:compatible/sim/alpha}
      \alpha_{i'}(y) = \alpha_i(x)
      \quad \text{and} \quad
      \alpha_{j'}(y') = \alpha_j(x'),
    \end{align}
    and
    \begin{align}
      \label{claim:compatible/sim/beta}
      \beta(y) = \beta(y').
    \end{align}
    By \eqref{claim:compatible/sim/alpha}
    and Claim~\ref{claim:basics}(\ref{claim:basics/beta-eq}),
    we have $\beta(y) = \beta(x)$ and $\beta(y') = \beta(x')$,
    which by \eqref{claim:compatible/sim/beta} yields $\beta(x) = \beta(x')$,
    as desired.

    \ad{\ref{claim:compatible/comp}}
    It is easy to verify that $\sim$ is an equivalence relation on $D$.
    Reflexivity and symmetry are clear,
    and transitivity is easy to show using \ref{claim:compatible/sim}.

    It follows easily from \ref{claim:compatible/sim}
    that $\sim$ satisfies condition~\ref{def:compatible/cons}
    of Definition~\ref{def:compatible}:
    Let $i,j \in \set{1,\dotsc,k}$ and $x \in X_i \intersection X_j$.
    Since $\beta(x) = \beta(x)$,
    \ref{claim:compatible/sim} yields $\alpha_i(x) \sim \alpha_j(x)$.

    For proving that $\sim$ satisfies condition~\ref{def:compatible/const}
    of Definition~\ref{def:compatible},
    let $u,u' \in D$ be such that $u \sim u'$ and $u \in \Const$.
    Since $u \sim u'$,
    there are $i,j \in \set{1,\dotsc,k}$, $x \in X_i$ and $x' \in X_j$
    such that $\alpha_i(x) = u$, $\alpha_j(x') = u'$, and
    \begin{align}
      \label{claim:compatible/comp/beta}
      \beta(x) = \beta(x').
    \end{align}
    By Claim~\ref{claim:basics}(\ref{claim:basics/proj}),
    we have $v_i(h_i(\alpha_i(x))) = \beta(x)$.
    Since $\alpha_i(x)$ is a constant
    and $h_i,v_i$ are the identity on constants,
    this implies that $\alpha_i(x) = \beta(x)$.
    In particular,
    \begin{align}
      \label{claim:compatible/comp/beta2}
      \beta(x')
      \stackrel{\eqref{claim:compatible/comp/beta}}{=}
      \beta(x)
      =
      \alpha_i(x)
      \,\in\,
      \const(T_i)
      \,\subseteq\,
      \const(T) \union C.
    \end{align}
    By Claim~\ref{claim:basics}(\ref{claim:basics/const}),
    this yields $\beta(x') = \alpha_j(x')$,
    and therefore,
    \[
      u
      =
      \alpha_i(x)
      \stackrel{\eqref{claim:compatible/comp/beta2}}{=}
      \beta(x')
      =
      \alpha_j(x')
      = u',
    \]
    as desired.

    Finally, for proving that $\sim$ satisfies condition~\ref{def:compatible/eq}
    of Definition~\ref{def:compatible},
    let $i \in \set{1,\dotsc,k}$ and $x,x' \in X_i$.
    Then,
    \begin{align*}
      \alpha_i(x) = \alpha_i(x')
      \ \stackrel{%
        \text{Claim~\ref{claim:basics}(\ref{claim:basics/eq})}}{\iff}\
      \beta(x) = \beta(x')
      \ \stackrel{%
        \text{Claim~\ref{claim:compatible}(\ref{claim:compatible/sim})}}{\iff}\
      \alpha_i(x) \sim \alpha_i(x'),
    \end{align*}
    as desired.
    \varqed
  \end{proofofclaim}

  By Claim~\ref{claim:compatible},
  $(T_1,\alpha_1),\dotsc,(T_k,\alpha_k)$ are compatible.
  Let $(T_0,\alpha_0)$ be their join.
  We show that
  \[
    T^* \,\isdef\, T_0 \union \bigunion_{i=k+1}^s T_i
  \]
  satisfies $\tilde{q}$,
  where $T_i \isdef \rho_i(T)$ for each $i \in \set{k+1,\dotsc,s}$.
  To this end, we construct an assignment $\alpha$ for $\psi$
  such that $T^* \models \psi(\alpha)$.

  \begin{sclaim}
    \label{claim:h0}
    There is a homomorphism $h_0$ from $T_0$ to
    $\hat{T}_0 \isdef \bigunion_{i=1}^k \hat{T}_i$
    with $h_0(T_0) = \hat{T}_0$,
    and $h_0(\alpha_0(\tup{x}_i)) = \beta(\tup{x}_i)$
    for each $i \in \set{1,\dotsc,k}$.
  \end{sclaim}

  \begin{proofofclaim}
    Let $r_1,\dotsc,r_k$ be the mappings used to construct
    the join $(\tilde{T},\tilde{\alpha})$.
    Then,
    \begin{align}
      \label{claim:h0/T_0}
      T_0 \,=\, \bigunion_{i=1}^k r_i(T_i),
    \end{align}
    and for all $i \in \set{1,\dotsc,k}$ and $x \in X_i$,
    \begin{align}
      \label{claim:h0/alpha_0}
      \alpha_0(x) = r_i(\alpha_i(x)).
    \end{align}
    By Proposition~\ref{prop:XGCWA-answers/universal/composition/join},
    each $r_i$ is injective;
    furthermore, for all $i,j \in \set{1,\dotsc,k}$,
    $u \in \dom(T_i)$ and $u' \in \dom(T_j)$,
    \begin{align}
      \label{claim:h0/r-eq}
      r_i(u) = r_j(u')\ \Longrightarrow\
      \text{$u = u'$, or:
        $u \in \alpha_i(X_i)$, $u' \in \alpha_j(X_j)$ and $u \sim u'$}.
    \end{align}

    Define $h_0\colon \dom(T_0) \to \dom(\hat{T}_0)$
    such that for all $i \in \set{1,\dotsc,k}$
    and $u \in \dom(r_i(T_i))$,
    \begin{align}
      \label{claim:h0/def}
      h_0(u) = v_i(h_i(r_i^{-1}(u))).
    \end{align}
    We claim that $h_0$ is a homomorphism from $T_0$ to $\hat{T}_0$
    with $h_0(T_0) = \hat{T}_0$,
    and that for each $i \in \set{1,\dotsc,k}$
    we have $h_0(\alpha_0(\tup{x}_i)) = \beta(\tup{x}_i)$.

    \step{1}{$h_0$ is well-defined.}
    Let $u \in \dom(r_i(T_i)) \intersection \dom(r_j(T_j))$,
    where $i,j \in \set{1,\dotsc,k}$ are distinct.
    Let $u_i \isdef r_i^{-1}(u) \in \dom(T_i)$
    and $u_j \isdef r_j^{-1}(u) \in \dom(T_j)$.
    We must show that
    \[
      v_i(h_i(u_i)) = v_j(h_j(u_j)).
    \]

    Since $r_i(u_i) = u = r_j(u_j)$, \eqref{claim:h0/r-eq} implies that
    $u_i = u_j$, or:
    $u_i \in \alpha_i(X_i)$, $u_j \in \alpha_j(X_j)$ and $u_i \sim u_j$.
    If $u_i = u_j$,
    then both $u_i$ and $u_j$ are constants,
    since $\nulls(T_i) \intersection \nulls(T_j) = \emptyset$ for $i \neq j$;
    therefore,
    \[
      v_i(h_i(u_i)) = u_i = u_j = v_j(h_j(u_j)),
    \]
    as desired.
    On the other hand, let $x_i \in X_i$ and $x_j \in X_j$
    such that $u_i = \alpha_i(x_i)$, $u_j = \alpha_j(x_j)$ and
    $\alpha_i(x_i) \sim \alpha_j(x_j)$.
    Then Claim~\ref{claim:compatible}(\ref{claim:compatible/sim})
    implies $\beta(x_i) = \beta(x_j)$.
    By Claim~\ref{claim:basics}(\ref{claim:basics/proj}),
    \[
      v_i(h_i(u_i))
      =
      v_i(h_i(\alpha_i(x_i)))
      =
      \beta(x_i)
      =
      \beta(x_j)
      =
      v_j(h_j(\alpha_j(x_j)))
      =
      v_j(h_j(u_j)),
    \]
    as desired.
    Altogether, this shows that $h_0$ is well-defined.

    \step{2}{$h_0$ is a homomorphism from $T_0$ to $\hat{T}_0$
      with $h_0(T_0) = \hat{T}_0$.}
    First note that for each $i \in \set{1,\dotsc,k}$, we have
    \begin{align*}
      h_0(r_i(T_i))
      \stackrel{\eqref{claim:h0/def}}{=}
      v_i(h_i(T_i))
      =
      \hat{T}_i.
    \end{align*}
    Hence,
    \begin{align*}
      h_0(T_0)
      \,\stackrel{\eqref{claim:h0/T_0}}{=}\,
      h_0\left(\bigunion_{i=1}^k r_i(T_i)\right)
      \,=\,
      \bigunion_{i=1}^k h_0(r_i(T_i))
      \,=\,
      \bigunion_{i=1}^k \hat{T}_i
      \,=\,
      \hat{T}_0.
    \end{align*}

    \step{3}{For each $i \in \set{1,\dotsc,k}$,
      we have $h_0(\alpha_0(\tup{x}_i)) = \beta(\tup{x}_i)$.}
    We have
    \begin{align*}
      h_0(\alpha_0(\tup{x}_i))
      \stackrel{\eqref{claim:h0/alpha_0}}{=}
      h_0(r_i(\alpha_i(\tup{x}_i)))
      \stackrel{\eqref{claim:h0/def}}{=}
      v_i(h_i(\alpha_i(\tup{x}_i)))
      \stackrel{\text{Claim~\ref{claim:basics}(\ref{claim:basics/proj})}}{=}
      \beta(\tup{x}_i).
      \varqedeq
    \end{align*}
  \end{proofofclaim}

  Let $h_0$ be a homomorphism as in Claim~\ref{claim:h0}.
  It is easy to extend $h_0$ to a mapping $h$ on $\dom(T^*) \union C$
  with the following properties:
  \begin{enumerate}[(1)]
  \item
    $h(T_0) = h_0(T_0) = \bigunion_{i=1}^k \hat{T}_i$,
  \item
    $h(T_i) = \hat{T}_i$ for each $i \in \set{k+1,\dotsc,s}$, and
  \item
    $h(c) = c$ for each $c \in C$.
  \end{enumerate}
  Note that the second condition can be satisfied,
  since for all distinct $i \in \set{k+1,\dotsc,s}$
  and $j \in \set{1,\dotsc,s}$,
  we have $\nulls(T_i) \intersection \nulls(T_j) = \emptyset$,
  $T_i \isomorphic T$ and $\hat{T}_i \in \rep(T)$.
  Note also that
  \begin{align}
    \label{eq:XGCWA-answers/universal/composition/correctness/h-T*}
    h(T^*) \,=\, \bigunion_{i=1}^s \hat{T}_i.
  \end{align}

  Furthermore, extend $\alpha_0$ to an assignment $\alpha$ for $\tup{y}$
  such that
  \begin{align}
    \label{eq:XGCWA-answers/universal/composition/correctness/h-alpha}
    h(\alpha(y)) = \beta(y)
    \quad
    \text{for each $y \in \tup{y}$}.
  \end{align}
  Note that \eqref{eq:XGCWA-answers/universal/composition/correctness/h-alpha}
  holds for all variables $y$
  that occur in $\tup{x}_i$ for some $i \in \set{1,\dotsc,k}$,
  because $h$ is an extension of $h_0$,
  and $\alpha$ is an extension of $\alpha_0$.
  For each variable $y \in \tup{y}$
  that does not occur in $\tup{x}_i$ for some $i \in \set{1,\dotsc,k}$,
  we pick an arbitrary value $u \in \dom(T^*) \union C$
  with $h(u) = \beta(y)$
  and define $\alpha(y) \isdef u$.
  Note that such a value $u$ always exists.
  First recall that the range of $\beta$ is in
  $\dom(\bigunion_{i=1}^s \hat{T}_i) \union C$.
  If $\beta(y) \in \dom(\bigunion_{i=1}^s \hat{T}_i)$,
  then by \eqref{eq:XGCWA-answers/universal/composition/correctness/h-T*}
  there is some $u \in \dom(T^*)$ with $h(u) = \beta(y)$.
  On the other hand, if $\beta(y) \in C$,
  then $h(\beta(y)) = \beta(y)$, because $h$ is the identity on constants,
  so that we can choose $u = \beta(y)$.

  We are finally ready to show that $T^* \models \psi(\alpha)$.
  First note that
  by Proposition~\ref{prop:XGCWA-answers/universal/composition/join}%
  (\ref{prop:XGCWA-answers/universal/composition/join/tuples}),
  we have $\alpha_0(\tup{x}_i) \in R_i^{T_0}$ for each $i \in \set{1,\dotsc,k}$;
  since $T_0 \subseteq T^*$ and $\alpha$ extends $\alpha_0$, this implies
  \begin{align}
    \label{eq:XGCWA-answers/universal/composition/correctness/pos-rel-atoms}
    \alpha(\tup{x}_i) \in R_i^{T^*}
    \quad
    \text{for each $i \in \set{1,\dotsc,k}$}.
  \end{align}
  Furthermore, we have
  \begin{align}
    \label{eq:XGCWA-answers/universal/composition/correctness/neg-rel-atoms}
    \alpha(\tup{w}_i) \notin Q_i^{T^*}
    \quad
    \text{for each $i \in \set{1,\dotsc,l}$}.
  \end{align}
  Otherwise, if there is some $i \in \set{1,\dotsc,l}$
  with $\alpha(\tup{w}_i) \in Q_i^{T^*}$,
  then by \eqref{eq:XGCWA-answers/universal/composition/correctness/h-T*}
  and \eqref{eq:XGCWA-answers/universal/composition/correctness/h-alpha},
  we have
  \[
    \beta(\tup{w}_i) \in Q_i^{\bigunion_{i=1}^s \hat{T}_i},
  \]
  which is impossible by
  \eqref{eq:XGCWA-answers/universal/composition/correctness/beta2}.
  Finally, we have
  \begin{align}
    \label{eq:XGCWA-answers/universal/composition/correctness/neg-eq}
    \alpha(v_i) \neq \alpha(v_i')
    \quad
    \text{for each $i \in \set{1,\dotsc,m}$}.
  \end{align}
  Indeed, let $i \in \set{1,\dotsc,m}$.
  By \eqref{eq:XGCWA-answers/universal/composition/correctness/h-alpha},
  we have $h(\alpha(v_i)) = \beta(v_i)$ and $h(\alpha(v'_i)) = \beta(v'_i)$.
  On the other hand,
  \eqref{eq:XGCWA-answers/universal/composition/correctness/beta2}
  implies that $\beta(v_i) \neq \beta(v'_i)$,
  so that $\alpha(v_i)$ and $\alpha(v'_i)$ must be distinct.

  Altogether,
  \eqref{eq:XGCWA-answers/universal/composition/correctness/pos-rel-atoms}--%
  \eqref{eq:XGCWA-answers/universal/composition/correctness/neg-eq}
  imply that $T^* \models \psi(\alpha)$.
  In particular,
  Algorithm~\ref{alg:XGCWA-answers/universal/composition}
  outputs ``yes'' on input $T$ and $\tup{t}$.
\end{proof}

\subsubsection{Proof of Proposition~\ref{prop:universal-queries-and-st-tgds}}
\label{sec:XGCWA-answers/universal/coNP}

We conclude this section by proving
Proposition~\ref{prop:universal-queries-and-st-tgds}.
Let $M = \Des$ be a schema mapping, where $\Sigma$ consists of st-tgds,
and let $q$ be a universal query over $\tau$.
As in Section~\ref{sec:XGCWA-answers/universal/proof},
we can assume that $\lnot q$ is logically equivalent to a query $\bar{q}$
of the form
\begin{align*}
  \bar{q}(\tup{x})\ =\ \biglor_{i=1}^m  q_i(\tup{x}),
\end{align*}
where each $q_i$ is an existential query of the form
\begin{align*}
  q_i(\tup{x})\ =\ \exists \tup{y}_i \bigland_{j=1}^{n_i} \phi_{i,j},
\end{align*}
and each $\phi_{i,j}$ is an atomic FO formula
or the negation of an atomic FO formula.

Let $S$ be a source instance for $M$,
and let $\tup{t} \in \Const^{\length{\tup{x}}}$.
As shown in Section~\ref{sec:XGCWA-answers/universal/core},
we have $\tup{t} \notin \certXGCWA{M}{S}{q}$ if and only if
there is a nonempty finite set $\CT$ of minimal instances
in $\rep(\core(M,S))$ with $\bigunion \CT \models \lnot q(\tup{t})$.
Hence, on input $S$ and $\tup{t}$,
a nondeterministic Turing machine can decide whether
$\tup{t} \notin \certXGCWA{M}{S}{q}$
by computing $\core(M,S)$,
and by deciding for each $i \in \set{1,\dotsc,m}$
whether there is a nonempty finite set $\CT$ of minimal instances
in $\rep(\core(M,S))$ with $\bigunion \CT \models q_i(\tup{t})$.
If so, it accepts the input, and otherwise, it rejects it.

By Theorem~\ref{thm:core-algorithm},
$\core(M,S)$ can be computed in time polynomial in the size of $S$
(for fixed $M$).

In order to check whether there is a nonempty finite set $\CT$
of minimal instances in $\rep(\core(M,S))$
with $\bigunion \CT \models q_i(\tup{t})$,
it suffices to ``guess'' a set $\CT$ of at most
\[
  s \isdef n_i \cdot \max\, \set{\arity(R) \mid R \in \tau}
\]
instances in $\rep(\core(M,S))$,
and to check whether $\bigunion \CT \models q_i(\tup{t})$.
Indeed, let $\CT$ be a set of minimal instances
in $\rep(\core(M,S))$ with $\bigunion \CT \models q_i(\tup{t})$.
Then there is an assignment $\alpha$ for the variables in $\tup{x}$
and $\tup{y}_i$ such that $\alpha(\tup{x}) = \tup{t}$ and
$\bigunion \CT \models \phi_{i,j}(\alpha)$ for each $j \in \set{1,\dotsc,n_i}$.
Without loss of generality,
assume that $\phi_{i,1},\dotsc,\phi_{i,k}$ (for $0 \leq k \leq n_i$)
are all the relational atomic FO formulas in $q_i$.
For each $j \in \set{1,\dotsc,k}$,
there is an instance $T_j \in \CT$ with $T_j \models \phi_{i,j}(\alpha)$.
Let $\CT_0' \isdef \set{T_1,\dotsc,T_k} \subseteq \CT$.
Then $\bigunion \CT_0' \models \phi_{i,j}(\alpha)$
for each $j \in \set{1,\dotsc,k}$.
To obtain a set $\CT_0 \subseteq \CT$
that satisfies $\bigunion \CT_0 \models q_i(\tup{t})$,
we extend $\CT_0'$ as follows.
Let $j \in \set{k+1,\dotsc,n_i}$.
Then there are at most $\max\, \set{\arity(R) \mid R \in \tau}$
values that occur in $\phi_{i,j}(\alpha)$.
In particular,
we can pick $\max\, \set{\arity(R) \mid R \in \tau}$ instances from $\CT$
that contain all these values.
Add those instances to $\CT_0'$.
The resulting set $\CT_0$ is a subset of $\CT$,
and satisfies $\bigunion \CT_0 \models q_i(\tup{t})$,
since $\CT_0 \models \phi_{i,j}(\alpha)$ for each $j \in \set{1,\dotsc,k}$.
Furthermore, $\CT_0$ contains at most $s$ instances.

Note also that to find a nonempty finite set $\CT$
of minimal instances in $\rep(\core(M,S))$
with $\card{\CT} \leq s$ and $\bigunion \CT \models q_i(\tup{t})$,
it suffices to consider valuations $v$ of $\core(M,S)$
with range in $C$,
where $C$ contains all constants in $\core(M,S)$,
all constants in $q_i$,
all constants in $\tup{t}$,
and all constants in $\set{c_1,\dotsc,c_{s \cdot k}}$,
where $k$ is the number of nulls in $\core(M,S)$,
and $c_1,\dotsc,c_{s \cdot k}$ is a sequence of pairwise distinct constants
that do not occur in $\core(M,S)$, $q_i$ and $\tup{t}$.

Finally, it is easy for a Turing machine to check whether a given
$T \in \rep(\core(M,S))$ is minimal.
For each atom $A \in T$, it just has to check that the instance
$T \setminus \set{A}$ is not a solution for $S$ under $M$.

Altogether, given a source instance $S$ for $M$,
and a tuple $\tup{t} \in \Const^{\length{\tup{x}}}$,
a nondeterministic Turing machine can check whether
$\tup{t} \notin \certXGCWA{M}{S}{q}$.
This proves $\Eval(M,q) \in \co\NP$,
and in particular, Proposition~\ref{prop:universal-queries-and-st-tgds}.

\section{Conclusion}
\label{sec:conclusion}

A new semantics, called \emph{\XGCWA-semantics},
for answering non-monotonic queries in relational data exchange
has been proposed.
The \XGCWA-semantics is inspired by non-monotonic query answering semantics
from the area of deductive databases,
where the problem of answering non-monotonic queries
has been studied extensively since the late seventies.
In contrast to non-monotonic query answering semantics
proposed earlier in the data exchange literature,
the \XGCWA-semantics can be applied to a broader class of schema mappings
(not just schema mappings defined by tgds and egds),
and possesses the following natural properties:
(1) it is invariant under logically equivalent schema mappings, and
(2) it interprets existential quantifiers ``inclusively''
    as explained in Section~\ref{sec:problems}.
Furthermore, under schema mappings defined by st-tgds and egds
(and even more general schema mappings like schema mappings
 defined by right-monotonic $\LInf$-st-tgds),
the answers to a query under the \XGCWA-semantics
can be defined as the certain answers to the query
with respect to all ground solutions that are unions of minimal solutions.

However, the \XGCWA-semantics is not meant to be a replacement
for earlier semantics proposed in the data exchange literature.
Each of the earlier semantics is interesting in its own right.
In fact, I think that there is no ultimate semantics for answering
non-monotonic queries in relational data exchange.
Depending on the concrete application, and the user's expectations,
one or the other of the proposed semantics may be appropriate.
Nevertheless, query answers under the \XGCWA-semantics
seem to be very natural --
especially due to the two properties mentioned above.

We have shown that the problem of answering non-monotonic queries
under the \XGCWA-semantics can be hard, or even undecidable,
in considerably simple settings.
Unfortunately, this is true not only for the \XGCWA-semantics,
but also for earlier semantics.
This seems to be the price that one has to pay for automatically inferring
``negative data''.
Nevertheless, we were able to show (Theorem~\ref{thm:universal-queries})
that for schema mappings $M$ defined by packed st-tgds,
and for universal queries $q$,
there is a polynomial time algorithm that,
given the core solution for some source instance $S$ for $M$ as input,
outputs the set of answers to $q$ with respect to $M$ and $S$
under the \XGCWA-semantics.

Quite a number of interesting research problems remain open.
First, I believe that the techniques used for proving
Theorem~\ref{thm:universal-queries}
can be extended to prove the analogous result for the more general case
of schema mappings defined by st-tgds.
In fact, it seems that all that has to be done is to provide
a proof of Lemma~\ref{lemma:universal-queries/basis}
for the case that the blocks of $T$ are not packed.
Second, a lot of more work has to be done for understanding the complexity
of answering non-monotonic queries not only under the \XGCWA-semantics,
but also under the semantics proposed earlier.
The fact that for some schema mappings $M$ defined by st-tgds,
and for some existential queries $q$
the data complexity of computing the \XGCWA-answers to $q$ under $M$ is hard
does not imply that it could not be in polynomial time
for other schema mappings defined by st-tgds and other existential queries.
Third, we only considered the data complexity of evaluating queries --
we did not consider the combined complexity,
where the schema mapping and the query to be answered belong to the input.
Finally, instead of answering queries under a non-monotonic semantics,
it could be an interesting task to study the problem
of answering queries using the OWA-semantics,
but allow more expressive constraints
to explicitly exclude ``unwanted'' tuples from solutions
(rather than implicitly by a variant of the CWA).
For instance, instead of using the st-tgd $\theta$ in Example~\ref{exa:copying},
we could have used
$\forall x \forall y\, \bigl(R(x,y) \leftrightarrow R'(x,y)\bigr)$.
Then, under the OWA-semantics, the answer to a query would be as desired.
However, this approach requires schema mappings to be fully specified.

\section*{Acknowledgment}

I am grateful to Nicole Schweikardt
for many helpful discussions on the subject and comments
on the proceedings version of this paper.
Also, I thank the referees of this paper and the referees
of the conference version
for their comments and suggestions regarding the presentation of this paper's
results.

\bibliographystyle{abbrv}
\bibliography{bibliography}

\end{document}

%% file: macros.tex
\newcommand{\isdef}{\mathrel{\mathop:}=}

\renewcommand{\phi}{\varphi}

\newcommand{\CA}{\mathcal{A}}
\newcommand{\CI}{\mathcal{I}}
\newcommand{\CJ}{\mathcal{J}}
\newcommand{\CS}{\mathcal{S}}
\newcommand{\CT}{\mathcal{T}}
\newcommand{\CX}{\mathcal{X}}

\newcommand{\set}[1]{\{{#1}\}}
\newcommand{\Set}[1]{\left\{#1\right\}}
\newcommand{\card}[1]{\lvert{#1}\rvert}
\newcommand{\union}{\cup}
\newcommand{\bigunion}{\bigcup}
\newcommand{\intersection}{\cap}
\newcommand{\bigintersection}{\bigcap}
\newcommand{\compl}[1]{\overline{#1}}

\newcommand{\tup}[1]{\bar{#1}}
\newcommand{\length}[1]{\lvert{#1}\rvert}

\newcommand{\compose}[2]{{#2}\circ{#1}}

\newcommand{\isomorphic}{\cong}
\newcommand{\core}{\operatorname{Core}}

\newlength{\tmplenx}
\newlength{\tmpleny}
\newcommand{\problem}[3]{%
  \begin{center}
    \fbox{
      \parbox[t]{0.95\linewidth}{%
        \noindent%
        \ifthenelse{\equal{#1}{}}{}{{#1}\\[1ex]}%
        \settowidth{\tmplenx}{\textit{Question:}}%
        \setlength{\tmpleny}{\linewidth}%
        \addtolength{\tmpleny}{-1.15\tmplenx}%
        \parbox[t]{\tmplenx}{\textit{Input:}}\hfill\parbox[t]{\tmpleny}{#2}
        \\[1ex]
        \parbox[t]{\tmplenx}{\textit{Question:}}\hfill\parbox[t]{\tmpleny}{#3}
      }%
    }
  \end{center}}

\newcommand{\ProblemName}[1]{\textsc{#1}}
\newcommand{\CLIQUE}{\ProblemName{Clique}}

\newcommand{\complexityclassname}[1]{\textup{#1}}
\newcommand{\PTIME}{\complexityclassname{PTIME}}
\newcommand{\NP}{\complexityclassname{NP}}
\newcommand{\co}[1]{\textup{co-}\!{#1}}

\newcommand{\arity}{\operatorname{ar}}

\newcommand{\limplies}{\rightarrow}
\newcommand{\bigland}{\bigwedge}
\newcommand{\biglor}{\bigvee}

\newcommand{\FO}{\text{FO}}
\newcommand{\LInf}{L_{\infty\omega}}
\newcommand{\CQ}{\text{CQ}}
\newcommand{\CQneg}{\ensuremath{\CQ^{\lnot}}}

\newcommand{\DatalogConstNeq}{\ensuremath{\text{Datalog}^{\textbf{C}(\neq)}}}

\newcommand{\dom}{\operatorname{dom}}
\newcommand{\Dom}{\textit{Dom}}
\newcommand{\const}{\operatorname{const}}
\newcommand{\Const}{\textit{Const}}
\newcommand{\nulls}{\operatorname{nulls}}
\newcommand{\Nulls}{\textit{Null}}

\newcommand{\legal}{\operatorname{\textit{legal}}}
\newcommand{\rep}{\operatorname{\textit{poss}}}
\newcommand{\val}{\operatorname{\textit{val}}}
\newcommand{\minrep}{\operatorname{\textit{min}}}
\newcommand{\minval}{\operatorname{\textit{minval}}}

\newcommand{\DDB}{D}

\newcommand{\bs}{\textit{bs}}

\newcommand{\XGCWA}{\text{GCWA\nobreak\!$^*$}}

\newcommand{\cert}{\operatorname{\textit{cert}}}
\newcommand{\tgd}{\forall \tup{x} \forall \tup{y} (
    \phi(\tup{x},\tup{y}) \limplies \exists \tup{z}\, \psi(\tup{x},\tup{z})
  )}
\newcommand{\egd}{\forall \tup{x} \bigl(\phi(\tup{x}) \limplies x_i=x_j\bigr)}
\newcommand{\Des}{(\sigma,\tau,\Sigma)}
\newcommand{\cansol}{\operatorname{CanSol}}

\newcommand{\fcert}[1]{\cert_\text{#1}}
\newcommand{\fcertOWA}{\fcert{OWA}}

\newcommand{\fcertCWA}{\fcert{CWA}}

\newcommand{\fcertRCWA}{\fcert{RCWA}}
\newcommand{\fcertGCWA}{\fcert{GCWA}}
\newcommand{\fcertEGCWA}{\fcert{EGCWA}}
\newcommand{\fcertPWS}{\fcert{PWS}}
\newcommand{\fcertXGCWA}{\fcert{GCWA$^*$}}

\newcommand{\certOWA}[3]{\fcertOWA({#3},{#1},{#2})}

\newcommand{\certCWA}[3]{\fcertCWA({#3},{#1},{#2})}

\newcommand{\certRCWA}[3]{\fcertRCWA({#3},{#1},{#2})}
\newcommand{\certGCWA}[3]{\fcertGCWA({#3},{#1},{#2})}
\newcommand{\certEGCWA}[3]{\fcertEGCWA({#3},{#1},{#2})}
\newcommand{\certPWS}[3]{\fcertPWS({#3},{#1},{#2})}
\newcommand{\certXGCWA}[3]{\fcertXGCWA({#3},{#1},{#2})}
\newcommand{\Eval}{\ProblemName{Eval}_{\XGCWA}}
\newcommand{\CoreEval}{\ProblemName{CoreEval}}